%% file: main.tex
\documentclass[11pt,a4paper,reqno]{amsart} 
\pdfoutput=1

\usepackage[utf8]{inputenc}
\usepackage{hyperref}
\hypersetup{
  colorlinks   = true,
  urlcolor     = blue,
  linkcolor    = blue,
  citecolor   = blue
}

\usepackage{amsmath}
\usepackage{amsthm}
\usepackage{amsfonts}
\usepackage{amssymb}
\usepackage{array}
\usepackage[margin=2.7cm]{geometry} 
\usepackage{enumitem}

\usepackage{longtable}
\usepackage[table]{xcolor}
\usepackage{colortbl}
\usepackage{tikz}
\usepackage{graphicx}

\include{Spec}

\title[Equivalence and Characterizations of Linear Rank-Metric Codes]{Equivalence and Characterizations of Linear Rank-Metric Codes Based on Invariants}

\author{Alessandro Neri}
\address{Alessandro Neri, Institute for Communications Engineering, Technical University of Munich, Germany}
\email{alessandro.neri@tum.de}
\author{Sven Puchinger}
\address{Sven Puchinger, Department of Applied Mathematics and Computer Science, Technical University of Denmark, Denmark}
\email{svepu@dtu.dk}
\author{Anna-Lena Horlemann-Trautmann}
\address{Anna-Lena Horlemann-Trautmann, Faculty of Mathematics and Statistics, University of St. Gallen, Switzerland}
\email{anna-lena.horlemann@unisg.ch}

\thanks{$^*$Alessandro Neri was supported by the Swiss National Science Foundation through grants no.~169510 and 187711.
Sven Puchinger has received funding from the European Union’s Horizon 2020 research and innovation programme under the Marie Skłodowska-Curie grant agreement no.~713683.
This work was partly done while Sven Puchinger was at Technical University of Munich (TUM), where he was supported by the German Israeli Project Cooperation (DIP) grant no.~KR3517/9-1.}
\thanks{This work is an extension of the conference paper ``Invariants and Inequivalence of
Linear Rank-Metric Codes'' \cite{ne19invariants}, which appeared in \emph{Proceedings of IEEE International Symposium on Information Theory} 2019.}



\begin{document}

\maketitle
\thispagestyle{empty}

\begin{abstract}
 We show that the sequence of dimensions of the linear spaces, generated by a given rank-metric code together with itself under several applications of a field automorphism, is an invariant for the whole equivalence class of the code. The same property is proven for
the sequence of dimensions of the intersections of itself under several applications of a field automorphism.
These invariants give rise to  easily computable criteria to check if two codes are inequivalent. We derive some concrete values and bounds for these dimension sequences for some known families of rank-metric codes, namely Gabidulin and (generalized) twisted Gabidulin codes. We then derive conditions on the length of the codes with respect to the field extension degree, such that codes from different families cannot be equivalent. Furthermore, we derive  upper and lower bounds on the number of equivalence classes of Gabidulin codes and twisted Gabidulin codes, improving a result of Schmidt and Zhou for a wider range of parameters. In the end we use the aforementioned sequences to determine a characterization result for Gabidulin codes.
\end{abstract}

\section{Introduction}

Over the last decades rank-metric codes have become an active research area due to several applications, such as
crisscross error correction~\cite{ro91},
post-quantum cryptography \cite{gabidulin1991ideals,faure2005new,overbeck2008structural,gaborit2013low,loidreau2017new,aguilar2018efficient},
space-time coding for MIMO systems \cite{gabidulin2000space,lusina2002spacetime,liu2002rank,lusina2003,robert2015new,puchinger2016space},
network coding \cite{silva2008rank,etzion2009error,nobrega2010multishot,silva2011universal,etzion2018vector},
distributed data storage~\cite{silberstein2012error,rawat2013optimal,calis2017general,ne18pmds}, and
digital image watermarking~\cite{lefevre2019application}.
They can either be defined as sets of matrices of fixed dimensions over some finite field, where the distance of two elements is measured by the rank of their difference, or equivalently as sets of vectors over an extension field, where the distance is measured as the rank of a vector over the base field. In this work we will use the latter. Furthermore, we will denote an underlying finite field by $\F_{q^m}$, where $q$ is a prime power. 

A rank-metric code in $\F_{q^m}^n$ is called \emph{linear} if it forms an $\F_{q^m}$-linear subspace of $\F_{q^m}^n$. 
Rank-metric codes attaining the Singleton bound, which upper-bounds the code cardinality for given minimum distance, are called \emph{maximum rank distance (MRD) codes}. 
Delsarte \cite{de78}, Gabidulin \cite{ga85a}, and Roth \cite{ro91} independently introduced a prominent class of linear MRD codes for all possible code parameters, which are today called \emph{Gabidulin codes}.

Driven by applications and fundamental questions, finding new (linear) MRD codes inequivalent to Gabidulin codes has become one of the most actively studied research problems within the field of rank-metric codes in the last years.
The topic has been further encouraged by the non-constructive results in \cite{neri2018genericity}, which showed that for many parameter ranges there are plenty of non-Gabidulin linear MRD codes.
After early works on generalizing Gabidulin codes using different automorphisms \cite{roth1996tensor,ks05}, Sheekey \cite{sheekey2016new} was the first to find a general construction for different (linear and non-linear) MRD codes, called \emph{twisted Gabidulin codes}.
Otal and {\"O}zbudak independently discovered a special case of twisted Gabidulin codes \cite{otal2016explicit}.
Starting from Sheekey's construction, several generalizations have been proposed, e.g.\ in \cite[Remark~9]{sheekey2016new} and in \cite{lunardon2018generalized, ot17, sheekey2019new, puchinger2017further}.
Moreover, other non-Gabidulin MRD codes have been constructed, see e.g.\ \cite{ho16, csajbok2018new, csajbok2018maximum, csajbok2018newbis, marino2019mrd, bartoli2019new}. For an overview of non-Gabidulin MRD constructions (also non-linear ones), we refer the reader to the survey \cite{sheekey2019mrd}.

A central question in all of the above mentioned works is whether the new codes are actually inequivalent to known codes. A distinguisher for (generalized) Gabidulin codes, which is based on the dimension of the intersection of the code with itself under some field automorphism, was given in \cite{ho16}. The idea was extended in \cite{giuzzi2019identifiers} to distinguish certain twisted Gabidulin codes. For the other constructions mentioned above, some authors proposed methods tailored to their code construction (e.g., \cite{sheekey2016new}), showed the inequivalence only for special cases (e.g., \cite{puchinger2017further}), or did not study the equivalence problem at all (e.g., \cite{gabidulin2017new}). 
The question of inequivalence is thus still open for many cases. Moreover, not all of the above criteria are efficiently computable. It is thus important to derive an easily computable criterion to check if two codes are equivalent. Since any non-linear code cannot be equivalent to a linear code, we will focus on the question if two linear codes are equivalent.

In this work we widely generalize the results of \cite{ho16,giuzzi2019identifiers} to sums and intersections of the code under arbitrary field automorphisms.
We obtain a class of invariants under equivalence, which can be used as an efficiently computable sufficient tool to prove inequivalence of linear rank-metric codes in general.
The method is particularly powerful for codes constructed as evaluation codes of skew polynomials, and hence suitable for the majority of proposed linear MRD code constructions in the literature.  Furthermore, we show how these invariants  can be used in order to derive theoretical results on MRD codes. In particular, we prove in an elementary way upper and lower bounds on the number of equivalence classes of Gabidulin and twisted Gabidulin codes, and compare them with a result due to Schmidt and Zhou in \cite{sc17}. A second application is a characterization result of Gabidulin codes involving the invariants provided in this work.

The paper is structured as follows. In Section \ref{sec:prelim} we give some preliminaries on finite fields, linearized polynomials, and rank-metric and MRD codes. We define the intersection and sum sequences, and the corresponding dimension sequences in Section \ref{sec:invariants}, and show that these are invariants for the equivalence class of a rank-metric code. Moreover, we derive some general properties of these sequences. In Section \ref{sec:sequences} we compute the sequences for various code families, and use the results to show when two codes from different families are not equivalent. Thereafter, in Section \ref{subsection:numberinequivalent} we use the sequences to derive 
bounds on the number of inequivalent Gabidulin and twisted Gabidulin codes, with an exact formula for %
the case $m=n$. Furthermore, for small code parameters, we present computational results on the number of equivalence classes of generalized twisted Gabidulin codes. In Section \ref{sec:CharGabidulin} we derive new characterization results for Gabidulin codes, based on our sequences. Finally, we conclude this work in Section~\ref{sec:conc} and present some open questions for further research.

\section{Preliminaries}\label{sec:prelim}

\subsection{Finite Fields and Moore Matrices}
\label{ssec:finite_fields_moore}

Let $q$ be a prime power and denote by $\Fq$ the finite field of size $q$.
For a positive integer $m$, the extension field $\Fm$ is a vector space of dimension $m$ over $\Fq$. 
It is also well-known that the extension $\Fm/\Fq$ is Galois, with a cyclic Galois group. More precisely the set
$$\Gal(\Fm/\Fq):= \{ \sigma: \Fm \rightarrow \Fm  \mbox{ field automorphism } \mid \sigma(a)=a \mbox{ for every } a \in \Fq \}$$
is a group endowed with the operation of composition, and is isomorphic to $\Z/m\Z$. The elements of $\Gal(\Fm/\Fq)$ are given by the homomorphisms
$$\begin{array}{rcl}
\theta_i : \Fm & \longrightarrow & \Fm \\
a & \longmapsto & a^{q^i},
\end{array}$$
for every $i= 0,1,\ldots, m-1$, and the generators are given by all the $\theta_i$'s such that $\gcd(i,m)=1$.

Moreover, the \emph{norm} with respect to $\Fm/\Fq$ is the map
$$\begin{array}{rcl}
\Norm: \Fm & \longrightarrow & \Fq \\
a & \longmapsto &\prod\limits_{i=0}^{m-1}\theta_i(a)= \prod\limits_{i=0}^{m-1}a^{q^i}.
\end{array}$$

We now introduce the notions of $\Fq$-support and $q$-rank of a vector over $\Fm$. These notions will play a fundamental role in determining the dimension of certain subspaces of $\Fm^n$.

\begin{definition}\label{def:suppq}
 Let $\Fm$ be an extension field of $\Fq$, and let $g=(g_1,\ldots, g_n)\in \Fm^n$. We define the \emph{$\Fq$-support of $g$ over $\Fq$} the $\Fq$-subspace
$$\rsu_q(g):=\langle g_1,\ldots, g_n\rangle_{\Fq}.$$
Moreover, we define $\rk_q(g):=\dim_{\Fq}(\rsu_q(g))$, and call it the \emph{$q$-rank} (or simply \emph{rank}) \emph{of $g$}.
\end{definition}

The following result is a reformulation of \cite[Corollary 2.38]{lidl1997finite} and is a consequence of \cite[Corollary 4.13]{lam1988vandermonde}. We will widely use it in the rest of the paper.

\begin{proposition}\label{cor:rankMoore}
Let $\theta$ be a generator of $\Gal(\Fm/\Fq)$,  $g \in \Fm^n$ and $k$ be a positive integer. Then 
$$\dim_{\Fm}\langle g, \theta(g), \ldots, \theta^{k-1}(g) \rangle_{\Fm}=\min\{k,\rk_q(g)\}.$$
In particular, the set $\{g, \theta(g), \dots, \theta^{\rk_q(g)-1}(g)\}$ is linearly independent over $\Fm$.
\end{proposition}

The above result is usually stated in terms of the rank of the \emph{Moore matrix}, which is the $k\times n$ matrix defined as

\[M_{k,\tau}(v) := \left( \begin{array}{cccc} v_1 & v_2
    &\dots &v_n \\ \tau(v_1) & \tau(v_2) &\dots &\tau(v_n) \\
    \vdots&&&\vdots \\ \tau^{k-1}(v_1)& \tau^{k-1}(v_2) &\dots
    &\tau^{k-1}(v_n) \end{array}\right),\] 
where $v=(v_1,\ldots,v_n) \in \Fm^n$ and $\tau \in \Aut(\Fm)$. Note that
the Moore matrix is the $q$-analogue of the Vandermonde matrix.

\subsection{The Skew Group Algebra $\Fm[\Gal(\Fm/\Fq)]$}
\label{ssec:linearized_polynomials}

In the context of rank-metric codes, an important related object is the ring of \emph{linearized polynomials}, which was first studied by {\O}re in~\cite{ore_special_1933}. Its elements are polynomials in $\Fm[x]$ that involve only monomials of the form $x^{q^i}$, for some non-negative integers $i$. Their importance is due to the fact that, seen as functions corresponding to their evaluation, they are $\Fq$-linear maps from $\Fm$ to itself. On the other hand, any $\Fq$-linear map from $\Fm$ to itself can be represented as a $q$-polynomial of degree at most $q^{m-1}$. Let $\cL(\Fm)$ denote the set of $q$-polynomials with coefficients in $\Fm$. This set is closed under addition and composition, and together with these two operations, $\cL(\Fm)$ is a non-commutative ring. 

When one only cares about the evaluation in $\Fm$, one can reduce to studying the set 
$\cL_m(\Fm):=\cL(\Fm)/(x^{q^m}-x),$ 
since $a^{q^m}=a$ for every $a\in \Fm$, and  $(x^{q^m}-x)$ is a two-sided ideal. In this framework, one can easily verify that 
$$\cL_m(\Fm)\cong \Fq^{m\times m}.$$

However, this is not the end of the story: The ring $\cL_m(\Fm)$ is also isomorphic to the skew group algebra $\Fm[G]$, where $G=\Gal(\Fm/\Fq)=\langle \theta \rangle$, which is a ring endowed with addition and composition. More in detail, the elements $f,g \in \Fm[G]$ are of the form $f=\sum_{i=0}^{m-1}f_i \theta^i$, $g=\sum_{i=0}^{m-1}g_i \theta^i$, for some $f_i, g_i \in \Fm$. The addition is defined by
$f+g =\sum_{i=0}^{m-1} (f_i+g_i)\theta^i$; the composition is defined on monomials by $(f_i\theta^i)\circ (g_j\theta^j)=f_i\theta^i(g_j)\theta^{i+j}$, and then extended by associativity and distributivity. In this framework, we also have that 
$$\Fm[G] \cong \End_{\Fq}(\Fm)=\{ \phi:\Fm \rightarrow \Fm \mid \phi \mbox{ is } \Fq\mbox{-linear} \}.$$

The importance of this point of view is that it can be generalized to fields of any characteristic, provided that the field extension has a cyclic Galois group.
This was the key point of the works by Roth~\cite[Section~VI]{roth1996tensor} and Augot, Loidreau and Robert \cite{augot2013rank,augot2014generalization,augot2018generalized}.
For a deeper understanding on this topic over finite fields, the interested reader is referred to \cite{wu2013linearized}. A brief summary of this for general fields can be also found in \cite[Chapter 4]{neri2019PhD}.  This explains why our notation will follow the skew group algebra setting.

\subsection{Rank-Metric and MRD Codes}\label{ssec:rank-metric_codes}

We now explain the basics of linear rank-metric codes and their equivalence maps, and define maximum rank distance (MRD) codes.

\begin{definition}
The \emph{rank distance} between $u,v \in \Fm^n$ is defined as $$\drk(u,v):=\rk_q(u-v).$$ A \emph{linear (vector) rank-metric code} is an $\Fm$-linear subspace $C \subseteq \Fm^n$. If $C \neq \{0\}$ is a linear rank-metric code, then the \emph{minimum distance} of $C$ is the integer 
$$d(C):= \min\{\drk(u,v) \mid u,v \in C, \ u \neq v\}= \min\{\rk_q(u) \mid u \in C, \ u \neq 0\} .$$
\end{definition}

It is easy to verify  that the map $\drk: \Fm^n \times \Fm^n \rightarrow \mathbb N$ defines a metric on $\Fm^n$. From now on we will refer to a linear (vector) rank-metric code $C \subseteq \Fm^n$ of dimension $k$  as an $\Fmk$ code. When the minimum distance  $d=d(C)$ is known, we will call it an $\Fmkd$ code.

Let $V,W$ be vector spaces over a field $\Fqm$. Recall that a map $\varphi :V \longrightarrow W$ is called \emph{semilinear}, if there exists $\tau \in \Aut(\Fqm)$ such that, for all $x, y \in V$ and $\lambda \in \Fqm$,  it holds that
\begin{enumerate}
\item $f(x+y)=f(x)+f(y)$.
\item $f(\lambda x)=\tau(\lambda)f(x)$.
\end{enumerate}
If $V=W$, then the set of invertible semilinear maps is a group, called \emph{general semilinear group} and denoted by $\Gamma\mathrm{L}(V)$. Furthermore, $\Gamma\mathrm{L}(V) \cong \GL(V)  \rtimes \Aut(\F)$.

\begin{definition}
Two  rank-metric codes $C,C' \subseteq \Fm^n$ are \emph{(semilinearly) equivalent} if there exists an $\Fm$-semilinear isometry (i.e., distance-preserving mapping) $\varphi: (\Fm^n,\drk) \to (\Fm^n,\drk)$  such that $\varphi(C)=C'$.  If $C, C'\in \Fm^n$ are equivalent 
rank-metric codes, then we will write $C\sim C'$.
\end{definition}

The semilinear rank isometries on $\Fm^n$ are induced by the semilinear
isometries on $\mat$ (see
\cite{be03,mo14,wan1996geometry}) and  are characterized as follows.
\begin{theorem}\cite[Corollary 1]{be03}\cite[Proposition~2]{mo14}\label{isometries}
  The semilinear $\Fq$-rank isometries on $\Fm^n$ are of the
  form
  \[(\lambda, A, \tau) \in \left( \Fm^* \times \GL_n(q) \right)
  \rtimes \Aut(\Fm) ,\] acting on $ \Fm^n$  via
  \[(v_1,\dots,v_n) (\lambda, A, \tau) = (\tau(\lambda
  v_1),\dots,\tau(\lambda v_n)) A .\] In particular, if
  ${C}\subseteq \Fm^n$ is a rank-metric  code with minimum rank
  distance $d$, then
 ${C}' = \tau(\lambda {C}) A $
  is a rank-metric  code with minimum rank distance $d$.
\end{theorem}

Observe that we can always reduce to the case $\lambda=1$, because if  $C$ and $C'$ are $\Fmk$ codes and ${C}' = \tau(\lambda {C}) A $, then by $\Fm$-linearity, we also have $\tau(\lambda {C}) A=\tau(\lambda)\tau(C) A=\tau(C)A$.
Hence, $\lambda$ is only relevant when the considered codes are not linear over $\Fqm$.

Recall that the \emph{standard inner-product} (or \emph{dot product}) of $u,v \in \Fm^n$ is $\langle u ; v \rangle:=\sum_{i=1}^n u_iv_i$. It is well-known that the map
$(u,v) \mapsto \langle u; v \rangle$ defines an $\Fm$-bilinear, symmetric and nondegenerate form on $\Fm^n$.

\begin{definition}\label{def:vectordual}
The \emph{dual} of a $\Fmk$ (vector) rank-metric code $C$  is $$C^\perp:=\{u \in \Fm^n \mid \langle u; v \rangle =0 \mbox{ for all } v \in C\}.$$
Note that $C^\perp$ is an $[n,n-k]_{q^m}$ code.
\end{definition}

The following result is the rank-metric analogue of the Singleton bound for codes with the Hamming metric.

\begin{theorem}\cite[Theorem 5.4]{de78}\label{singbound}
Let $d,k,n,m$ be  positive integers such that $0<k\leq n\leq m$, and $C$ be an $\Fmkd$ code. Then 
$$d \leq n-k+1.$$
\end{theorem}

\begin{definition}
A non-zero code $C$ is a \emph{maximum rank distance (MRD) code} if it meets the bound of Theorem \ref{singbound}.
\end{definition}

It was shown in \cite{de78, ga85a} that MRD codes exist for any parameter set with $n \leq m$. Moreover, the condition $n\leq m$ is also necessary for linear codes. Therefore we assume $n\leq m$ throughout the paper.

\subsection{Known MRD Constructions}

The first construction of MRD codes, generally known as Gabidulin codes, was found independently by Delsarte \cite{de78}, Gabidulin \cite{ga85a}, and Roth \cite{ro91}. It was then generalized in~\cite{roth1996tensor,ks05}. 

\begin{definition}\cite{de78,ga85a,ro91,roth1996tensor,ks05}
Let $k,n,m$ be positive integers such that $1\leq k \leq n \leq m$ and let $\theta$ be a generator of $G=\Gal(\Fm/\Fq)$. We denote by $\mathcal G_{k,\theta}$  the $\Fm$-subspace of the skew group algebra  $\Fm[\theta]=\Fm[G]$ generated by the first $k$ powers of $\theta$, that is
$$ \G_{k,\theta}:= \left\{f_0\mathrm{id}+f_1\theta+\ldots+f_{k-1}\theta^{k-1} \mid f_i\in \Fm\right\}.$$
Let $g=(g_1,\ldots,g_n)\in \Fm^n$ such that $\rk_q(g)=n$. 
The \emph{$\theta$-Gabidulin code} $\G_{k,\theta}(g)$ is defined as
$$\G_{k,\theta}(g):=\left\{(f(g_1),\ldots,f(g_n)) \mid f \in \mathcal G_{k,\theta} \right\}.$$
\end{definition}

\begin{proposition}\cite{ga85a,ks05}\label{prop:GabisMRD}
The  $\theta$-Gabidulin code $\G_{k,\theta}(g)$ has cardinality $q^{km}$ and minimum distance $d=n-k+1$, i.e., $\G_{k,\theta}(g)$ is an MRD code.
\end{proposition}

The following result gives an explicit expression for the dual of a $\theta$-Gabidulin code, which is in turn a $\theta$-Gabidulin code.

\begin{proposition}\cite[Sections~2 and 4]{ga85a}\cite[Subsection IV.C]{ks05}\label{prop:dualGab}
 Let  $C=\G_{k,\theta}(g)$ be a $\theta$-Gabidulin code. Then 
 $$C^\perp=\G_{n-k, \theta}(g') ,$$ 
 where $g'$ is any non-zero vector in the code $\G_{n-1,\theta}(\theta^{-(n-k-1)}(g))^\perp$. Moreover, $\rk_q(g')=n$.
\end{proposition}

Gabidulin codes are not the only known MRD codes. There are some other families of codes which attain the Singleton-like bound of Theorem \ref{singbound}, that have been discovered in the last years. Here we give an overview of some of these families.

\begin{definition}\cite{sheekey2016new}\label{def:TGabcode}
Let $k,n,m,b$ be positive integers such that $1\leq k \leq n \leq m$ and $0 \leq b <m$. Let  $\theta$ be a generator of $G=\Gal(\Fm/\Fq)$  and $\eta \in \Fm$. We denote by $\HShee{k,\theta}{\eta,b}$  the $\Fq$-subspace of the skew group algebra  $\Fm[\theta]=\Fm[G]$ given by
$$ \HShee{k,\theta}{\eta,b}:= \left\{f_0\mathrm{id}+f_1\theta+\ldots+f_{k-1}\theta^{k-1}+\eta\theta^b(f_0)\theta^k \mid f_i\in \Fm\right\}.$$
Moreover, let $g=(g_1,\ldots,g_n)\in \Fm^n$ such that $\rk_q(g)=n$ and suppose that $\Norm(\eta)\neq(-1)^{km}$. The \emph{$\theta$-twisted Gabidulin code} with parameters $\eta$ and $b$  is defined as
$$\TGabShee{k,\theta}{\eta,b}{g}:=\left\{(f(g_1),\ldots,f(g_n)) \mid f \in \HShee{k,\theta}{\eta,b} \right\}.$$
\end{definition}

\begin{proposition}\cite{sheekey2016new}\label{prop:TGabisMRD}
Let $g \in \Fm^n$ such that $\rk_q(g)=n$ and $\eta\in \Fm$ such that $\Norm(\eta)\neq(-1)^{km}$. The  $\theta$-twisted Gabidulin code $\TGabShee{k,\theta}{\eta,b}{g}$ has cardinality $q^{km}$ and minimum distance $d=n-k+1$, i.e., $\TGabShee{k,\theta}{\eta,b}{g}$ is an MRD code.
\end{proposition}

\begin{remark}
Observe that in general a $\theta$-twisted Gabidulin code is not $\Fm$-linear, but only $\Fq$-linear. It is $\Fm$-linear if and only if $b=0$, in which case, we will denote the set $\HShee{k,\theta}{\eta,0}$ by $\HShee{k,\theta}{\eta}$, and the corresponding code by $\TGabShee{k,\theta}{\eta}{g}$. Therefore, by Proposition \ref{prop:TGabisMRD}, the code $\TGabShee{k,\theta}{\eta}{g}= \left\langle g+\eta \theta^k(g), \theta(g), \dots, \theta^{k-1}(g) \right\rangle_{\Fqm}$ is an $[n,k,n-k+1]_{q^m}$ code.
\end{remark}

 This family of codes was given by Sheekey in \cite{sheekey2016new}, and  was first introduced only considering the $q$-Frobenius automorphism $\theta(x)=x^{q}$. 
In \cite[Remark~9]{sheekey2016new} and \cite{lunardon2018generalized}, it was generalized to any generator of $\Gal(\Fm/\Fq)$. Further generalizations were given   in \cite{ot17}, where the codes obtained are only linear over the prime subfield of $\Fq$.

 From now on we fix the following notation. Let $1\leq k \leq n \leq m$ be integers. Choose a positive integer $\numTwists \in \N$, which we call the \emph{number of twists}. Let $\hVec \in \{0,\dots,k-1\}^\numTwists$ and $\tVec \in \{1,\dots,n-k\}^\numTwists\cup\{m-n+1,\ldots,m-k\}^\numTwists$ such that the $h_i$'s are distinct and the $t_i$'s are distinct. Furthermore, let $\etaVec \in (\Fqm)^\numTwists$
and $\theta$ be a generator of $G=\Gal(\Fm/\Fq)$. We can now define the generalized twisted Gabidulin codes from \cite{puchinger2017further}:

\begin{definition}[\cite{puchinger2017further}]\label{def:genTwistedCodes}
We denote by $\PTw{k,\theta}{\eta,\tVec,\hVec}$  the $\Fm$-subspace of the skew group algebra  $\Fm[\theta]=\Fm[G]$ given by
$$ \PTw{k,\theta}{\etaVec,\tVec,\hVec}:= \left\{f_0\mathrm{id}+f_1\theta+\ldots+f_{k-1}\theta^{k-1}+\sum_{j=1}^{\numTwists} \eta_j f_{h_j}\theta^{k-1+t_j}\mid f_i\in \Fm\right\}.$$

Moreover, let $g \in \Fm^n$ with $\rk_q(g)=n$. The \emph{generalized $\theta$-twisted Gabidulin code} $\GTw{k,\theta}{\etaVec,\tVec,\hVec}{g}$ is defined as
$$ \GTw{k,\theta}{\etaVec,\tVec,\hVec}{g}:= \left\{(f(g_1),\ldots,f(g_n)) \mid f \in \PTw{k,\theta}{\etaVec,\tVec,\hVec} \right\}.  $$
\end{definition}

Note that generalized $\theta$-twisted Gabidulin codes are  $\Fm$-linear by definition. In particular, the code $\GTw{k,\theta}{\etaVec,\tVec,\hVec}{g}$ can be written as
$$ \left\langle \left\{\theta^{h_i}(g)+\eta_i \theta^{k-1+t_i}(g) \mid i\in[\numTwists] \right\} \\
\cup \left\{\theta^{i}(g) \mid i \in \{0,\dots,k-1\} \setminus \{h_1,\dots,h_\ell\} \right\} \right\rangle_{\Fm}.$$

In general, there is a sufficient MRD condition if the $g_i$'s are chosen from a subfield $\Fqr \subseteq \Fm$ with $r 2^\numTwists \mid m$ and a suitable choice of the $\eta_i$ \cite{puchinger2017further} (see also \cite[Chapter~7]{puchinger2018construction} for more details). Note that this gives codes of length $n \leq 2^{-\numTwists}m$. It is an open problem whether longer MRD codes exist in $\PTw{k,\theta}{\eta,\tVec,\hVec}$ for arbitrary $\tVec$ and $\hVec$.
In the special case $\numTwists=1$, we write $\tVec = t_1 \in \N$ and $ \hVec = h_1 \in \N_0$.

Note that these codes have been originally defined only for $\tVec \in \{1,\dots,n-k\}^\numTwists$.  This was done in order to assure that the codes have dimension equal to $k$. Here we relax this condition, since we can still guarantee that the dimension of  $\GTw{k,\theta}{\etaVec,\tVec,\hVec}{g}$ is equal to $k$ when the $t_i$'s belong to $\{m-n+1,\ldots,m-k\}$, by Proposition \ref{cor:rankMoore}, using $\theta^{m-n+k}(g)$ instead of $g$.

In the following we describe two further constructions due to Gabidulin in \cite{gabidulin2017new}.

\begin{definition}\cite{gabidulin2017new}\label{def:GabNewIcodes}
Let $k,n,m$ be positive integers such that $1\leq k \leq n \leq m$ and $m-k>k$. Let  $\theta$ be a generator of $G=\Gal(\Fm/\Fq)$  and $\eta \in \Fm$.  We denote by $\NewG{k,\theta}{\eta, I}$
the $\Fm$-subspace of the skew group algebra  $\Fm[\theta]=\Fm[G]$ given by
$$\NewG{k,\theta}{\eta,I}:= \left\langle \left\{ \theta^{i}+\theta^{i}(\eta) \theta^{k+i}\mid i \in\{0,\ldots, k-1\} \right\} \right\rangle_{\Fm}.$$
Moreover, let $g \in \Fm^n$ with $\rk_q(g)=n$. The \emph{new $\theta$-Gabidulin code of first kind} $\NewG{k,\theta}{\eta,I}(g)$ is defined as
$$ \NewG{k,\theta}{\eta,I}(g):= \left\{(f(g_1),\ldots,f(g_n)) \mid f \in \NewG{k,\theta}{\eta, I} \right\}.  $$
\end{definition}

Observe that the new $\theta$-Gabidulin codes of first kind can be seen as a special case of generalized $\theta$-twisted Gabidulin codes in the sense of Definition \ref{def:genTwistedCodes}, with
\begin{equation*}
\ell = k, \quad h_i = i-1, \quad t_i = i, \quad \mbox{and} \quad \eta_i = \theta^{i-1}(\eta)
\end{equation*}
for $i \in [k]$.

\begin{definition}\cite{gabidulin2017new}\label{def:GabNewIIcodes}
Let $k,n,m$ be positive integers such that $1\leq k \leq n \leq m$ and $m-k \leq k$. Let $\theta$ be a generator of $G=\Gal(\Fm/\Fq)$ and $\eta \in \Fm$. We denote by $\NewG{k,\theta}{\eta, II}$
the $\Fm$-subspace of the skew group algebra  $\Fm[\theta]=\Fm[G]$ given by
$$\NewG{k,\theta}{\eta,II}:= \left\langle \left\{ \theta^{i}+\theta^{i}(\eta) \theta^{k+i}\mid i \in\{0,\ldots, m-k-1\} \right\} 
\cup \left\{ \theta^i\mid m-k \leq i <k \right\}  \right\rangle_{\Fm}.$$
Moreover, let $g \in \Fm^n$ with $\rk_q(g)=n$. The \emph{new $\theta$-Gabidulin code of second kind} $\NewG{k,\theta}{\eta,II}(g)$ is defined as
$$ \NewG{k,\theta}{\eta,II}(g):= \left\{(f(g_1),\ldots,f(g_n)) \mid f \in \NewG{k,\theta}{\eta, II} \right\}.  $$
\end{definition}

Also the new $\theta$-Gabidulin codes of second kind can be seen as a special case of generalized $\theta$-twisted Gabidulin codes in a loosened variant of Definition \ref{def:genTwistedCodes}, where we allow the twist vectors to be from $t \in \{1,\dots,m-k\}^\ell$, with
\begin{equation*}
\ell = m-k, \quad h_i = i-1, \quad t_i = i, \quad \mbox{and} \quad \eta_i = \theta^{i-1}(\eta)
\end{equation*}
for $i \in[m-k]$.

\begin{proposition}\cite{gabidulin2017new}\label{prop:newGabsareMRD}
Let $1\leq k \leq m$ be integers, $\theta$ be a generator of $\Gal(\Fm/\Fq)$, $g \in \Fm^n$ with $\rk_q(g)=n$. Suppose, moreover, that $\Norm(\eta) \neq (-1)^{km}$. 
\begin{enumerate}
\item If $m-k > k$, then the new $\theta$-Gabidulin code of first kind $\NewG{k,\theta}{\eta,I}(g)$ is an $\Fmk$ MRD code.
\item If $m-k \leq k$, then the new $\theta$-Gabidulin code of second kind $\NewG{k,\theta}{\eta,II}(g)$ is an $\Fmk$ MRD code.
\end{enumerate}
\end{proposition}

An overview of the code constructions discussed in this subsection can be found in Figure~\ref{fig:overview_code_constructions}.

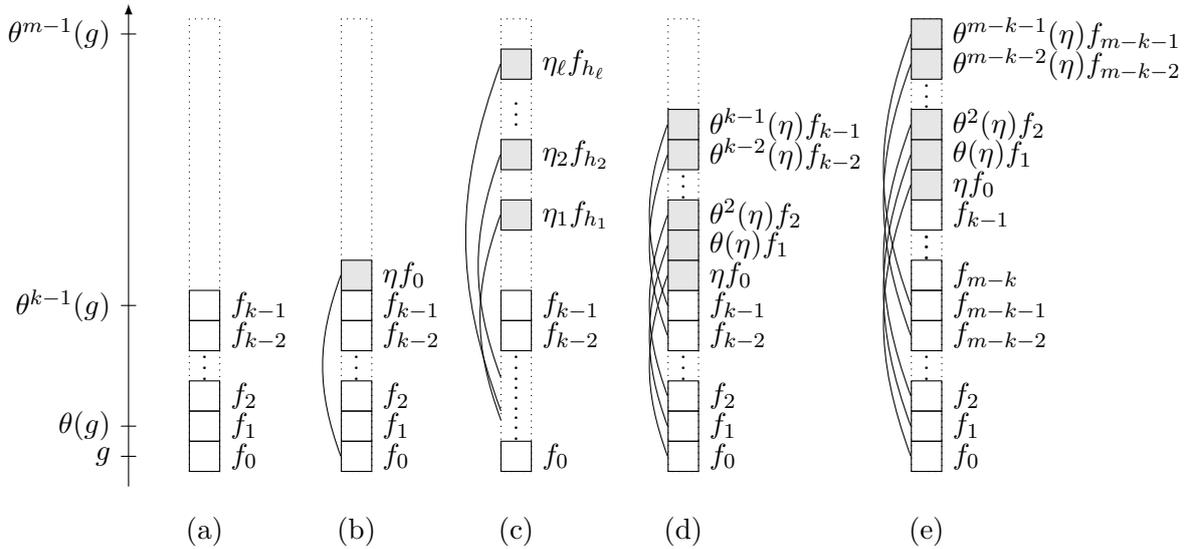
\begin{figure}[ht!]
	\begin{center}
		\begin{tikzpicture}
		\def\ylabelpos{-1cm}
		\def\xlevelminusone{-1cm}
		\def\xlevelzero{0cm}
		\def\xlevelone{2cm}
		\def\xleveltwo{4.1cm}
		\def\xlevelthree{6.3cm}
		\def\xlevelfour{9.5cm}
		\def\xwidth{0.4cm}
		\def\ywidth{0.4cm}
		\def\ydist{1cm}
		\def\nval{15}
		\def\kval{5}
		\def\twistval{6}
		\def\twistvaltwo{8}
		\def\twistvalthree{12}
		\def\myeps{0.1cm}
		\tikzstyle{coeffNodesPure}=[draw, rectangle, color=blue, minimum width=\xwidth, minimum height=\ywidth, inner sep=0pt]
		\tikzstyle{coeffNodes}=[draw, rectangle, minimum width=\xwidth, minimum height=\ywidth, inner sep=0pt]
		\tikzstyle{twistNodes}=[draw, rectangle, fill=black!10, minimum width=\xwidth, minimum height=\ywidth, inner sep=0pt]
		\tikzstyle{coeffNodesLevels}=[draw, rectangle, color=blue, minimum width=\xwidth, minimum height=\ywidth, inner sep=0pt]
		
		\draw[->,>=latex] (\xlevelminusone,-\ywidth) to (\xlevelminusone,\nval*\ywidth);
		\draw (\xlevelminusone+\myeps,0) -- (\xlevelminusone-\myeps,0) node[left] {$g$};
		\draw (\xlevelminusone+\myeps,\ywidth) -- (\xlevelminusone-\myeps,\ywidth) node[left] {$\theta(g)$};
		\draw (\xlevelminusone+\myeps,5*\ywidth) -- (\xlevelminusone-\myeps,5*\ywidth) node[left] {$\theta^{k-1}(g)$};
		\draw (\xlevelminusone+\myeps,\nval*\ywidth-\ywidth) -- (\xlevelminusone-\myeps,\nval*\ywidth-\ywidth) node[left] {$\theta^{m-1}(g)$};

		\draw[dotted] (\xlevelzero-0.5*\xwidth,-0.5*\ywidth) rectangle (\xlevelzero+0.5*\xwidth,\nval*\ywidth-0.5*\ywidth);
		\node[coeffNodes,label=right:{$f_0$}]  (levelZerocoeff0) at (\xlevelzero+0,0*\ywidth) {};
		\node[coeffNodes,label=right:{$f_1$}]  (levelZerocoeff1) at (\xlevelzero+0,1*\ywidth) {};
		\node[coeffNodes,label=right:{$f_2$}]  (levelZerocoeff2) at (\xlevelzero+0,2*\ywidth) {};
		\node 								   (levelZerocoeffdots) at (\xlevelzero+0,3.2*\ywidth) {\footnotesize $\vdots$};
		\node[coeffNodes,label=right:{$f_{k-2}$}]  (levelZerocoeffk2) at (\xlevelzero+0,4*\ywidth) {};
		\node[coeffNodes,label=right:{$f_{k-1}$}]  (levelZerocoeffk1) at (\xlevelzero+0,5*\ywidth) {};

		\draw[dotted] (\xlevelone-0.5*\xwidth,-0.5*\ywidth) rectangle (\xlevelone+0.5*\xwidth,\nval*\ywidth-0.5*\ywidth);
		\node[coeffNodes,label=right:{$f_0$}]  (levelOnecoeff0) at (\xlevelone+0,0*\ywidth) {};
		\node[coeffNodes,label=right:{$f_1$}]  (levelOnecoeff1) at (\xlevelone+0,1*\ywidth) {};
		\node[coeffNodes,label=right:{$f_2$}]  (levelOnecoeff2) at (\xlevelone+0,2*\ywidth) {};
		\node 								   (levelOnecoeffdots) at (\xlevelone+0,3.2*\ywidth) {\footnotesize $\vdots$};
		\node[coeffNodes,label=right:{$f_{k-2}$}]  (levelOnecoeffk2) at (\xlevelone+0,4*\ywidth) {};
		\node[coeffNodes,label=right:{$f_{k-1}$}]  (levelOnecoeffk1) at (\xlevelone+0,5*\ywidth) {};
		\node[twistNodes,label=right:{$\eta f_0$}] (levelOnetwist1) at (\xlevelone+0,\twistval*\ywidth) {};
		\draw[bend angle=20, bend left] (levelOnecoeff0.west) to (levelOnetwist1.west);

		\draw[dotted] (\xleveltwo-0.5*\xwidth,-0.5*\ywidth) rectangle (\xleveltwo+0.5*\xwidth,\nval*\ywidth-0.5*\ywidth);
		\node[coeffNodes,label=right:{$f_0$}]  (levelThreecoeff0) at (\xleveltwo+0,0*\ywidth) {};
		\node 								   (levelThreecoeff1) at (\xleveltwo+0,1.2*\ywidth) {\footnotesize $\vdots$};
		\node 								   (levelThreecoeff2) at (\xleveltwo+0,2.2*\ywidth) {\footnotesize $\vdots$};
		\node 								   (levelThreecoeffdots) at (\xleveltwo+0,3.2*\ywidth) {\footnotesize $\vdots$};
		\node[coeffNodes,label=right:{$f_{k-2}$}]  (levelThreecoeffk2) at (\xleveltwo+0,4*\ywidth) {};
		\node[coeffNodes,label=right:{$f_{k-1}$}]  (levelThreecoeffk1) at (\xleveltwo+0,5*\ywidth) {};
		\node[twistNodes,label=right:{$\eta_1 f_{h_1}$}] (levelThreetwist1) at (\xleveltwo+0,8*\ywidth) {};
		\node[twistNodes,label=right:{$\eta_2 f_{h_2}$}] (levelThreetwist2) at (\xleveltwo+0,10*\ywidth) {};
		\node (levelThreetwistdots) at (\xleveltwo+0,11.6*\ywidth) {\footnotesize $\vdots$};
		\node[twistNodes,label=right:{$\eta_\ell f_{h_\ell}$}] (levelThreetwist3) at (\xleveltwo+0,13*\ywidth) {};
		\draw[bend angle=20, bend left] (\xleveltwo-0.5*\xwidth,1.2*\ywidth) to (levelThreetwist1.west);
		\draw[bend angle=20, bend left] (\xleveltwo-0.5*\xwidth,2.6*\ywidth) to (levelThreetwist2.west);
		\draw[bend angle=20, bend left] (\xleveltwo-0.5*\xwidth,1.5*\ywidth) to (levelThreetwist3.west);

		\draw[dotted] (\xlevelthree-0.5*\xwidth,-0.5*\ywidth) rectangle (\xlevelthree+0.5*\xwidth,\nval*\ywidth-0.5*\ywidth);
		\node[coeffNodes,label=right:{$f_0$}]  (levelFourcoeff0) at (\xlevelthree+0,0*\ywidth) {};
		\node[coeffNodes,label=right:{$f_1$}]  (levelFourcoeff1) at (\xlevelthree+0,1*\ywidth) {};
		\node[coeffNodes,label=right:{$f_2$}]  (levelFourcoeff2) at (\xlevelthree+0,2*\ywidth) {};
		\node 								   (levelFourcoeffdots) at (\xlevelthree+0,3.2*\ywidth) {\footnotesize $\vdots$};
		\node[coeffNodes,label=right:{$f_{k-2}$}]  (levelFourcoeffk2) at (\xlevelthree+0,4*\ywidth) {};
		\node[coeffNodes,label=right:{$f_{k-1}$}]  (levelFourcoeffk1) at (\xlevelthree+0,5*\ywidth) {};
		\node[twistNodes,label=right:{$\eta f_{0}$}] (levelFourtwist0) at (\xlevelthree+0,6*\ywidth) {};
		\node[twistNodes,label=right:{$\theta(\eta) f_{1}$}] (levelFourtwist1) at (\xlevelthree+0,7*\ywidth) {};
		\node[twistNodes,label=right:{$\theta^2(\eta) f_{2}$}] (levelFourtwist2) at (\xlevelthree+0,8*\ywidth) {};
		\node (levelFourtwistdots) at (\xlevelthree+0,9.2*\ywidth) {\footnotesize $\vdots$};
		\node[twistNodes,label=right:{$\theta^{k-2}(\eta) f_{k-2}$}] (levelFourtwistk2) at (\xlevelthree+0,10*\ywidth) {};
		\node[twistNodes,label=right:{$\theta^{k-1}(\eta) f_{k-1}$}] (levelFourtwistk1) at (\xlevelthree+0,11*\ywidth) {};
		\draw[bend angle=20, bend left] (levelFourcoeff0.west) to (levelFourtwist0.west);
		\draw[bend angle=20, bend left] (levelFourcoeff1.west) to (levelFourtwist1.west);
		\draw[bend angle=20, bend left] (levelFourcoeff2.west) to (levelFourtwist2.west);
		\draw[bend angle=20, bend left] (levelFourcoeffk2.west) to (levelFourtwistk2.west);
		\draw[bend angle=20, bend left] (levelFourcoeffk1.west) to (levelFourtwistk1.west);

		\draw[dotted] (\xlevelfour-0.5*\xwidth,-0.5*\ywidth) rectangle (\xlevelfour+0.5*\xwidth,\nval*\ywidth-0.5*\ywidth);
		\node[coeffNodes,label=right:{$f_0$}]  (levelFivecoeff0) at (\xlevelfour+0,0*\ywidth) {};
		\node[coeffNodes,label=right:{$f_1$}]  (levelFivecoeff1) at (\xlevelfour+0,1*\ywidth) {};
		\node[coeffNodes,label=right:{$f_2$}]  (levelFivecoeff2) at (\xlevelfour+0,2*\ywidth) {};
		\node 								   (levelFivecoeffdots) at (\xlevelfour+0,3.2*\ywidth) {\footnotesize $\vdots$};
		\node[coeffNodes,label=right:{$f_{m-k-2}$}]  (levelFivecoeffk2) at (\xlevelfour+0,4*\ywidth) {};
		\node[coeffNodes,label=right:{$f_{m-k-1}$}]  (levelFivecoeffk1) at (\xlevelfour+0,5*\ywidth) {};
		\node[coeffNodes,label=right:{$f_{m-k}$}]  (levelFivecoeffkirrel1) at (\xlevelfour+0,6*\ywidth) {};
		\node  (levelFivecoeffirrel0) at (\xlevelfour+0,7.2*\ywidth) {$\vdots$};
		\node[coeffNodes,label=right:{$f_{k-1}$}]  (levelFivecoeffkirrel2) at (\xlevelfour+0,8*\ywidth) {};
		\node[twistNodes,label=right:{$\eta f_{0}$}] (levelFivetwist0) at (\xlevelfour+0,9*\ywidth) {};
		\node[twistNodes,label=right:{$\theta(\eta) f_{1}$}] (levelFivetwist1) at (\xlevelfour+0,10*\ywidth) {};
		\node[twistNodes,label=right:{$\theta^2(\eta) f_{2}$}] (levelFivetwist2) at (\xlevelfour+0,11*\ywidth) {};
		\node (levelFivetwistdots) at (\xlevelfour+0,12.2*\ywidth) {\footnotesize $\vdots$};
		\node[twistNodes,label=right:{$\theta^{m-k-2}(\eta) f_{m-k-2}$}] (levelFivetwistk2) at (\xlevelfour+0,13*\ywidth) {};
		\node[twistNodes,label=right:{$\theta^{m-k-1}(\eta) f_{m-k-1}$}] (levelFivetwistk1) at (\xlevelfour+0,14*\ywidth) {};
		\draw[bend angle=20, bend left] (levelFivecoeff0.west) to (levelFivetwist0.west);
		\draw[bend angle=20, bend left] (levelFivecoeff1.west) to (levelFivetwist1.west);
		\draw[bend angle=20, bend left] (levelFivecoeff2.west) to (levelFivetwist2.west);
		\draw[bend angle=20, bend left] (levelFivecoeffk2.west) to (levelFivetwistk2.west);
		\draw[bend angle=20, bend left] (levelFivecoeffk1.west) to (levelFivetwistk1.west);

		\node at (\xlevelzero, \ylabelpos) {$\mathrm{(a)}$};
		\node at (\xlevelone,  \ylabelpos) {$\mathrm{(b)}$};
		\node at (\xleveltwo,  \ylabelpos) {$\mathrm{(c)}$};
		\node at (\xlevelthree,\ylabelpos) {$\mathrm{(d)}$};
		\node at (\xlevelfour, \ylabelpos) {$\mathrm{(e)}$};
		
		\end{tikzpicture}
	\end{center}
	\vspace{-0.5cm}
	\caption{Overview of discussed code constructions. Boxes $\square$ represent (possibly non-zero) coefficients that can be chosen independently, filled boxes represent coefficients that depend on other coefficients (dependency indicated by arrows). (a) Gabidulin codes, (b) $\theta$-twisted Gabidulin codes, (c) generalized $\theta$-twisted Gabidulin codes, (d) new Gabidulin (type I) codes, and (e) new Gabidulin (type II) codes (dimension $k$ is chosen larger in this case due to requirement $m-k \leq k$).}
	\label{fig:overview_code_constructions}
\end{figure}

 In the following sections, we will mainly discuss $\theta$-Gabidulin, $\theta$-twisted Gabidulin, and generalized $\theta$-twisted Gabidulin codes. Both types of new $\theta$-Gabidulin codes are only treated in Section~\ref{sec:CharGabidulin}, where we use the new characterization results for Gabidulin codes to show that both new code classes are equivalent to Gabidulin codes whenever they are MRD.

\section{Invariants}\label{sec:invariants}

Let $\theta \in \Gal(\Fm/\Fq)$ and, for a code $C$ and integer $i$, denote by $\theta^i(C)$ the code obtained by applying $\theta^i$ entry-wise to all codewords of $C$.
Sums and intersections of such codes have been considered for several purposes in the literature of rank-metric codes.
Overbeck \cite{overbeck2008structural} proposed an attack on the Gabidulin--Paramonov--Tretjakov (GPT) cryptosystem \cite{gabidulin1991ideals}, which is based on the fact that (for small enough $i$) the code $C+\theta(C)+ \dots + \theta^i(C)$ has much smaller dimension for a Gabidulin code than for the majority of linear codes.
The attack has been modified several times to break multiple modifications of the GPT system \cite{ht16p,otmani2018improved,horlemann2018extension,coggia2019security}.
Furthermore, it is known that the dimension of the code $C + \theta(C)$ is an invariant of an $\Fmk$ code $C$ under code equivalences (\cite{ho16,puchinger2018construction, giuzzi2019identifiers}).
In \cite{ho16}, this dimension was used to derive a criterion for checking whether a given code is Gabidulin or not.
In \cite{puchinger2018construction}, it was used to show that some generalized $\theta$-twisted Gabidulin codes are inequivalent to known constructions. 
Moreover, in \cite{giuzzi2019identifiers}, Giuzzi and Zullo considered the dimensions of $C \cap \theta(C)$ and $C\cap \theta(C) \cap \theta^2(C)$, in order to give a distinguisher for twisted Gabidulin codes.
In the following we generalize these invariants.

\begin{lemma}\label{lem:invariant}
 Let $\sigma_1,\ldots, \sigma_r$ be distinct elements of $\Gal(\Fm/\Fq)$ and let $C_1,C_2$ be two equivalent $\Fmk$  codes. Then the following facts hold.
\begin{enumerate}
\item The codes
 $\cS_1 :=\sigma_1(C_1)+\sigma_2(C_1)+\ldots+\sigma_r(C_1)$ and $\cS_2 :=\sigma_1(C_2)+\sigma_2(C_2)+\ldots+\sigma_r(C_2)$ are equivalent.
In particular, $\dim \cS_1 = \dim \cS_2$.
\item
 The codes  $\T_1 :=\sigma_1(C_1)\cap \sigma_2(C_1)\cap \ldots\cap \sigma_r(C_1) $ and $\T_2:= \sigma_1(C_2)\cap \sigma_2(C_2)\cap \ldots\cap \sigma_r(C_2) $ are equivalent.
In particular, $\dim \T_1 = \dim \T_2$.
\end{enumerate}
\end{lemma}

\begin{proof}
 Since $C_1$ and $C_2$ are equivalent, there exist $\tau \in \Aut(\Fm)$, $A\in \GL_n(\Fq)$ such that
$C_1=\tau(C_2)A$. Therefore,
\begin{align*}
\cS_1&=\sigma_1(C_1)+\sigma_2(C_1)+\ldots+\sigma_r(C_1) \\
&=\sigma_1(\tau(C_2))\sigma_1(A)+\sigma_2(\tau(C_2))\sigma_2(A)+\ldots+\sigma_r(\tau(C_2))\sigma_r(A) \\
&\stackrel{(*)}{=} \tau(\sigma_1(C_2))A+\tau(\sigma_2(C_2))A+\ldots+\tau(\sigma_r(C_2))A  \\
&=\tau(\sigma_1(C_2)+\sigma_2(C_2)+\ldots+\sigma_r(C_2))A=
\tau(\cS_2)A,
\end{align*}
and
\begin{align*}
\T_1 &=\sigma_1(C_1)\cap \sigma_2(C_1)\cap \ldots\cap \sigma_r(C_1) \\
&=\sigma_1(\tau(C_2))\sigma_1(A)\cap\sigma_2(\tau(C_2)) \sigma_2(A)\cap\ldots\cap \sigma_r(\tau(C_2)) \sigma_r(A) \\
&\stackrel{(*)}{=} \tau(\sigma_1(C_2))A\cap\tau(\sigma_2(C_2))A\cap\ldots\cap\tau(\sigma_r(C_2))A  \\
&=\tau( \sigma_1(C_2)\cap \sigma_2(C_2)\cap \ldots\cap \sigma_r(C_2) )A=\T_2,
\end{align*}
where the equalities $(*)$ follow from the fact that $\Aut(\Fm)$ is a cyclic group, and therefore abelian, $\Gal(\Fm/\Fq)\subseteq \Aut(\Fm)$ and the $\sigma_i$'s fix all the elements in $\Fq$.
\end{proof}

Lemma~\ref{lem:invariant} implies that if two $\Fmk$ codes $C_1,C_2$ have different dimensions of $\cS_1$ and $\cS_2$ (or of $\T_1$ and $\T_2$), then they must be inequivalent. Hence, checking the dimensions of $\cS_1$ and $\cS_2$ (or of $\T_1$ and $\T_2$) for different choices of the $\sigma_i$'s gives a sufficient condition for codes to be inequivalent.

In the following, we restrict to the special case of consecutive powers of a fixed $\sigma \in \Gal(\Fm/\Fq)$, i.e.,  $\sigma_i = \sigma^{i-1}$, since in this case, we have additional interesting properties.
Motivated by Lemma \ref{lem:invariant}, we introduce the following setting and definitions. Let $\mathcal P_{q^m}(n)$ denote the set of all $\Fm$-subspaces of $\Fm^n$. For any automorphism $\autom \in \mathrm{Gal}(\Fm/\Fq)$ and integer $0\leq i \leq n$, we consider the maps
$$\begin{array}{rccl}
\cS_i^\autom : & \mathcal P_{q^m}(n)& \longrightarrow & \mathcal P_{q^m}(n) \\
&C & \longmapsto & \sum\limits_{j=0}^i\autom^j(C), \\
\end{array}$$
$$\begin{array}{rccl}
\T_i^\autom : & \mathcal P_{q^m}(n)& \longrightarrow & \mathcal P_{q^m}(n) \\
&C & \longmapsto & \bigcap\limits_{j=0}^i\autom^j(C), \\
\end{array}$$
and  the integers
$$\begin{array}{rlcrl}
s_i^\autom(C)&\!\!\!\!:= \dim(\cS_i^\autom(C)), & \qquad \qquad &  t_i^\autom(C)&\!\!\!\!:=\dim(\T_i^\autom(C)), \\
 \Delta_i^\autom(C)&\!\!\!\!:=s_{i+1}^\autom(C)-s_i^\autom(C), & \qquad \qquad & \Lambda_i^\autom(C)&\!\!\!\!:=t_{i}^\autom(C)-t_{i+1}^\autom(C).\end{array}$$

\begin{definition}
With the notation above:
\begin{enumerate}
\item $s_i^\autom(C)$ is called the \emph{$i$-th $\autom$-sum-dimension of $C$}, and $\Delta_i^\autom(C)$ the \emph{$i$-th $\autom$-sum-increment of $C$}. 

\item $t_i^\autom(C)$ is called the \emph{$i$-th $\autom$-intersection-dimension of $C$}, and $\Lambda_i^\autom(C)$ the \emph{$i$-th $\autom$-intersection-decrement of $C$}. 
\end{enumerate}
\end{definition}

As a consequence of Lemma \ref{lem:invariant}, we get that the sequences $\{s^\autom_i(C)\}$, $\{\Delta_i^\autom(C) \}$, $\{t^\autom_i(C)\}$ and $\{\Lambda_i^\autom(C)\}$ are invariants of linear rank-metric codes, i.e., they are stable under code equivalence. They can also be efficiently computed, as shown in Theorem~\ref{thm:computational_complexity}

A first property that we show is that the maps $\cS_i^\autom$ and $\T_i^\autom$ are connected by a duality relation:

\begin{proposition}\label{prop:DualitySiTi}
Let $C$ be an $\Fmk$ code. Then $\T_i^\autom(C)^\perp=\cS_i^\autom(C^\perp)$. In particular, $t_i^\autom(C)=n-s_i^\autom(C^\perp)$ and $\Lambda_i^\autom(C)=\Delta_i^\autom(C^\perp)$.
\end{proposition} 

\begin{proof}
Since $\autom(C^\perp)=\autom(C)^\perp$, we get 
\begin{align*}
\T_i^\autom(C)^\perp =\left(\bigcap\limits_{j=0}^i\autom^j(C)\right)^\perp = \sum\limits_{j=0}^i\left(\autom^j(C)^\perp\right) =  \sum\limits_{j=0}^i\left(\autom^j(C^\perp)\right) =\cS_i^\autom(C^\perp).
\end{align*}
The equalities $t_i^\autom(C)=n-s_i^\autom(C^\perp)$ and $\Lambda_i^\autom(C)=\Delta_i^\autom(C^\perp)$ immediately follow, using the fact that $\dim(U^\perp)=n-\dim(U)$, for any $U \in \mathcal P_{q^m}(n)$.
\end{proof}

Using the fact that $\cS_0^\autom(C)=\T_0^\autom(C)=C$ we also get the following relations.

\begin{proposition}\label{prop:SiTi}
Let $C$ be an $\Fmk$ code. Then
\begin{enumerate}
\item $t^\autom_1(C)= 2k - s^\autom_1(C)$, 
\item $\Delta_0^\autom(C)=\Lambda_0^\autom(C)$.
\end{enumerate}
\end{proposition}

\begin{proof}
\begin{enumerate}
\item We have $\Taut_1(C)=C\cap\autom(C)$ and thus $t^\autom_1(C)=\dim(C\cap\autom(C)) = \dim(C)+\dim(\autom(C)) - s^\autom_1(C) = 2k - s^\autom_1(C)$.  
\item $\Delta_0^\autom(C)=s_1^\autom(C)-s_0^\autom(C)=s_1^\autom(C)-k=k-t_1^\autom(C)=t_0^\autom(C)-t_1^\autom(C)=\Lambda_0^\autom(C)$. \qedhere
\end{enumerate}
\end{proof}

We now derive more properties of the $\autom$-sum sequence, before doing the analogue for the intersection sequence.

\begin{proposition}\label{prop:propertiesSi}
Let $C\subseteq \Fm^n$ be an $\Fmk$ code. Then:
\begin{enumerate}
\item $k=s_0^\autom(C)\leq s_1^\autom(C)\leq \ldots \leq s_{n-k}^\autom(C)\leq n$.
\item $\cS^\autom_i\circ \cS^\autom_j=\cS^\autom_{i+j}$.
\item $s^\autom_i(C)=s^\autom_{i+1}(C)$ if and only if $\cS^\autom_i(C)$ has a basis of elements in $\Fq^n$.
\item If $s^\autom_i(C)=s^\autom_{i+1}(C)$ then $s^\autom_{i+j}(C)=s^\autom_i(C)$ for all $j \geq 0$.
\item $s^\autom_{n-k}(C)=s^\autom_{n-k+j}(C)$ for all $j \geq 0$.
\item  $k\geq \Delta_0^\autom(C)\geq \Delta_1^\autom(C)\geq \ldots \geq \Delta_{n-k}^\autom(C)=0$.
\item $s_i^{\autom}(C)=k+\sum_{j=0}^{i-1}\Delta_j^\autom(C)$.
\end{enumerate}
\end{proposition}

\begin{proof}
\begin{enumerate}
\item By definition we have $ s_0^\autom(C) =k$. The rest follows from $\Saut_i(C)\subseteq \Saut_{i+1}(C) \subseteq \Fm^n$.
\item It holds that $\Saut_i(\Saut_j(C))=\sum_{\ell=0}^i\autom^\ell(\Saut_j(C))=\sum_{\ell=0}^i\sum_{r=0}^j \autom^{\ell+r}(C)=\sum_{h=0}^{i+j}\autom^h(C)=\Saut_{i+j}(C)$.
\item Suppose $s^\autom_i(C)=s^\autom_{i+1}(C)$, then $\Saut_i(C)=\Saut_{i+1}(C)$, and by part 2, we get $\Saut_1(\Saut_i(C))=\Saut_i(C)$. This is true if and only if $\autom(\Saut_i(C))=\Saut_i(C)$, and we can conclude using \cite[Lemma 4.5]{ho16}.
\item The equality $s^\autom_i(C)=s^\autom_{i+1}(C)$ implies that $\autom(\Saut_i(C))=\Saut_i(C)$, and therefore, $\Saut_{i+j}(C)=\autom^j(\Saut_i(C))=\Saut_i(C)$
 for all $j \geq 0$. 
\item Let $r^\autom(C)=\min\{i \mid s^\autom_i(C)=s^\autom_{i+1}(C) \}$. If $r^\autom(C)\leq n-k$, then by part 4 we can conclude. Suppose by contradiction that $r^\autom(C)>n-k$. Then we get a chain 
$k=s^\autom_0(C)<s^\autom_1(C)< \ldots <s^\autom_{n-k}(C)<s^\autom_{n-k+1}(C).$
This implies that $s^\autom_i(C) \geq k+i$, and in particular $s^\autom_{n-k+1}(C)\geq k+n-k+1=n+1$, but this is impossible since $\Saut_{n-k+1}(C)\subseteq \Fm^n$.
\item First we prove that $\Delta_0^\autom(C)\leq k$.  We have $\Saut_1(C)=C+\autom(C)$ and thus $s^\autom_1(C)=\dim(C+\autom(C))\leq \dim(C)+\dim(\autom(C))=s^\autom_0(C)+k$. Furthermore, we have $\Delta_{n-k} = 0$ by part 5.

Now we prove that $\Delta_i \geq \Delta_{i+1}$. Suppose $\Delta_i=s^\autom_{i+1}(C)-s^\autom_{i}(C)=r$. Then $\dim (\Saut_i(C)+\autom(\Saut_i(C)))=\dim (\Saut_i(C))+r$. This implies that $\autom(\Saut_i(C))=W+U$, where $W\subseteq \Saut_i(C)$, $U\cap \Saut_i(C)=\{0\}$ and $\dim U=r$. Hence, $\Saut_{i+2}(C)=\Saut_1(\Saut_{i+1}(C))=\Saut_{i+1}(C)+\autom(\Saut_{i+1}(C))=\Saut_i(C)+U+\autom(\Saut_i(C))+\autom(U)$. However, $U \subseteq \autom(\Saut_i(C))$, and therefore $\Saut_{i+2}(C)=\Saut_i(C)+\autom(\Saut_i(C))+\autom(U)=\Saut_{i+1}(C)+\autom(U)$. Since $\dim \autom(U)=\dim(U)=r$, we get $\Delta_{i+1} = \dim(\Saut_{i+2}(C)) - \dim(\Saut_{i+1}(C))\leq r = \Delta_i$.

\item $\sum_{j=0}^{i-1}\Delta_j^\autom(C) = \sum_{j=0}^{i-1} (s_{j+1}^\autom(C)-s_{j}^\autom(C)) = s_i^\autom(C) - s_0^\autom(C) =  s_i^\autom(C) -k$. \qedhere
\end{enumerate}
\end{proof}

The following results are analogous results for the $\sigma$-intersection sequences.

\begin{proposition}\label{prop:propertiesTi}
Let $C\subseteq \Fm^n$ be an $\Fmk$ code. Then:
\begin{enumerate}
\item $k=t_0^\autom(C)\geq t_1^\autom(C)\geq \ldots \geq t_{k}^\autom(C)\geq 0$.
\item $\T^\autom_i\circ \T^\autom_j=\T^\autom_{i+j}$.
\item $t^\autom_i(C)=t^\autom_{i+1}(C)$ if and only if $\T^\autom_i(C)$ has a basis of elements in $\Fq^n$.
\item If $t^\autom_i(C)=t^\autom_{i+1}(C)$ then $t^\autom_{i+j}(C)=t^\autom_i(C)$ for all $j \geq 0$.
\item $t^\autom_{k}(C)=t^\autom_{k+j}(C)$ for all $j \geq 0$.
\item $k\geq \Lambda_0^\autom(C) \geq \Lambda_1^\autom(C) \geq \ldots \geq \Lambda_k^\autom(C)=0$.
\item $ t_i^{\autom}(C)=k-\sum_{j=0}^{i-1}\Delta_j^\autom(C)$.
\end{enumerate}
\end{proposition}

\begin{proof}
\begin{enumerate}
\item This is clear, since $\Taut_{i+1}(C)\subseteq \Taut_{i}(C) \subseteq \Fm^n$.
\item We have $\Taut_i(\Taut_j(C))=\bigcap_{\ell=0}^i\autom^\ell(\Taut_j(C))=\bigcap_{\ell=0}^i\bigcap_{r=0}^j \autom^{\ell+r}(C)=\bigcap_{s=0}^{i+j}\autom^s(C)=\Taut_{i+j}(\C)$.
\item By Proposition \ref{prop:DualitySiTi}, we have that $t_i^\autom(C)=t_{i+1}^\autom(C)$ if and only if $s_i^\autom(C^\perp)=s_{i+1}^\autom(C^\perp)$, which in turn  is equivalent to  $C^\perp$ having a basis of elements in $\Fq^n$, by Proposition \ref{prop:propertiesSi}. This is in turn equivalent to $C$ having a parity check matrix with entries from $\Fq$, which 
 happens if and only if there exists a basis  of elements in $\Fq^n$ for the code itself.
\item[4.-7.]
In an analogous way, the remaining statements follow from Proposition \ref{prop:propertiesSi}  and the duality result of Proposition \ref{prop:DualitySiTi}. \qedhere
\end{enumerate}
\end{proof}

The following theorem shows that the invariants presented in this section can be computed efficiently for any given code.

\begin{theorem}\label{thm:computational_complexity}
Let $C$ be an $[n,k]_{q^m}$ code, $\sigma,\sigma_1,\dots,\sigma_r \in \mathrm{Gal}(\Fqm/\Fq)$ be automorphisms.
\begin{enumerate}
\item The dimension of $\sigma_1(C)+\sigma_2(C)+\ldots+\sigma_r(C)$ can be computed in
\begin{align*}
O\!\left(\min\left\{n^{\omega-1}rk, nr^{\omega-1}k^{\omega-1}\right\}\right)
\end{align*}
operations over $\Fqm$, where $2 \leq \omega \leq 3$ is the matrix multiplication exponent. 
\item The dimension of $\sigma_1(C)\cap \sigma_2(C)\cap \ldots\cap \sigma_r(C)$ can be computed in
\begin{align*}
O\!\left(\min\left\{n^{\omega-1}r(n-k), nr^{\omega-1}(n-k)^{\omega-1}\right\}\right)
\end{align*}
operations over $\Fqm$.
\item The sequences $\{s^\autom_i(C)\}_{i=0}^{\infty}$, $\{\Delta_i^\autom(C) \}_{i=0}^{\infty}$, $\{t^\autom_i(C)\}_{i=0}^{\infty}$, $\{\Lambda_i^\autom(C)\}_{i=0}^{\infty}$ can be computed in
\begin{align*}
O\!\left(n^{\omega-1} k(n-k)\right)
\end{align*}
operations over $\Fqm$.
\end{enumerate}
\end{theorem}

\begin{proof}
Let $G \in \Fqm^{k \times n}$ be a generator and $H \in \Fqm^{(n-k) \times n}$ be a parity-check matrix of the code $C$. Without loss of generality we assume that we know (or can efficiently compute) one of them. The other matrix can be determined from the known one in $O(\max\{n-k,k\}^{\omega-1}n)$ field operations using Gaussian elimination (for the complexity, see  e.g., \cite[Theorem~2.10]{storjohann2000algorithms}).

\begin{enumerate}
\item  The code $\sigma_1(C)+\sigma_2(C)+\ldots+\sigma_r(C)$ is generated by the rows of the matrix
\begin{align*}
\begin{pmatrix}
\sigma_1(G) \\
\sigma_2(G) \\
\vdots \\
\sigma_r(G)
\end{pmatrix} \in \Fqm^{k r \times n},
\end{align*}
so its dimension can be computed by determining the rank of the matrix.
Since the rank of an $a \times b$ matrix can be computed in $O(\min\{a^{\omega-1}b,ab^{\omega-1}\})$, the claimed complexity follows.

\item Analogously, the dimension of
$\sigma_1(C)\cap \sigma_2(C)\cap \ldots\cap \sigma_r(C)$ can be computed by determining the rank of
$(\sigma_1(H)^\top, \sigma_2(H)^\top, \dots, \sigma_r(H)^\top)^\top  \in \Fqm^{(n-k) r \times n}$
which is a parity-check matrix of the intersection code.

\item It suffices to show that we can compute $s^\autom_1(C),\dots,s^\autom_{n-k}(C)$ efficiently since 
\begin{itemize}
\item by Proposition~\ref{prop:propertiesSi}, the sequence $s^\autom_i(C)$ converges after at most $n-k$ steps (i.e., $s^\autom_{n-k+j}(C) = s^\autom_{n-k}(C)$ for all $j\geq 0$),
\item by Proposition~\ref{prop:DualitySiTi}, the sequence $\{t^\autom_i(C)\}$ can be computed from the $s^\autom_i(C^\bot)$ sequence of the dual code (which converges after at most $k$ steps), and 
\item the difference sequences follow by $O(\max\{k,n-k\})$ subtractions.
\end{itemize}
We can compute these $n-k$ values $s^\autom_i(C)$ by first determining the column rank profile (i.e., the row indices of leading ones in a column echelon form) of the matrix
\begin{align*}
\begin{pmatrix}
G \\
\sigma(G) \\
\vdots \\
\sigma^{n-k}(G)
\end{pmatrix} \in \Fqm^{k(n-k+1) \times n},
\end{align*}
which can be done in $O(n^{\omega-1} k(n-k))$ operations \cite[Thm.~2.10]{storjohann2000algorithms}. We get $s^\autom_i(C)$ by counting the elements of the column rank profile that are contained in the first $(i+1)k$ rows. We save a factor $n-k$ compared to naively computing each $s^\autom_i(C)$ individually as in part 1. \qedhere
\end{enumerate}
\end{proof}

Note that the difference sequences $\{\Delta_i^\sigma(C)\}$ and $\{\Lambda_i^\sigma(C)\}$ are mainly studied in this section. 
In the subsequent section(s) we focus on the sum and the intersection dimension sequences. The difference sequences can then easily be determined from them. We will use $\{\Delta_i^\sigma(C)\}$ and $\{\Lambda_i^\sigma(C)\}$ again in Section \ref{sec:CharGabidulin}, to characterize Gabidulin codes.

\section{The Sum and Intersection Sequences for Known MRD Constructions}\label{sec:sequences}

In this section we are going to study the properties of the sequences introduced above for Gabidulin, twisted Gabidulin and some generalized twisted Gabidulin codes. Since the intersection sequence is fully determined by the sum sequence, we exemplarily determine the intersection sequence for Gabidulin and (narrow-sense) twisted Gabidulin codes, but not for generalized twisted Gabidulin codes.

For simplicity we will represent the field automorphism $\sigma$ as a power of a generator $\theta$ of the Galois group  $\Gal(\Fm/\Fq)$.

\subsection{Gabidulin Codes}

In general a $\theta$-Gabidulin code $C$ can also be a $\bar \theta$-Gabidulin code for another generator $\bar \theta$ of $\Gal(\Fm/\Fq)$. We will prove in Theorem \ref{thm:upperboundnumGab} that this cannot happen for many $\bar \theta$'s. However, it is straightforward to see that a $\theta$-Gabidulin code is always also a $\theta^{-1}$-Gabidulin code.

\begin{lemma}\label{lem:Gabtheta=theta-1}
Let $\theta$ be a generator of the Galois group $\Gal(\Fm/\Fq)$ and $g \in \Fm^n$ such that $\rk_q(g)=n$. Then $\G_{k,\theta}(g)=\G_{k,\theta^{-1}}(\theta^{k-1}(g))$.
\end{lemma}
\begin{proof}We have  
\begin{align*}
\G_{k,\theta}(g) &=\langle g,\theta(g), \dots, \theta^{k-1}(g)\rangle \\
&=\langle \theta^{k-1}(g),\theta^{-1}(\theta^{k-1}(g)), \dots, \theta^{-(k-1)}(\theta^{k-1}(g))\rangle\\
&=\G_{k,\theta^{-1}}(\theta^{k-1}(g)).  \qedhere
\end{align*}
\end{proof}

The following result gives the $\theta^r$-sequences for $\theta$-Gabidulin codes. Note that the computation for the special case in which $r=1$ and $\theta$ is the $q$-Frobenius automorphism was already derived by Overbeck in \cite{overbeck2008structural}. 
\begin{proposition}\label{prop:GabidulinSeq}
Let $C:= \G_{k,\theta}(g)$ be a $\theta$-Gabidulin code and $i,r \in \mathbb N$ such that $0\leq r<m$.
\begin{enumerate}
\item For $0\leq r \leq k$ we have $\cS^{\theta^{r}}_i(C)=\cS^{\theta}_{ir}(C)= \G_{k+ir, \theta}(g)$ and for $ m-k \leq r \leq m$ we have $\cS^{\theta^{r}}_i(C)= \G_{k+i(m-r), \theta^{-1}}(\theta^{k-1}(g))$.\footnote{Although not properly defined, we assume that $\G_{s, \theta}(g) =\F_{q^m}^n$ for $s\geq n$.}
In particular, we have $s^{\theta^r}_i(C)=\min\{e,n\}$, where
$$e=\begin{cases} k+ir  & \mbox{ if }  0\leq r \leq k\\
k+i(m-r) & \mbox{ if } m-k \leq r \leq m \end{cases}.$$
\item  If $k<r\leq n-k$ or $m-n+k\leq r < m-k$ we have  $s_1^{\theta^r}(C)=2k$.
\item  If $r>k$ and $r >n-k$, then  $s_1^{\theta^r}(C)\geq k+n-r$.   If $r<m-k$ and $r <m-n+k$, then $s_1^{\theta^r}(C)\geq k+n-m+r$.
\end{enumerate}
\end{proposition}

\begin{proof}
\begin{enumerate}
\item If $0\leq r\leq k$, then $\cS^{\theta^r}_i(C)=\langle g, \theta(g), \ldots,\theta^{k-1}(g),\theta^k(g),\ldots, \theta^{k+ir-1}(g)\rangle = \G_{k+ir, \theta}(g)$. 
If $m-k \leq r \leq m-1$, we can write  $\theta^r=(\theta^{-1})^{m-r}$, with $0\leq m-r\leq k$, and thus get $\cS^{\theta^{r}}_i(C)=\cS^{(\theta^{-1})^{m-r}}_i(\G_{k,\theta^{-1}}(\theta^{k-1}(g)))=\G_{k+i(m-r),\theta^{-1}}(\theta^{k-1}(g))$, by Lemma~\ref{lem:Gabtheta=theta-1}. 
The computation of $s^{\theta^{r}}_i(C)$ follows from Proposition \ref{cor:rankMoore}.

\item If $k<r\leq n-k$, then $\cS^{\theta^r}_1(C)=\langle g, \theta(g),\ldots, \theta^{k-1}(g),\theta^r(g),\ldots, \theta^{r+k-1}(g)\rangle$, which has dimension $2k$, by Proposition \ref{cor:rankMoore}. If $m-n+k\leq r < m-k$ the statement follows with Lemma \ref{lem:Gabtheta=theta-1}.

\item If $r>k$ and $r >n-k$, then $\cS^{\theta^r}_1(C)\supseteq\langle g, \theta(g),\ldots, \theta^{k-1}(g),\theta^r(g),\ldots, \theta^{n-1}(g)\rangle$, which has dimension $n+k-r$, by Proposition \ref{cor:rankMoore}. If $r<m-k$ and $r <m-n+k$ the statement follows with Lemma \ref{lem:Gabtheta=theta-1}. \qedhere
\end{enumerate}
\end{proof}

In particular, for $m=n\geq 2k$, we have 
$$s^{\theta^r}_1(C)
=\begin{cases} k+r  & \mbox{ if }  0\leq r \leq k\\
2k & \mbox{ if }  k+1\leq r \leq n-k-1\\
k+(n-r) & \mbox{ if } n-k \leq r \leq n-1 \end{cases}.$$

We present example sequences and the corresponding statements from Proposition~\ref{prop:GabidulinSeq} in Appendix~\ref{app:examples}.

\begin{proposition}\label{prop:GabidulinSeqT}
Let $C:= \G_{k,\theta}(g)$ be a $\theta$-Gabidulin code and $i,r \in \mathbb N$ such that $0\leq r<m$.
\begin{enumerate}
\item For $0\leq r \leq k$ we have $\T^{\theta^{r}}_i(C)=\T^{\theta}_{ir}(C)= \G_{k-ir, \theta}(\theta^{ir}(g))$ and for $ m-k \leq r < m$ we have $\T^{\theta^{r}}_i(C)= \G_{k-i(m-r), \theta^{-1}}(\theta^{k-i(m-r)-1}(g))$.
\footnote{Although not properly defined, we assume that $\G_{s, \theta}(g) =\{0\}$ for $s\leq 0$.}.
In particular, we have $t^{\theta^r}_i(C)=\max\{e,0\}$, where
$$e=\begin{cases} k-ir  & \mbox{ if }  0\leq r \leq k\\
k-i(m-r) & \mbox{ if } m-k \leq r \leq m \end{cases}.$$
\item  If $k<r\leq n-k$ or $m-n+k\leq r < m-k$ we have  $t_1^{\theta^r}(C)=0$.
\item  If $r>k$ and $r >n-k$, then  $t_1^{\theta^r}(C)\leq k-n+r$.   If $r<m-k$ and $r <m-n+k$, then $t_1^{\theta^r}(C)\leq k-n+m-r$.
\end{enumerate}
\end{proposition}

\begin{proof}
\begin{enumerate}
\item It is enough to prove it for $i=1$, then the claim follows by induction, since $\T^{\theta^{r}}_{i+1}=\T^{\theta^{r}}_1 \circ \T^{\theta^{r}}_i$ by part 2 of Proposition \ref{prop:propertiesTi}. For $0 \leq r \leq k$ we have  $\T^{\theta^r}_1(C)=\langle g,\ldots,\theta^{k-1}(g)\rangle \cap \langle\theta^r(g),\ldots, \theta^{k+r-1}(g)\rangle \supseteq \langle \theta^r(g),\ldots,\theta^{k-1}(g)\rangle$. The equality follows by part 1 of Proposition \ref{prop:SiTi} and part 1 of Proposition \ref{prop:GabidulinSeq}. If $m-k \leq r \leq m$ we write $\theta^r=(\theta^{-1})^{m-r}$ and use Lemma \ref{lem:Gabtheta=theta-1}.
\item[2.--3.] They follow from part 1 of Proposition \ref{prop:SiTi}. and parts 2--3 of Proposition \ref{prop:GabidulinSeq}. \qedhere
\end{enumerate}
\end{proof}

In particular, for $m=n\geq 2k$, we have 
$$t^{\theta^r}_1(C)
=\begin{cases} k-r  & \mbox{ if }  0\leq r \leq k\\
0 & \mbox{ if }  k+1\leq r \leq n-k-1\\
k-(n-r) & \mbox{ if } n-k \leq r \leq n-1 \end{cases}.$$

\subsection{Twisted Gabidulin Codes}
In this subsection we analyze the family of $\Fmk$ $\theta$-twisted Gabidulin codes, i.e., those which are linear over $\Fm$. 
The following result is a straightforward computation, analogous to Lemma \ref{lem:Gabtheta=theta-1}.

\begin{lemma}\label{lem:TGabtheta=theta-1}
Let $\theta$ be a generator of the Galois group $\Gal(\Fm/\Fq)$ and $g \in \Fm^n$ such that $\rk_q(g)$, and $\eta \in \Fm^*$. Then $\TGabShee{k,\theta}{\eta}{g} = \TGabShee{k,\theta^{-1}}{\eta^{-1}}{\theta^k(g)}$.
\end{lemma}
\begin{proof}
We have
\begin{align*}
\TGabShee{k,\theta}{\eta}{g}&= \langle g+\eta\theta^k(g),\theta(g), \dots, \theta^{k-1}(g)\rangle\\
&=\langle \eta^{-1}\theta^{-k}(\theta^{k}(g)) + \theta^k(g),   \theta^{-1}(\theta^{k}(g)),\theta^{-2}(\theta^{k}(g)), \dots, \theta^{-(k-1)}(\theta^{k}(g)) \rangle\\
& = \TGabShee{k,\theta^{-1}}{\eta^{-1}}{\theta^k(g)}. \qedhere
\end{align*}
\end{proof}

\begin{proposition}\label{prop:TwGabidulinSeq}
Let $g\in \Fm^n$ with $\rk_q(g)=n$, and $C := \TGabShee{k,\theta}{\eta}{g}$ be a $\theta$-twisted Gabidulin code, where $\eta \in \Fm^*$, 
and $i,r \in \mathbb N$ such that $0\leq r<m$.
\begin{enumerate}
\item For $1\leq r \leq k-1$ we have $\cS^{\theta^{r}}_i(C)= \G_{k+ir+1, \theta}(g)$ and for $ m-k-1 \leq r < m$ we have $\cS^{\theta^{r}}_i(C)= \G_{k+i(m-r)+1, \theta^{-1}}(\theta^{k}(g))$.\footnote{Although not properly defined, we assume that $\TGabShee{s, \theta}{\eta}{g} =\F_{q^m}^n$ for $s\geq n$.} In particular, for any $i>0$, we have $s^{\theta^r}_i(C)=\min\{e,n\}$, where
$$e=\begin{cases} k+ir+1  & \mbox{ if }  1\leq r \leq k-1\\
k+i(m-r)+1 & \mbox{ if } m-k+1 \leq r \leq m-1 \end{cases}.$$
\item  If $k\leq r \leq n-k$ or $m-n+k\leq r \leq m-k$, then  $s_1^{\theta^r}(C)=2k$.
\item  If $r \geq k$ and $r >n-k$, then  $s_1^{\theta^r}(C)\geq k+n-r$.   If $r \leq m-k$ and $r <m-n+k$, then $s_1^{\theta^r}(C)\geq k+n-m+r$.
\end{enumerate}
\end{proposition}

\begin{proof}
\begin{enumerate}
\item Let $1\leq r\leq k-1$ and $i=1$. Then 
$\cS^{\theta^r}_1(C)=\langle g+\eta \theta^k(g),\theta(g),\dots, \theta^{k-1}(g),$ $\theta^r(g)+\theta^r(\eta)\theta^{k+r}(g), \theta^{r+1}(g), \dots, \theta^{k+r-1}(g)\rangle,$
which is contained in $ \G_{k+r+1, \theta}(g)$. As $r\leq k-1$, we have that  $\cS^{\theta^r}_1(C) \supseteq \{\theta(g),\ldots, \theta^{r+k-1}(g)\}$.
Furthermore, it contains $g+\eta\theta^k(g)$ and $\theta^r(g)+\theta^r(\eta)\theta^{k+r}(g)$.
Since $1\leq r \leq k-1$, it contains  $\theta^k(g)$, and $\theta^{r}(g)$, and therefore  $g, \theta^{k+r}(g) \in \cS^{\theta^r}_1(C)$, and we deduce $\cS^{\theta^r}_1(C)= \G_{k+r+1, \theta}(g)$.
If $i >1$, by part 2 of Proposition~\ref{prop:propertiesSi}, we have $\cS^{\theta^r}_i(C)=\cS^{\theta^r}_{i-1}(\cS^{\theta^r}_1(C))$, and we conclude using part 1 of Proposition \ref{prop:GabidulinSeq}. 

 If $m-k+1 \leq r \leq m-1$ we write $\theta^{r}=(\theta^{-1})^{m-r}$, with $1\leq m-r \leq k-1$. With Lemma \ref{lem:TGabtheta=theta-1} we get  $\cS^{\theta^{r}}_i(C)= \cS^{(\theta^{-1})^{m-r}}_i( \TGabShee{k,\theta^{-1}}{\eta^{-1}}{\theta^k(g)}) =\G_{k+i(m-r)+1, \theta^{-1}}(\theta^{k}(g))$.

\item If $k\leq r\leq n-k$, then 
$\cS^{\theta^r}_1(C)=\langle g+\eta \theta^k(g),\theta(g) \ldots, \theta^{k-1}(g),\theta^r(g)+\theta^r(\eta)\theta^{r+k}(g),$ $\theta^{r+1}(g), \ldots, \theta^{r+k-1}(g)\rangle$
which  has dimension $2k$ by Proposition \ref{cor:rankMoore}.

If  $m-n+k\leq r \leq m-k$ the claim follows with Lemma \ref{lem:TGabtheta=theta-1}, writing $\theta^r=(\theta^{-1})^{m-r}$.

\item If $r \geq k$ and $k+r> n-1$, then 
$$\cS^{\theta^r}_1(C)\supseteq\langle g+\eta \theta^k(g), \theta(g),\ldots, \theta^{k-1}(g),\theta^r(g)+\theta^r(\eta)\theta^{r+k}(g),\theta^{r+1}(g),\ldots, \theta^{n-1}(g)\rangle,$$ which has dimension $k+n-r$, by Proposition  \ref{cor:rankMoore}.

If $r\leq m-k$ and $r <m-n+k+1$ the claim follows again with Lemma \ref{lem:TGabtheta=theta-1}. \qedhere
\end{enumerate}
\end{proof}

In particular, for $m=n\geq 2k$, we have 
$$s^{\theta^r}_1(C)=\begin{cases} k+r+1  & \mbox{ if }  0\leq r < k\\
2k & \mbox{ if }  k\leq r \leq n-k\\
k+(n-r)+1 & \mbox{ if } n-k < r \leq n-1 \end{cases}.$$

We present example sequences and the corresponding statements from Proposition~\ref{prop:TwGabidulinSeq} in Appendix~\ref{app:examples}.

\begin{proposition}
Let $g\in \Fm^n$ with $\rk_q(g)=n$, and $C := \TGabShee{k,\theta}{\eta}{g}$ be a $\theta$-twisted Gabidulin code, where $\eta \in \Fm^*$, %
and $i,r \in \mathbb N$ such that $0\leq r<m$.
\begin{enumerate}
\item For $1\leq r \leq k-1$ we have $\T^{\theta^{r}}_i(C)=\T^{\theta}_{ir}(C)= \G_{k-ir-1, \theta}(\theta^{ir}(g))$ and for $ m-k+1 \leq r \leq m-1$ we have $\T^{\theta^{r}}_i(C)= \G_{k-i(m-r)-1, \theta^{-1}}(\theta^{k-i(m-r)-1}(g))$.\footnote{Although not properly defined, we assume that $\G_{s, \theta}(g) =\{0\}$ for $s\leq 0$.} In particular, for any $i>0$, we have $t^{\theta^r}_i(C)=\max\{e,0\}$, where
$$e=\begin{cases} k-ir-1  & \mbox{ if }  1\leq r \leq k-1\\
k-i(m-r)-1 & \mbox{ if } m-k+1 \leq r \leq m-1 \end{cases}.$$
\item  If $k\leq r\leq n-k$ or $m-n+k\leq r \leq m-k$ we have  $t_1^{\theta^r}(C)=0$.
\item  If $r\geq k$ and $r >n-k$, then  $t_1^{\theta^r}(C)\leq k-n+r$.   If $r\leq m-k$ and $r <m-n+k$, then $t_1^{\theta^r}(C)\leq k-n+m-r$.
\end{enumerate}
\end{proposition}

\begin{proof}
\begin{enumerate}
\item For $i=1$ the claim follows from part 1 of Proposition \ref{prop:SiTi} and part 1 of Proposition \ref{prop:TwGabidulinSeq}. That $\T^{\theta^{r}}_1(C)$ is equal to some Gabidulin code can be shown analogously to the proof of Proposition \ref{prop:GabidulinSeqT}. For larger $i$ the claim follows by induction, since $\T^{\theta^{r}}_{i+1}=\T^{\theta^{r}}_1 \circ \T^{\theta^{r}}_i$ by part 2 of Proposition \ref{prop:propertiesTi}, together with part 1 of Proposition\ref{prop:GabidulinSeqT} . 
\item[2.--3.] They follow from part 1 of Proposition \ref{prop:SiTi} and parts 2--3 of Proposition \ref{prop:TwGabidulinSeq}. \qedhere
\end{enumerate}
\end{proof}

In particular, for $m=n\geq 2k$, we have 
$$t^{\theta^r}_1(C)=\begin{cases} k-r -1 & \mbox{ if }  0\leq r < k\\
0 & \mbox{ if }  k\leq r \leq n-k\\
k-(n-r)-1 & \mbox{ if } n-k < r \leq n-1 \end{cases}.$$

The previous results imply:
\begin{corollary}\label{cor:inequiv1}
Let $1<k<n-1$, $g,h \in \Fm^n$ with $\rk_q(g)=\rk_q(h)=n$ and let $\theta$ be a generator of $\Gal(\Fm/\Fq)$. Moreover, let $1\leq r < m$ be such that $\gcd(r,m)=1$. Suppose that \emph{at least one} of the following holds:
\begin{enumerate}
\item[\mylabel{h1}{\textnormal{(h1)}}] $r < n-1$;
\item[\mylabel{h2}{\textnormal{(h2)}}] $r > m-n+1$;
\item[\mylabel{h3}{\textnormal{(h3)}}] $k<2n-m-1$;
\item[\mylabel{h4}{\textnormal{(h4)}}] $k>m-n+1$.
\end{enumerate}
Then $\TGabShee{k,\theta}{\eta}{g}$ is not equivalent to $\G_{k,\theta^r}(h)$. 
\end{corollary}

\begin{proof}
By Proposition \ref{prop:GabidulinSeq} we have  $ s_1^{\theta^r}(\G_{k,\theta^r}(h)) = k+1$. Hence, if $s_1^{\theta^r}(\TGabShee{k,\theta}{\eta}{g}) =s_1^{\theta^r}(\TGabShee{k,\theta}{\eta}{g}) \neq k+1$ then the two codes are inequivalent by Lemma \ref{lem:invariant}.

If $r \in \{1,\dots,\max(k-1,n-k)\}\cup \{\min(m-k+1,m-n+k),\dots,m-1\}$  we have that $s_1^{\theta^r}(\TGabShee{k,\theta}{\eta}{g}) > k+1$ by parts 1 and 2 of Proposition \ref{prop:TwGabidulinSeq}. 
If $r\in\{\max(k,n-k+1),\dots,n-3\}\cup\{m-n+3,\dots,\min(m-k,m-n+k-1)\}$ we have that $s_1^{\theta^r}(\TGabShee{k,\theta}{\eta}{g}) > k+1$ by part 3 of Proposition \ref{prop:TwGabidulinSeq}.
This proves the statement for one between \ref{h1} and \ref{h2} holds.. 

Assume now that $n-2<r<m-n+2$. This implies that $\max(k,n-k)< r < \min(m-k,m-n+k)$. Part 3 of Proposition \ref{prop:TwGabidulinSeq} implies that
$$s_1^{\theta^r} \geq k+n-r >2n-m $$ 
and also that
$$s_1^{\theta^r} \geq k+n-m+r >2k+n-m .$$ 
Thus $s_1^{\theta^r}$ is strictly greater than $k+1$ if one between \ref{h3} and \ref{h4} is satisfied.
\end{proof}

This implies the following general result on the equivalence of Gabidulin and twisted Gabidulin codes in dependence on $m$ and $n$.
The special case $m=n$ of was already proven in \cite{sheekey2016new} and \cite{lunardon2018generalized}.

\begin{corollary}\label{cor:sufficient_inequivalence_condition_Gab_TGab}
Let $1<k<n-1$. If $m< 2n-2$ (in particular $m=n$), then twisted Gabidulin codes of length $n$ and dimension $k$ over $\Fm$ are never equivalent to a Gabidulin code.
\end{corollary}
\begin{proof}
For $m<2n-2$ we get that $n-1\geq m-n+2$, hence either the first or the second condition of Corollary \ref{cor:inequiv1} is fulfilled.
\end{proof}

\subsection{Generalized Twisted Gabidulin Codes}

In this subsection, we analyze some special instances of generalized twisted Gabidulin codes with $\ell=1$ twists.
We derive (parts of) the sequences $s_i^\sigma$ for three cases (recall that $t=1$, $h=0$ are narrow-sense twisted Gabidulin codes discussed in the previous section):
\begin{itemize}
\item $h=0$ and arbitrary $1<t\leq n-k$ for any evaluation point vector $g$ (Proposition~\ref{prop:GTwGabidulinSeqh0}),
\item $h=k-1$ and arbitrary $m-n+1\leq t < m-k$ for any $g$ (Corollary~\ref{cor:GTwGabidulinSeqhk-1}),
\item and almost all $t,h$ for any $g$ spanning a subfield of $\Fqm$ (Proposition~\ref{prop:GTwGabidulinSeqSubfieldGenth}).
\end{itemize}
This enables us to give several classes of generalized twisted Gabidulin codes that are inequivalent to any Gabidulin or (narrow-sense) twisted Gabidulin code (cf.~Corollaries~\ref{cor:twisted_gab_inequivalence_h0_1}, \ref{cor:twisted_gab_inequivalence_h0_2}, %
 \ref{cor:twisted_gab_inequivalence_generalth}).
The techniques can be carried over to more classes of twisted Gabidulin codes.
In order to demonstrate the suitability of the approach without becoming too technical, we concentrate on the cases mentioned above.
 For the same reason, we refrain from determining the intersection sequences in this subsection.

The following result is again analogous to Lemmas \ref{lem:Gabtheta=theta-1} and \ref{lem:TGabtheta=theta-1}.

\begin{lemma}\label{lem:GTGabtheta=theta-1}
Let $\theta$ be a generator of the Galois group $\Gal(\Fm/\Fq)$, $g \in \Fm^n$ with $\rk_q(g)=n$ and $\eta \in \Fm^*$. Then $\GTw{k,\theta}{\eta,t,h}{g}= \GTw{k,\theta^{-1}}{\eta,m-(k+t-1),k-h-1}{\theta^{k-1}(g)}$.
\end{lemma}
\begin{proof}
We have the following chain of equalities
\begin{align*}
& \GTw{k,\theta^{-1}}{\eta,m-(k+t-1),k-h-1}{\theta^{k-1}(g)} \\
&= \left\langle\!\left\{ \theta^{-i}(\theta^{k-1}(g)) \!\right\}_{i\in \{0,\dots,k-1\}\backslash \{k-h-1\}} \!\!\!\!\cup \{   \eta\theta^{-(k-1+m-(k+t-1))}(\theta^{k-1}(g))  +  \theta^{-(k-h-1)}(\theta^{k-1}\!(g))\}\!\right\rangle \\
& = \left\langle\left\{ \theta^{i}(g) \right\}_{i\in \{0,\dots,k-1\}\backslash \{h\}} \cup \{ \eta\theta^{k+t-1}(g) + \theta^{h}(g) \}\right\rangle \\
&=  \GTw{k,\theta}{\eta,t,h}{g}. \qedhere
\end{align*}
\end{proof}

Note that for $h=0, t=1$ this gives a different equality from the one in Lemma \ref{lem:TGabtheta=theta-1}. Both results together give the following identities.

\begin{corollary}
$$ \GTw{k,\theta}{\eta,1,0}{g}=  \GTw{k,\theta^{-1}}{\eta,m-k,k-1}{\theta^{k-1}(g)} =  \GTw{k,\theta^{-1}}{\eta^{-1},1,0}{\theta^k(g)}=  \GTw{k,\theta}{\eta^{-1},m-k,k-1}{\theta^{2k-1}(g)}$$
\end{corollary}

We now determine the first elements of the $\theta^r$-sum sequence for generalized twisted Gabidulin codes for the case $h=0$ and $1<t\leq n-k$. 

\begin{proposition}\label{prop:GTwGabidulinSeqh0}
Let $g\in \Fm^n$ with $\rk_q(g)=n$, and $C :=\GTw{k,\theta}{\eta,t,0}{g}$ be a generalized $\theta$-twisted Gabidulin code with $h=0$ and $1<t\leq n-k$.

\begin{enumerate}
    \item 
If $    1\leq r \leq k-1 $, then $s^{\theta^r}_i(C)\geq \min\{n,k+ir\}$ and if $ t+ir \leq n-k$ then 
$$s^{\theta^r}_i(C)=\min\!\left\{ k+ir+\min\!\left(i,\left\lceil \tfrac{t}{r}\right\rceil\right) ,n\right\}. $$
  If $    m-k+1\leq r \leq m-1 $, then $s^{\theta^r}_i(C)\geq \min\{n,k+i(m-r)\}$ and if $ t+i(m-r) \leq n-k$ then 
  $$s^{\theta^r}_i(C)=\min\!\left\{ k+i(m-r)+\min\!\left(i,\left\lceil \tfrac{t}{m-r}\right\rceil\right) ,n\right\}.$$
\item  If $k\leq r\leq n-k-1$ or $m-n+k+1\leq r \leq m-k$ we have  $s_1^{\theta^r}(C)=2k$.
\item  If $r\geq k$ and $r \geq n-k$, then  $s_1^{\theta^r}(C)\geq k+n-r-1$.   If $r\leq m-k$ and $r \leq m-n+k$, then $s_1^{\theta^r}(C)\geq k+n-m+r-1$.
\end{enumerate}
\end{proposition}

\begin{proof}
\begin{enumerate}
\item If $1\leq r \leq k-1$ we have 
\begin{align*}
 S_i^{\theta^r}(C) = \, &\langle  g+\eta\theta^{k+t-1}(g),\theta(g), \theta^2(g),\dots,\theta^{k-1}(g)\rangle \\
    &+\langle \theta^r(g)+\theta^r(\eta)\theta^{k+r+t-1}(g),\theta^{r+1}(g),\dots,\theta^{r+k-1}(g)\rangle + \dots\\
   &+\langle \theta^{ir}(g)+\theta^{ir}(\eta)\theta^{k+ir+t-1}(g),\theta^{ir+1}(g),\dots,\theta^{ir+k-1}(g)\rangle\\
= \, &\langle  g+\eta \theta^{k+t-1}(g), \theta(g), \theta^2(g),\dots,\theta^{k+ir-1}(g)\rangle + \\
&\langle  \theta^{k+t+r-1}(g),\theta^{k+t+2r-1}(g),\dots, \theta^{k+t+ir-1}(g)\rangle .
\end{align*}
The first summand contains $\min\{n,k+ir\}$ linearly independent elements and implies $s^{\theta^r}_i(C)\geq \min\{n,k+ir\}$. If $t+ir \leq n-k$, all involved powers of $\theta$ are smaller than $n$. In this case, the number of elements from the second summand that are not contained in the first summand is $\min\!\left(i,\left\lceil \tfrac{t}{r}\right\rceil\right)$.

If $m-k+1\leq r \leq m-1$, then analogously, with $h:=\theta^{ir}(g)$, 
\begin{align*}
 S_i^{\theta^r}(C) =
&\langle  \theta^{ir}(g)+\theta^{ir}(\eta)\theta^{k+ir+t-1}(g),\theta^{ir+1}(g),\dots,\theta^0(g),\dots, \theta^{k-1}(g)\rangle \\
&+\langle \theta^{k+t-1}(g),\dots, \theta^{k+t+(i-1)r-1}(g) \rangle ,\\
= \, &\langle h+\theta^{ir}(\eta) \theta^{k+t-1}(h), \theta(h), \theta^2(h),\dots,\theta^{k+i(m-r)-1}(h)\rangle + \\
&\langle  \theta^{k+t+(m-r)-1}(h),\theta^{k+t+2(m-r)-1}(h),\dots, \theta^{k+t+i(m-r)-1}(h)\rangle .
\end{align*}
which implies the second statement.

\item 
If $k\leq r \leq n-k-1$ we have 
$$ S_1^{\theta^r}(C) = \langle  \theta(g),\dots, \theta^{k-1}(g),\theta^{r+1}(g),\dots,\theta^{k+r-1}(g),g+\eta \theta^{k+t-1}(g), \theta^r(g)+\eta\theta^{k+t+r-1}(g)\rangle $$
which has dimension $2k$. 
Similarly we get only linearly independent elements if $m-n+k+1\leq r \leq m-k$.

\item 
 If $r \geq\max (k,n-k)$, then 
$$\cS^{\theta^r}_1(C)\supseteq\langle g+\eta \theta^{k+t-1}(g),\theta(g), \ldots, \theta^{k-1}(g),\theta^{r+1}(g),\ldots, \theta^{n-1}(g)\rangle,$$ which has dimension $k+n-r-1$, by Proposition  \ref{cor:rankMoore}.
Similarly the claim follows for $r \leq \min(m-k,m-n+k)$. \qedhere
\end{enumerate}
\end{proof}

\begin{corollary}\label{cor:twisted_gab_inequivalence_h0_1}
Let $1<k<n-1$, $g,h \in \Fm^n$ with $\rk_q(g)=\rk_q(h)=n$ and let $\theta$ be a generator of $\Gal(\Fm/\Fq)$. Moreover, let $t\in\{2,\dots,n-k\}$ and  let $1\leq r < m$ be such that $\gcd(r,m)=1$. 
\begin{enumerate}
\item Suppose that \emph{at least one} of the following holds:
{\setlist{leftmargin=45pt}
\begin{enumerate}
\item[\mylabel{a1}{\textnormal{(a1)}}] $1<r<  n-2$, or
\item[\mylabel{a2}{\textnormal{(a2)}}] $m-n+2< r <m-1$, or
\item[\mylabel{a3}{\textnormal{(a3)}}] $r\in\{1,m-1\}$ and $t<n-k$, or
\item[\mylabel{a4}{\textnormal{(a4)}}] $k<\min(2n-m-2,n-t)$, or
\item[\mylabel{a5}{\textnormal{(a5)}}] $m-n+2<k<n-t$.
\end{enumerate}
}
Then $\GTw{k,\theta}{\eta,t,0}{g}$ is not equivalent to $\G_{k,\theta^r}(h)$. 

\item Let $\eta,\bar\eta \in \Fm^*$ such that $\Norm(\eta)\neq(-1)^{km}$. Suppose that \emph{at least one} of the following holds:
{\setlist{leftmargin=45pt}
\begin{enumerate}
\item[\mylabel{b1}{\textnormal{(b1)}}] $1<r<  n-3$ and $k>2$, or
\item[\mylabel{b2}{\textnormal{(b2)}}] $m-n+3< r <m-1$ and $k>2$, or
\item[\mylabel{b3}{\textnormal{(b3)}}] $r\in\{1,m-1\}$ and $t<n-k-1$, or
\item[\mylabel{b4}{\textnormal{(b4)}}] $k=2, t<n-3$ and $r\in\{1,\dots,\lfloor(n-2)/2\rfloor\}\cup\{\lfloor m-(n-2)/2,\dots,m-1\rfloor\}$.
\end{enumerate}
}
Then $\GTw{k,\theta}{\eta,t,0}{g}$ is not equivalent to $\TGabShee{k,\theta^r}{\bar\eta}{h}$. 
\end{enumerate}
\end{corollary}

\begin{proof}
\begin{enumerate}
\item
We use again the fact that, by Proposition \ref{prop:GabidulinSeq} we have  $ s_1^{\theta^r}(\G_{k,\theta^r}(h)) = k+1$. Hence, if $s_1^{\theta^r}(\GTw{k,\theta}{\eta,t,0}{g}) =s_1^{\theta^r}(\GTw{k,\theta}{\eta,t,0}{g}) \neq k+1$ then the two codes are inequivalent by Lemma \ref{lem:invariant}.

If $1<r< k$, then by part 1 of Proposition \ref{prop:GTwGabidulinSeqh0} we have 
$$ s^{\theta^r}_1(\GTw{k,\theta^r}{\eta,t,0}{g})\geq \min\{n,k+r\}  >k+1.$$
If $k\leq r <n-k$, then by part 2 of Proposition \ref{prop:GTwGabidulinSeqh0} we have 
$$ s^{\theta^r}_1(\GTw{k,\theta^r}{\eta,t,0}{g})=2 k >k+1.$$
 If $\max(k,n-k)\leq r< n-2$, then by part 3 of Proposition \ref{prop:GTwGabidulinSeqh0} we have 
$$ s^{\theta^r}_1(\GTw{k,\theta^r}{\eta,t,0}{g})\geq k+n-r-1  >k+1.$$
Hence, if \ref{a1} is satisfied, then $\GTw{k,\theta}{\eta,t,0}{g}$ is not equivalent to $\G_{k,\theta^r}(h)$.  Analogously the statement follows if one between \ref{a2} and \ref{a3} holds. 

Observe that if one between \ref{a4} and \ref{a5} holds, then  $t<n-k$. If in addition \ref{a1}, \ref{a2} and \ref{a3} are not fulfilled, then $n-2\leq r\leq m-n+2$, which we assume now. This implies that $\max(k,n-k)\leq r \leq \min(m-k,m-n+k)$. Part 3 of Proposition \ref{prop:GTwGabidulinSeqh0} implies that
$$s_1^{\theta^r} (\GTw{k,\theta^r}{\eta,t,0}{g})\geq k+n-r -1\geq 2n-m -1$$ 
and also that
$$s_1^{\theta^r} (\GTw{k,\theta^r}{\eta,t,0}{g})\geq k+n-m+r -1\geq 2k+n-m-1 .$$ 
Thus, also in this case $s_1^{\theta^r}(\GTw{k,\theta^r}{\eta,t,0}{g})$ is strictly greater than $k+1$ and we conclude. 

\item
Here we use the fact that, by Proposition \ref{prop:TwGabidulinSeq} we have  $ s_1^{\theta^r}(\TGabShee{k,\theta^r}{\bar\eta}{h}) = k+2$ and $ s_2^{\theta^r}(\TGabShee{k,\theta^r}{\bar\eta}{h}) = k+3$. Hence, if $s_1^{\theta^r}(\GTw{k,\theta}{\eta,t,0}{g}) =s_1^{\theta^r}(\GTw{k,\theta}{\eta,t,0}{g}) \neq k+2$ or $s_2^{\theta^r}(\GTw{k,\theta}{\eta,t,0}{g}) =s_2^{\theta^r}(\GTw{k,\theta}{\eta,t,0}{g}) \neq k+3$ then the two codes are inequivalent by Lemma \ref{lem:invariant}.

If $1<r< k$, then by part 1 of Proposition \ref{prop:GTwGabidulinSeqh0} we have 
$$ s^{\theta^r}_2(\GTw{k,\theta^r}{\eta,t,0}{g})\geq \min\{n,k+2r\}  >k+3.$$
If $k\leq r <n-k$, then by part 2 of Proposition \ref{prop:GTwGabidulinSeqh0} we have 
$$ s^{\theta^r}_1(\GTw{k,\theta^r}{\eta,t,0}{g})=2 k $$
which is greater than $k+2$ if $k>2$. 
 If $\max(k,n-k)\leq r< n-3$, then by part 3 of Proposition \ref{prop:GTwGabidulinSeqh0} we have 
$$ s^{\theta^r}_1(\GTw{k,\theta^r}{\eta,t,0}{g})\geq k+n-r-1  >k+2.$$
Hence, if \ref{b1} is satisfied, then $\GTw{k,\theta}{\eta,t,0}{g}$ is not equivalent to $\TGabShee{k,\theta^r}{\bar\eta}{h}$.  Analogously the statement follows if one between \ref{b2} and \ref{b3} holds. 

If $k=2$ and $t<n-3$, we get
$$ s^{\theta^r}_2(\GTw{k,\theta^r}{\eta,t,0}{g}) =6 >k+3$$
with the given conditions in \ref{b4}. \qedhere
\end{enumerate}
\end{proof}

\begin{corollary}\label{cor:twisted_gab_inequivalence_h0_2}

\begin{enumerate}
\item
Let $1<k<n-t<n-1$.  If 
$m< 2n-4$
then generalized twisted Gabidulin codes of length $n$ and dimension $k$ over $\Fm$ with hook $h=0$ and twist $t>1$ are never equivalent to a Gabidulin code.
\item
Let $2<k<n-t-1<n-2$. If 
$m< 2n-6$
then generalized twisted Gabidulin codes of length $n$ and dimension $k$ over $\Fm$ with hook $h=0$ and twist $t>1$ are never equivalent to a twisted Gabidulin code.
\end{enumerate}
\end{corollary}
\begin{proof}
\begin{enumerate}
\item
For $m<2n-4$ we get that $n-2> m-n+2$, hence one of the first three conditions of part 1 of Corollary \ref{cor:twisted_gab_inequivalence_h0_1} is fulfilled.
\item
Similarly, for $m< 2n-6$ one of the first three conditions of part 2 of Corollary \ref{cor:twisted_gab_inequivalence_h0_1} is fulfilled. \qedhere
\end{enumerate} 
\end{proof}

Proposition \ref{prop:GTwGabidulinSeqh0} combined with Lemma \ref{lem:GTGabtheta=theta-1} gives the analogues of the previous results for generalized twisted Gabidulin codes with hook $h=k-1$. For simplicity we only state the analogues of Proposition \ref{prop:GTwGabidulinSeqh0} and Corollary \ref{cor:twisted_gab_inequivalence_h0_2}, but the same can easily be done for Corollary \ref{cor:twisted_gab_inequivalence_h0_1}.

\begin{corollary}\label{cor:GTwGabidulinSeqhk-1}
Let $g\in \Fm^n$ with $\rk_q(g)=n$, and $C :=\GTw{k,\theta}{\eta,t,k-1}{g}$ be a generalized $\theta$-twisted Gabidulin code with $h=k-1$ and $m-n+1 \leq t< m-k$.

\begin{enumerate}
   \item 
If $    1\leq r \leq k-1 $, then $s^{\theta^r}_i(C)\geq \min\{n,k+ir\}$ and if $ t+ir < m-k$ then 
$$s^{\theta^r}_i(C)=\min\!\left\{ k+ir+\min\!\left(i,\left\lceil \tfrac{m-(k+t-1)}{r}\right\rceil\right) ,n\right\}. $$
 If $    m-k+1\leq r \leq m-1 $, then $s^{\theta^r}_i(C)\geq \min\{n,k+i(m-r)\}$ and if $ t+i(m-r) < m-k$ then 
 $$s^{\theta^r}_i(C)=\min\!\left\{ k+i(m-r)+\min\!\left(i,\left\lceil \tfrac{m-(k+t-1)}{m-r}\right\rceil\right) ,n\right\}.$$
\item  If $k\leq r\leq n-k-1$ or $m-n+k+1\leq r \leq m-k$ we have  $s_1^{\theta^r}(C)=2k$.
\item  If $r\geq k$ and $r \geq n-k$, then  $s_1^{\theta^r}(C)\geq k+n-r-1$.   If $r\leq m-k$ and $r \leq m-n+k$, then $s_1^{\theta^r}(C)\geq k+n-m+r-1$.
\end{enumerate}
\end{corollary}

\begin{corollary}\label{cor:twisted_gab_inequivalence_hk-1_2}

\begin{enumerate}
\item
Let $1<k<n-t<2n-m$.  If 
$m< 2n-4$
then generalized twisted Gabidulin codes of length $n$ and dimension $k$ over $\Fm$ with hook $h=k-1$ and twist $t$ are never equivalent to a Gabidulin code.
\item
Let $2<k<n-t-1<2n-m-1$. If 
$m< 2n-6$
then generalized twisted Gabidulin codes of length $n$ and dimension $k$ over $\Fm$ with hook $h=k-1$ and twist $t$ are never equivalent to a twisted Gabidulin code.
\end{enumerate}
\end{corollary}

As a third class of generalized twisted Gabidulin codes, we study the family of codes in which the evaluation points $g_i$ span a subfield of $\Fqm$. This is a relevant case since it includes the codes of maximal length (for a given field extension $m$) for which the codes can be guaranteed to be MRD by the sufficient condition in \cite{puchinger2017further}. We exclude the smallest and largest two choices of both $t$ and $h$ for the sake of an easier proof, but the result might as well hold for some of these cases.

\begin{proposition}\label{prop:GTwGabidulinSeqSubfieldGenth}
Let $g\in \Fqn^n$ with $\rk_q(g)=n$, where $\Fqn \subseteq \Fqm$ is a subfield (i.e., the $g_i$ span a subfield of $\Fqm$), and $C := \GTw{k,\theta}{\eta,t,h}{g}$ be a generalized $\theta$-twisted Gabidulin code with $\eta \in \Fqm^*$, $2<t<n-k-1$ and $1<h<k-2$ (note that this implies $4<k<n-4$). Then, for $0 < r <m$ with $\gcd(r,m)=1$, we have
\begin{align*}
s_1^{\theta^r}(C) \geq k+3.
\end{align*}
\end{proposition}

\begin{proof}
Since the $g_i$ span a subfield $\Fqn$ of $\Fqm$, we have $\theta^{n+r}(g) = \theta^r(g)$ for all $r$. 
Due to this periodicity, it suffices to show the claim for $r=1,\dots,\lceil \tfrac{n-1}{2} \rceil$ since $\gcd(m,n)=1$ implies $r \not\equiv 0 \mod n$ (due to $n \mid m$), we can always write $\theta^r(g) = \theta^{r \mod n}(g)$, and for $(r \mod n) > n/2$, we can express the code with respect to $g' = \theta^{(-r \mod n)}(g)$ instead of $g$.
Recall also that by Proposition~\ref{cor:rankMoore}, $\theta^0(g),\dots,\theta^{n-1}(g)$ are linearly independent. The set  $S_1^{\theta^r}(C)$ contains
$$  \langle \theta^{j}(g)\rangle_{j\in\{0,\dots,k-1\}\backslash \{h\}} + \langle\theta^{j}(g)\rangle_{j\in\{\max(r,k),\dots,\min(n,k+r-1)\}\backslash \{h+r\}}$$
$$+\langle \eta \theta^{k-1+t}(g) + \theta^h(g),  \theta^{r}(\eta) \theta^{k-1+t+r}(g) + \theta^{h+r}(g) \rangle $$ 
which contains $k+3$ linearly independent elements, as we show now. Note that the first summand above always has dimension $k-1$, thus we need to find four extra linearly independent elements. 
We distinguish three cases:
\begin{itemize}
\item Case $r\geq 4$: 
Due to $4 \leq r \leq n/2$ and $n-k\geq 5$, we get:
\begin{itemize}
\item If $h+r\in \{k,\dots,n-1\}$, then there are three distinct elements, say $i_1,i_2,i_3$, in $\{k,\dots,n-1\} \setminus \{h+r\}$, for which $\theta^{i_j}(g) \in \theta^{r}(C) \subseteq S_1^{\theta^r}(C)$ for $j=1,2,3$.
The fourth linearly independent element is $\theta^{h+r}(g)+\theta^r(\eta)\theta^{k-1+t+r}(g)$.

\item If $h+r \notin \{k,\dots,n-1\}$, then there are four distinct integers $i_1,\dots,i_4 \in \{k,\dots,n-1\}$ with $\theta^{i_j}(g) \in \theta^r(C) \subseteq S_1^{\theta^r}(C)$.
\end{itemize}

\item Case $r=3$: For each $i \in \{0,\dots,k+2\}$, we have that $\theta^i(g) \in S_1^{\theta^r}(C)$ if $i \notin \{h,h+3\}$, $\theta^h(g)+\eta \theta^{k-1+t}(g) \in S_1^{\theta^r}(C)$, and $\theta^{h+r}(g)+\theta^{r}(\eta) \theta^{k-1+t+r}(g) \in S_1^{\theta^r}(C)$. Since $\theta^{k-1+t}(g) \neq \theta^{h+r}(g)$ and $\theta^{k-1+t+r}(g) \neq \theta^h(g)$, there are at least $k+3$ linearly independent codewords in $S_1^{\theta^r}(C)$.

\item Case $r \in \{1,2\}$: Due to $2 \leq h \leq k-3$, we have $r \leq h$ and $h+r \leq k-1$, hence the codeword $\theta^{i}(g)$ is contained in $C$ or $\theta^r(C)$ for any $i=0,\dots,k-1+r$. Furthermore, $k-1+r < k-1+t < k-1+t+r <n$, so the codewords $\eta \theta^{k-1+t}(g) + \theta^h(g)$ and $\theta^{r}(\eta) \theta^{k-1+t+r}(g) + \theta^{h+r}(g)$ are linearly independent from $g,\theta(g),\dots,\theta^{k-1+r}(g)$. Thus, there are at least $k+3$ linearly independent elements in $S_1^{\theta^r}(C)$. \qedhere
\end{itemize}
\end{proof}

This implies:

\begin{corollary}\label{cor:twisted_gab_inequivalence_generalth}
Let $g\in \Fqn^n$ with $\rk_q(g)=n$, where $\Fqn \subseteq \Fqm$ is a subfield (i.e., $n$ divides $m$ and the $g_i$'s span a subfield of $\Fqm$), and $C := \GTw{k,\theta}{\eta,t,h}{g}$ be a generalized $\theta$-twisted Gabidulin code with $\eta \in \Fqm^*$, $2<t<n-k-1$ and $1<h<k-2$.

Then, $C$ is not equivalent to a $\theta'$-Gabidulin code or a $\theta'$-twisted Gabidulin code in the narrow sense ($t=1$, $h=0$) for  any generator $\theta'$ of  $\Gal(\Fqm/\Fq)$.
\end{corollary}

\begin{proof}
Let $\theta' \in \Gal(\Fqm/\Fq)$ be a generator of the Galois group. Then, there is an $r\in \{1,\dots,m-1\}$ with $\theta' = \theta^r$, where $\gcd(r,m)=1$. By Proposition~\ref{prop:GTwGabidulinSeqSubfieldGenth}, we have $s_1^{\theta'}(C)  \geq k+3$.
On the other hand, a $\theta'$-Gabidulin code $C'$ fulfills $s_1^{\theta'}(C') = k+1$ and any (narrow-sense) $\theta'$-twisted Gabidulin code $C''$  gives $s_1^{\theta'}(C'') = k+2$.
Hence, the codes cannot be equivalent.
\end{proof}

\section{Number of Inequivalent Known Codes}\label{subsection:numberinequivalent}

In this section we use the results from the previous section to determine upper and lower bounds on the number of inequivalent Gabidulin and twisted Gabidulin codes. For generalized twisted Gabidulin codes, we compute lower bounds on the number of equivalence classes for some small parameters.

\subsection{Gabidulin Codes}

In the following, $\Gr(k,\Fm^n)$ denotes the $k$-dimensional Grassmannian of $\Fm^n$, that is the set of all the $k$-dimensional $\Fm$-subspaces of $\Fm^n$.
Moreover, let $\Gab_q(k,n,m,\theta)$  be the set of all $\Fmk$  $\theta$-Gabidulin codes, i.e.,
$$\Gab_q(k,n,m,\theta):= \left\{ \mathcal U \in \Gr(k,\Fm^n) \mid \mathcal U=\G_{k,\theta}(g) \mbox{ for some $g\in\Fm^n$ with }  \rk_q(g)=n \right\}.$$
Furthermore, define $F:=\{\theta \in  \Gal(\Fm/\Fq) \mid \theta \mbox{ is a generator of } \Gal(\Fm/\Fq) \}$, and
\begin{align*}
\Gab_q(k,n,m):= & \left\{ \mathcal U \in \Gr(k,\Fm^n) \mid \mathcal U \mbox{ is a $\theta$-Gabidulin code for some }  \theta \in F  \right\} \\
= & \bigcup_{\theta \in F} \Gab_q(k,n,m, \theta).
\end{align*}

For a fixed $\theta$,  Berger provided the following result about the equality of two Gabidulin codes.

\begin{theorem}\cite[Theorem~2]{be03}\label{thm:numGab} Let $u,v \in \Fm^n$ be two vectors such that $\rk_q(u)=\rk_q(v)=n$. Then, for any generator $\theta$ of $\Gal(\Fm/\Fq)$, 
$\G_{k,\theta}(u)=\G_{k,\theta}(v)$ if and only if $u=\lambda v$ for some $\lambda \in \Fm^*$.
\end{theorem}

This leads to the following results.

\begin{corollary}\label{cor:numGab}
 Let $k,n,m$  be integers such that $2\leq k \leq n-2$ and $n \leq m$, and let $\theta$ be a generator of $\Gal(\Fm/\Fq)$. Then,
 \begin{enumerate}
 \item
$$|\Gab_{q}(k,n,m,\theta)| = \prod_{i=1}^{n-1}(q^m-q^i),$$
and hence
\item
$$|\mathrm{Gab}_{q}(k,n,m)| \leq \frac{\phi(m)}{2} \prod_{i=1}^{n-1}(q^m-q^i).$$
\end{enumerate}
\end{corollary}

\begin{proof}
 \begin{enumerate}
 \item 
 This follows from Theorem \ref{thm:numGab}.
 \item
This directly follows from part 1 and Lemma \ref{lem:Gabtheta=theta-1}. \qedhere
\end{enumerate}
\end{proof}

However, one can expect that some of those Gabidulin codes will be equivalent to each other. Then the natural question is ``how many inequivalent Gabidulin codes are there?''. Recently, Schmidt and Zhou provided a lower bound on this number.

\begin{theorem}\cite[Theorem 1.2]{sc17}\label{thm:schmidtandzhou}
For any $k,n,m$ integers such that $1\leq k \leq n-1$ and $2\leq n \leq m-2$, the number of inequivalent Gabidulin codes of dimension $k$ in $\Fm^n$ is at least
$$\frac{1}{m}\qbinom{m}{n}_q\frac{q-1}{q^m-1}.$$
\end{theorem}

In order to determine the number of equivalent Gabidulin codes, for any such code $C$, we need to give an estimate  of the number of automorphisms $\theta$ for which $C$ is a $\theta$-Gabidulin code. Note that the  notion of $\theta$-Gabidulin codes requires $\theta$ to be a generator of the Galois group. 
Since $\Gal(\Fm/\Fq) \cong (\Z/m\Z)$ and the set of generators is isomorphic to $(\Z/m\Z)^*$, we introduce the following notation. For a code $C$ we define the set
$$A_{C}:=\left\{ \theta \mbox{ generator of } \Gal(\Fm/\Fq) \mid C \mbox{ is a } \theta\mbox{-Gabidulin code } \right\}.$$
If we fix a generator $\theta$ of $\Gal(\Fm/\Fq)$, the set $A_C$ corresponds to the set 
$$A_{C,\theta}:=\left\{ r \in (\Z/m\Z)^* \mid C \mbox{ is a } \theta^r\mbox{-Gabidulin code } \right\}.$$
We now state an auxiliary result.

\begin{lemma}\label{lem:numautomorphismGab}
 Let $3\leq k \leq n-3$ be integers and $\theta$ be a generator of $\Gal(\Fm/\Fq)$. Let moreover $C$ be a $\theta$-Gabidulin code.
\begin{enumerate}
\item $A_{C,\theta} \cap \{2,\ldots,n-2\}=\emptyset$.
\item If $a\in A_{C,\theta}$ for some  $a\in (\Z/m\Z)^*$, then   $A_{C,\theta}\cap \{a+3, a+4. \ldots, a+n-3\}=\emptyset$.
\item If $a\in A_{C,\theta}$ for some  $a\in (\Z/m\Z)^*$, then $|A_{C,\theta}\cap\{a+1,a+2,a+n-2\}|\leq 1$.
\item For any $b \in \Z/m\Z$, $|A_{C,\theta}\cap \{b,b+1,\ldots, b+n-2\}| \leq 2$.
\end{enumerate}
\end{lemma}

\begin{proof} Let $g \in \Fm^n$ with $\rk_q(g)=n$ be such that $C=\G_{k,\theta}(g)$. 
 \begin{enumerate}
\item Let $r \in \{2,\ldots,n-2\}$. By Proposition \ref{prop:GabidulinSeq}, $s_1^{\theta^r}(C)\geq k+2$, so $C$ cannot be a $\theta^r$-Gabidulin code.
\item Suppose $a\in A_{C,\theta}$. By definition and by Propostion \ref{prop:GabidulinSeq},  $s_1^{\theta^a}(C)=k+1$. Suppose that $s_1^{\theta^{a+i}}(C) =k+1$ for some $i \in \{3, 4, \ldots, n-3\}$. Then $\dim(C+\theta^{a}(C)+\theta^{a+i}(C))\leq k+2$. In particular, $s_1^{\theta^i}(\theta^a(C))\leq k+2$, which is not  possible for $i \in \{3, 4. \ldots, n-3\}$, by Proposition \ref{prop:GabidulinSeq}, since $\theta^a(C)$ is a $\theta$-Gabidulin code.
\item Suppose $a+1, a+n-2 \in A_{C,\theta}$. Then, $\dim(C+\theta^{a+1}(C)+\theta^{a+n-2}(C))\leq k+2$. This implies that $s_1^{\theta^{n-3}}(\theta^{a+1}(C))\leq k+2$, which is not possible  by part 4 of Proposition \ref{prop:GabidulinSeq}. The same argument shows that $a+2$ and $a+n-2$ cannot belong simultaneously to $A_{C,\theta}$. It remains to show that $a+1$ and $a+2$ cannot belong to $A_{C,\theta}$ at the same time. Suppose on the contrary that this holds. We have $s_1^{\theta^a}(C)=s_1^{\theta^{a+1}}(C)=s_1^{\theta^{a+2}}(C)$. Consider the space
$\mathcal U:=\langle g,\theta(g), \ldots, \theta^{k-1}(g),\theta^{a+2}(g),\ldots,\theta^{a+k-1}(g) \rangle$. Thus, $C\subseteq \mathcal U \subseteq \cS^{\theta^{a}}_1(C)\cap \cS^{\theta^{a+1}}_1(C) \cap \cS^{\theta^{a+2}}_1(C)$. Then,  $\dim(\mathcal U)$ can only be equal to $k$ or $k+1$.

 If $\dim(\mathcal U)=k+1$, then  $\cS^{\theta^{a}}_1(C)= \cS^{\theta^{a+1}}_1(C) = \cS^{\theta^{a+2}}_1(C)$, and $\cS^{\theta^{a}}_1(C)+ \cS^{\theta^{a+1}}_1(C)+\cS^{\theta^{a+2}}_1(C)=C+\cS_2^{\theta}(\theta^a(C))$ has dimension $k+1$, which is impossible, since $s_2^{\theta}(\theta^a(C))=k+2$. 

If $\dim(\mathcal U)=k$, then $\theta^{a+2}(g),\ldots,\theta^{a+k-1}(g) \in \langle g,\ldots,\theta^{k-1}(g)\rangle$. Hence, we can write $\theta^{a+2}(g)=\sum_{i=0}^{k-1}\lambda_i \theta^i(g)$, for some $\lambda_i \in \Fm$. Imposing  also $\theta^{a+2+i}(g) \in \langle g,\ldots,\theta^{k-1}(g)\rangle$, for every $i \in \{0,\ldots, k-3\}$, we get that necessarily $\theta^{a+2}(g) \in \langle g, \theta(g), \theta^2(g) \rangle$. However, $C$ is a $\theta^{a+2}$-Gabidulin code, so $s_1^{\theta^{a+2}}(C)=k+1$. This implies $\lambda_2=0$ and $\lambda_1 \neq 0$.
At the same time, we can write $\theta^{a}(g)=\mu_0\theta^{m-2}(g)+\mu_1\theta^{m-1}(g)$, where $\mu_i=\theta^{-2}(\lambda_i)$ for $i \in\{0,1\}$. By assumption, $C$ is $\theta^a$-Gabidulin, therefore $s_1^{\theta^a}(C)=k+1$. This implies that necessarily $\mu_0=0=\lambda_0$ and $\mu_1\neq 0$. Therefore, $\theta^{a+2}(g)=\lambda_1\theta(g)$. However, this contradicts the fact that $C$ is a $\theta^{a+1}$-Gabidulin code, since in this case we  have $s_1^{\theta^{a+1}}(C)=k$.

\item This follows from the previous parts. \qedhere
\end{enumerate}
\end{proof}

\begin{theorem}\label{thm:upperboundnumGab}
Let $k,n,m$ be integers with $2 <k< n-2 $  and $n\leq m$. Then %
$$|\Gab_q(k,n,m)| \geq \frac{\phi(m)}{\lfloor\frac{2m}{n-1}\rfloor}\prod_{i=1}^{n-1}(q^m-q^i).$$
\end{theorem}

\begin{proof}
First we provide an upper bound on the cardinality of the set $A_{C,\theta}$, where $C$ is a $\theta$-Gabidulin code, using a double counting argument. Consider the number
$$\sum_{b \in \Z/m\Z} |A_{C,\theta} \cap \{b,b+1,\ldots, b+n-2\}|.$$
On one hand, by Lemma \ref{lem:numautomorphismGab}, we have that it is upper-bounded by $2m$. On the other hand, every $r \in A_{C,\theta}$ is counted exactly $n-1$ times. Therefore we get
$$(n-1)|A_{C,\theta}|=\sum_{b \in \Z/m\Z} |A_{C,\theta} \cap \{b,b+1,\ldots, b+n-2\}|\leq 2m,$$
from which we deduce that $|A_{C,\theta}|\leq \lfloor \frac{2m}{n-1}\rfloor$. 
At this point, combining this upper bound with Corollary \ref{cor:numGab}, we get the desired lower bound, since every $\theta$-Gabidulin code $C$ is counted at most $|A_{C,\theta}|$ times.
\end{proof}

Moreover, for $1\leq k <n \leq m$ we define  $N_q(k,n,m)$ as the number of inequivalent $\Fmk$ Gabidulin codes, i.e.
$$N_q(k,n,m):=\left| \Gab_q(k,n,m)/\sim \right|.$$

\begin{theorem}\label{thm:numberinequivalentGab}
Let $\Fq$ be a finite field of characteristic $p$ and $2< k < n-2$ be integers. 
\begin{enumerate}
\item \label{part:1numGab} If $m=n$, then
$$N_q(k,m,m)= \frac{\phi(m)}{2}.$$
\item If $m>n$, then 
$$\frac{(n-1)\phi(m)}{2m^2 [\Fq:\F_p]}\prod_{i=2}^n\frac{q^{m-i+1}-1}{q^i-1} \leq N_q(k,m,n)  \leq \frac{\phi(m)}{2}\prod_{i=2}^n\frac{q^{m-i+1}-1}{q^i-1}.$$
\end{enumerate}
\end{theorem}

\begin{proof}
\begin{enumerate}
\item Let $\theta_1, \theta_2$ be two generators of $\Gal(\Fm/\Fq)$, and $g,h\in \Fm^m$ be two vectors such that $\rk_q(g)=\rk_q(h)=m$. Consider two Gabidulin codes $C=\G_{k, \theta_1}(g)$ and $C'=\G_{k,\theta_2}(h)$. First we show that if $\theta_1=\theta_2$ then $C\sim C'$. Since $n=m$ then $\rsu_q(g)=\rsu_q(h)$. Therefore, there exists $A \in \GL_m(q)$ such that $g A = h$. This implies that $C\cdot A=C'$ and therefore, $C\sim C'$. 
Now, suppose that $\theta_1=\theta_2^{-1}$. Then, by the previous argument, we can assume $g = h$. Since by Lemma \ref{lem:Gabtheta=theta-1}, $\G_{k,\theta_2^{-1}}(h) = \G_{k, \theta_2}(\theta_2^{1-k}(h))$, we obtain again $C\sim C'$. 
Finally, if  $\theta_1 \notin \{\theta_2, \theta_2^{-1}\}$, then $\theta_1=\theta_2^r$, with $r \notin \{1,-1\}$, and by  Proposition \ref{prop:GabidulinSeq}, we have $s^{\theta_1}_1(C')=k+1$ and $s^{\theta_2}_1(C)\geq k+2$. By Lemma \ref{lem:invariant} we deduce that they cannot be equivalent. Since there are exactly $\phi(m)$ generators for $\Gal(\Fm/\Fq)$, we conclude.

\item By part 2 of Corollary \ref{cor:numGab}, we have that there are at most  $\frac{\phi(m)}{2} \prod_{i=1}^{n-1}(q^m-q^i)$ many Gabidulin codes. Moreover, consider the action
\begin{equation}\label{eq:action}\begin{array}{rcl}
\GL_n(q) \times \Gab_q(k,n,m) &\longrightarrow &\Gab_q(k,n,m) \\
(A,\G_{k,\theta}(g) ) & \longmapsto & \G_{k,\theta}(g) \cdot A :=\G_{k,\theta}(gA).
\end{array}
\end{equation}
We have that, by Theorem \ref{thm:numGab}, $\G_{k,\theta}(g) \cdot A=\G_{k,\theta}(g) $, if and only if $gA=\lambda g$, for some $\lambda\in \Fm^*$. Since $g=(g_1,\ldots, ,g_n)$ is such that the $g_i$'s are $\Fq$-linearly independent, it easily follows that
$gA=\lambda g$ if and only if $A=\lambda I_n$. Therefore, the action defined in \eqref{eq:action} induces a free action  of $\GL_n(q)/\Fq^*$, with the same orbits. Since this action is free, and every orbit is contained in an equivalence class, we have that
$$N_q(k,m,n) \leq |\Gab_q(k,n,m)| \frac{|\Fq^*|}{|\GL_n(q)|}=\frac{\phi(m)}{2}\prod_{i=1}^{n-1}\frac{q^{m-i}-1}{q^{i+1}-1}.$$

We now prove the lower bound. By Theorem \ref{thm:upperboundnumGab}, we have at least $\frac{\phi(m)}{\lfloor\frac{2m}{n-1}\rfloor}\prod_{i=1}^{n-1}(q^m-q^i)$ many distinct Gabidulin codes. Considering again the action
in \eqref{eq:action}, we get that the number of orbits under that action, is at least 
$$\frac{\phi(m)}{\lfloor\frac{2m}{n-1}\rfloor}\prod_{i=1}^{n-1}\frac{q^{m-i}-1}{q^{i+1}-1}.$$
If we now consider the equivalence classes of Gabidulin codes, it remains to study the action of the subgroup $\Aut(\Fm)$, which has cardinality exactly $m[\Fq:\F_p]$. Therefore, an equivalence class  can be union of at most $m[\Fq:\F_p]$ orbits of the action \eqref{eq:action}, which leads to the desired result. \qedhere
\end{enumerate}
\end{proof}

The first part of  Theorem \ref{thm:numberinequivalentGab} provides the exact number of inequivalent Gabidulin codes in the case $n=m$. Moreover, for the  case $n<m$, the second part of Theorem  \ref{thm:numberinequivalentGab}  gives both an upper and a lower bound on this number. It is important to observe that, whenever $(n-1)\phi(m)>2m[\Fq:\F_p]$, the lower bound improves the one given in Theorem \ref{thm:schmidtandzhou}, due to Schmidt and Zhou \cite{sc17}.

\subsection{Twisted Gabidulin Codes}
\label{ssec:num_ineq_twisted}

Let $F$ be the set of generators of  $\Gal(\Fm/\Fq) $. We denote by $\TGab_q(k,n,m,\theta)$  the set of all $\Fmk$ $\theta$-twisted Gabidulin codes, and by $\TGab_q(k,n,m)$ the set of all  $\Fmk$ twisted Gabidulin codes  i.e.
\begin{align*}
\TGab_q(k,n,m,\theta):= \big\{ \mathcal U \in \Gr(k,\Fm^n) \mid & \; \mathcal U=\TGabShee{k,\theta}{\eta}{g} \mbox{ for some $\eta \in \Fm^*$ and $g \in \Fm^n$} \\
& \mbox{ with  } \Norm(\eta)\neq (-1)^{km} \mbox{ and } \rk_q(g)=n \big\}, 
\end{align*}
\begin{align*}
\qquad \quad \!\! \TGab_q(k,n,m):= & \left\{ \mathcal U \in \Gr(k,\Fm^n) \mid \mathcal U \mbox{ is a $\theta$-twisted Gabidulin code for some }  \theta \in F  \right\} \\
= & \bigcup_{\theta \in F} \TGab_q(k,n,m, \theta).
\end{align*}

As for Gabidulin codes,  the dual of a $\theta$-twisted Gabidulin code with $b=0$ is another $\theta$-twisted Gabidulin code. For the case $n=m$ this duality result based on the Delsarte bilinear form was shown in \cite{sheekey2016new,lunardon2018generalized}. One can show a similar duality result for the general case $n \leq m$, whose proof can be also found in \cite[Chapter 7]{neri2019PhD}.

\begin{theorem}\cite[Theorem 7.21]{neri2019PhD}\label{thm:dualTwist}
Let $k,n,m$ be positive integers such that $2 \leq k \leq n-2$ and $n\leq m$. Let $g \in \Fm^n$ with $\rk_q(g)=n$, $\eta \in \Fm^*$ with $\Norm(\eta)\neq (-1)^{km}$ and consider the $\theta$-twisted Gabidulin code $C:=\TGabShee{k,\theta}{\eta}{g}$. Then, for any non-zero $g' \in \G_{n-1,\theta}(\theta^{-(n-k-1)}(g))^\perp$, we have 
$C^\perp=\TGabShee{n-k,\theta}{\eta'}{g'},$
where $$\eta'=(-1)^n\eta \frac{\theta^{k-n+1}(D)}{\theta^{k-n}(D)}\frac{\theta^{k-n}(\langle \theta^{n-k}(g'); g \rangle)}{\langle \theta^{n-k}(g'); g \rangle},$$
and  $D:=\det(M_{n,\theta}(g))$. 
Moreover $\Norm(\eta')=(-1)^{nm}\Norm(\eta)$.
\end{theorem}

As for $\theta$-Gabidulin codes, one can  find the exact number of $\theta$-twisted Gabidulin codes  for a  given generator  $\theta$  of $\Gal(\Fm/\Fq)$, as we show in the following.

\begin{theorem}\label{thm:numTGab} Let $u,v \in \Fm^n$ be two vectors such that $\rk_q(u)=\rk_q(v)=n$. Then, for any generator $\theta$ of $\Gal(\Fm/\Fq)$ and $\eta, \eta' \in \Fm^*$,
$\TGabShee{k,\theta}{\eta}{u}=\TGabShee{k,\theta}{\eta'}{v}$ if and only if there exists $\lambda \in \Fm^*$ such that $u=\lambda v$ and $\eta'=\eta\frac{\theta^k(\lambda)}{\lambda}$.
\end{theorem}

\begin{proof}
We divide the proof in three cases.

\underline{Case $3\leq k \leq \frac{n}{2}$}:
Suppose that $C:=\TGabShee{k,\theta}{\eta}{u}=\TGabShee{k,\theta}{\eta'}{v}$. Then $\theta(u)$ can be written as 
$$\theta(u) =\sum_{i=1}^{k-1}\lambda_i\theta^i(v)+ \lambda_k(v+\eta' \theta^k(v)),$$
for some $\lambda_i \in \Fm$ not all zeros.
Let $r:=\max\{i \in [k] \mid \lambda_i \neq 0\}$.
 If $r=k$,  then we would have
$\theta(\theta(u))=\theta(\lambda_k(v+\eta' \theta^k(v)))+\theta(\sum_{i=1}^{k-1}\lambda_i \theta^i(v)) \in C$, but this is not possible, since we would have  $C=\langle v+\eta' \theta^k(v), \theta(v), \ldots, \theta^{k-1}(v),\theta^{k+1}(v)+\mu\theta^k(v)\rangle$, for some $\mu \in \Fm$, which has dimension $k+1$ by Proposition \ref{cor:rankMoore}.
Then $r<k$. Suppose that $r>1$, then $0<k-r<k$ and $\theta^{k-r}(\theta(u)) \in C$, and we obtain
$$\theta^{k-r}(\theta(u))=\sum_{i=1}^r \theta^{k-r}(\lambda_i)\theta^{k-r+i}(v)\in C.$$
Also in this case, we obtain $C=\langle v, \theta(v), \ldots, \theta^{k-1}(v),\theta^k(v)\rangle$, which has dimension $k+1$ by Proposition \ref{cor:rankMoore}. Hence the only possibility is $r=1$, i.e. $\theta(u)=\lambda_1\theta(v)$, or equivalently, $u=\lambda v$ for some $\lambda \in \Fm^*$. It remains to study the conditions on $\eta$ and $\eta'$. At this point we have 
$C=\langle v+\eta'\theta^k(v), \theta(v),\ldots, \theta^{k-1}(v) \rangle =\langle \lambda v+\eta\theta^k(\lambda)\theta^k(v), \theta(v),\ldots, \theta^{k-1}(v) \rangle$. Therefore, $v+\eta'\theta^k(v)= \mu \lambda v+\mu\eta\theta^k(\lambda)\theta^k(v)+\sum_{i=1}^{k-1}\mu_i\theta^i(v)$, for some $\mu, \mu_i \in \Fm$, which we can rewrite as
$$(1-\mu\lambda)v-\sum_{i=1}^{k-1}\mu_i \theta^i(v)+(\eta'-\mu\eta \theta^k(\lambda))\theta^k(v)=0.$$
By Proposition \ref{cor:rankMoore}, $v,\theta(v),\ldots,\theta^k(v)$ are linearly independent (since obviously $k<n$), hence 
$\mu_i=0$ for every $i\in [k-1]$, $\mu=\lambda^{-1}$ and $\eta'=\frac{\theta^k(\lambda)}{\lambda}\eta$.

\underline{Case $k=2$}: Suppose $C=\langle u+\eta\theta^2(u), \theta(u)\rangle = \langle v+\eta'\theta^2(v), \theta(v) \rangle$. Then $\theta(u)$ can be written as linear combination of $v+\eta'\theta^2(v), \theta(v)$, i.e. $\theta(u)=\lambda_1\theta(v)+\lambda_2 v+\lambda_2\eta'\theta^2(v)$. Then one can write 
\begin{align*}
u+\eta\theta^2(u)= &\theta^{-1}(\lambda_2)\theta^{-1}(v)+\theta^{-1}(\lambda_1)v+(\theta(\lambda_2)+\theta^{-1}(\lambda_2)\theta^{-1}(\eta'))\theta(v)+\\
&\theta(\lambda_1)\theta^2(v)+\theta(\lambda_2)\theta(\eta')\theta^3(v).
\end{align*}
By Proposition \ref{cor:rankMoore} we deduce that $\lambda_2=0$, and therefore, $\theta(u)=\lambda_1\theta(v)$, or equivalently, $u=\lambda v$ for some $\lambda \in \Fm^*$. The relation between $\eta$ and $\eta'$ is derived in the same way as done in the proof of  the case $3\leq k \leq \frac{n}{2}$.

\underline{Case $k>\frac{n}{2}$}: It follows by the duality result in Theorem \ref{thm:dualTwist}, and the cases $k=2$ and $3\leq k \leq \frac{n}{2}$.
\end{proof}

\begin{corollary}\label{cor:numTGab}
 Let $k,n,m$  be integers such that $2\leq k \leq n-2$ and $n \leq m$, and let $\theta$ be a generator of $\Gal(\Fm/\Fq)$. Then,
$$|\TGab_{q}(k,n,m,\theta)| =\left(1-\frac{1}{q-1}\right) \prod_{i=0}^{n-1}(q^m-q^i),$$
\end{corollary}

\begin{proof}
We have exactly $\prod_{i=0}^{n-1}(q^m-q^i)$ many choices for the vector $g$ and $(q^m-1) -\frac{q^m-1}{q-1}$ choices for the element $\eta$ with norm different from $(-1)^{km}$.  By Theorem \ref{thm:numTGab}, the total number has to be divided by the number of non-zero multiples of $g$, which is $q^m-1$.
\end{proof}

\begin{corollary}\label{cor:numtotTGab}
 Let $k,n,m$  be integers such that $2\leq k \leq n-2$ and $n \leq m$. Then,
$$|\TGab_{q}(k,n,m)| \leq \frac{\phi(m)}{2}\left( 1-\frac{1}{q-1}\right)\prod_{i=0}^{n-1}(q^m-q^i).$$
\end{corollary}

\begin{proof}
It directly follows from Corollary \ref{cor:numTGab} and Lemma \ref{lem:TGabtheta=theta-1}.
\end{proof}

For a prime power $q$, and two integers $k,m$ we consider the set $X_q(m,k):=\{\alpha \in \Fm \mid \Norm(\alpha)\neq (-1)^{km} \}$, and the left group action 

\begin{equation}\label{eq:actiontwist}\begin{array}{rcl}
\Aut(\Fm) \times X_q(m,k) & \longrightarrow & X_q(m,k) \\
(\tau , \alpha) & \longmapsto & \tau(\alpha).
\end{array}
\end{equation}
Observe that the one above is well-defined. Indeed, if $\alpha \in X_q(m,k)$, and $\tau \in \Aut(\Fm)$, then
$$\Norm(\tau(\alpha))=\prod_{i=0}^{m-1}\theta^i(\tau(\alpha))=\prod_{i=0}^{m-1}\tau(\theta^i(\alpha))=\tau(\Norm(\alpha)).$$
Since $(-1)^{km}$ belongs to the prime field, and therefore is fixed by any automorphism $\tau \in \Aut(\Fm)$, we have that also $\tau(\Norm(\alpha))\neq (-1)^{km}$.

\begin{theorem}\label{thm:numberinequivalentTGab}\cite[Theorem 2]{ne19invariants}
Let $\Fq$ be a finite field of characteristic $p$ and  $2< k < n-2$ be integers. Denote by $\mathcal X_q(m,k)$ the cardinality of the set of orbits of this group action. 
 If $m=n$, then  the number of inequivalent twisted Gabidulin codes is 
$$\mathcal X_q(m,k)\frac{\phi(m)}{2}.$$
\end{theorem}

\begin{proof}
  The proof is similar to the one of part 1 of Theorem \ref{thm:numberinequivalentGab}.  Let $\theta_1, \theta_2$ be two generators of $\Gal(\Fm/\Fq)$, and $g,h\in \Fm^m$ be two vectors such that $\rk_q(g)=\rk_q(h)=n$, and let $\eta, \eta' \in \Fm$ with norm not equal to $(-1)^{km}$. Consider two twisted Gabidulin codes $C=\TGabShee{k,\theta_1}{\eta}{g}$ and $C'=\TGabShee{k,\theta_2}{\eta'}{h}$.
Suppose  $\theta_1\notin\{ \theta_2, \theta_2^{-1}\}$, 
then $\theta_1=\theta_2^r$, with $r \notin\{ 1,-1\}$, and by Proposition \ref{prop:TwGabidulinSeq}, we have $s^{\theta_1}_1(C')=k+2$ and $s^{\theta_1}_1(C)\geq k+3$. By Lemma \ref{lem:invariant} we conclude that they cannot be equivalent.
Now, recall that $C \sim C'$  if and only if there exist $\tau \in \Aut(\Fm)$ and $A \in \GL_n(q)$ such that $C'=\tau(C)A$. When $C=\TGabShee{k,\theta_2}{\eta}{g}$ we get $\tau(C)A= \TGabShee{k, \theta_2}{\tau(\eta)}{\tau(g)}$.
Assume $\eta'=\tau(\eta)$ for some $\tau \in \Aut(\Fm)$. If $\theta_2=\theta_1$ then, $\Fm^m=\rsu_q(\tau(g))=\rsu_q(h)$. This implies that there exists $A \in \GL_m(q)$ such that $\tau(g)A=h$ and  $\tau(C) A = C'$.
Hence, for every $\theta_2$ generator of $\Gal(\Fm/\Fq)$, and for every representative $\eta$ in an orbit of the action defined in \eqref{eq:actiontwist},  we have exactly one equivalent class of twisted Gabidulin codes. Moreover, observe that $\TGabShee{k, \theta_2^{-1}}{\eta}{g} = \TGabShee{k,\theta_2}{\eta^{-1}}{\theta_2^{-k}(h)}$, by Lemma \ref{lem:TGabtheta=theta-1}.
This shows that $C$ and $C'$ are equivalent if and only if $\theta_2=\theta_1$ and $\eta=\tau(\eta')$ for some $\tau \in \Aut(\Fm)$ or $\theta_1=\theta_2^{-1}$ and $\eta^{-1}=\tau(\eta')$ for some $\tau \in \Aut(\Fm)$. By counting, we get exactly $\frac{\phi(m)}{2}\mathcal X_q(m,k)$ inequivalent twisted Gabidulin codes.
\end{proof}

\subsection{Computational Results on the Number of Generalized Twisted Gabidulin Codes}
\label{ssec:computational_number_of_gTGab}

Due to the huge variation of parameters, studying the exact number of generalized $\theta$-twisted Gabidulin codes with the same techniques as in the previous two subsections would become extremely technical. To nevertheless give an idea on the suitability of the invariants for distinguishing these codes, we present computational results on the number of inequivalent generalized $\theta$-twisted Gabidulin codes with one twist ($\numTwists=1$) here.

We fix the following parameters
\begin{itemize}
	\item $n \geq 6$
	\item $m = 2n$
	\item $2 \leq k \leq n-2$
	\item $g \in \Fqn^{n}$ (chosen at random such that $\rk_q (g) = n$)
	\item $\eta \in  \Fqm \setminus \Fqn$ (chosen at random).
\end{itemize}
Note that $\Fqn$ is a subfield of $\Fqm$ and the choice of $m$ gives the smallest $m$ such that the sufficient MRD condition in \cite{puchinger2017further} can be applied. Also, $g$ and $\eta$ are chosen to satisfy this MRD condition.

For each such choice of fixed parameters, we consider the parameter set
\begin{align*}
\mathrm{param}(n,m) := \left\{ (\theta, t, h) \, : \, \theta \in \Gal(\Fqm/\Fq) \text{ generator}, \; 1\leq t \leq n-k, \; 0\leq h <k \right\} / \sim_\mathrm{p},
\end{align*}
where we say that two parameter triples are equivalent, $(\theta_1,t_1,h_1) \sim_\mathrm{p} (\theta_2,t_2,h_2)$, if $\theta_1 = \theta_2^{-1}$, $t_1 = n-k+1-t_2$, and $h_1 = k-1-h_2$. 
The choice of equivalence in parameters, $\sim_\mathrm{p}$, is due to the symmetry of generalized $\theta$-twisted Gabidulin codes described in Lemma~\ref{lem:GTGabtheta=theta-1}.

We consider the set of codes
\begin{align}
\gTGab_{q}(k,n,m,g,\eta) := \left\{ \GTw{k,\theta}{\eta,t,h}{g} \, : \, (\theta, t, h) \in \mathrm{param}(n,m) \right\} \label{eq:Gabidulin_codes}
\end{align}
and count the number of equivalence classes $|\gTGab_{q}(k,n,m,g,\eta)/\sim|$ among these codes.

Table~\ref{tab:lower_upper_bounds_equivalence_classes_gTGab} presents lower and upper bounds on $|\gTGab_{q}(k,n,m,g,\eta)/\sim|$, where each table entry contains three bounds $LB1(n,k)$, $LB2(n,k)$, and $UB(n,k)$ that are computed as follows:
\begin{itemize}
\item $LB1(n,k)$: Lower bound on $|\gTGab_{q}(k,n,m,g,\eta)/\sim|$ obtained by computing the sequences $\{s_i^\sigma(C)\}_{i=1}^{n-k}$ and $\{t_i^\sigma(C)\}_{i=1}^{k}$ (consecutive sums/intersections) for all codes $C$ in $\gTGab_{q}(k,n,m,g,\eta)$, where $\sigma$ ranges through all elements of the Galois group $\Gal(\Fqm/\Fq)$.
\item $LB2(n,k)$: Lower bound on $|\gTGab_{q}(k,n,m,g,\eta)/\sim|$ obtained by computing the dimensions of the sums $\sigma_1(C)+\sigma_2(C)+\sigma_3(C)$ and intersections $\sigma_1(C)\cap\sigma_2(C)\cap\sigma_3(C)$ for $100$ random choices of pairwise distinct automorphisms $\sigma_1,\sigma_2,\sigma_3 \in \Gal(\Fqm/\Fq)$.
\item $UB(n,k)$: Upper bound on $|\gTGab_{q}(k,n,m,g,\eta)/\sim|$ given by $|\mathrm{param}(n,m)|$.
\end{itemize}

\begin{table}
\caption{{Lower ($LB1$ and $LB2$) and Upper ($UB$) bounds on the number of equivalence classes $|\gTGab_{q}(k,n,m,g,\eta)/\sim|$ (cf.~\eqref{eq:Gabidulin_codes}) as described in Section~\ref{ssec:computational_number_of_gTGab}. Each cell is formatted as
\protect\mytableformat{LB1(n,k)}{LB2(n,k)}{UB(n,k)}.}}
\label{tab:lower_upper_bounds_equivalence_classes_gTGab}
\begin{center}
\resizebox{\textwidth}{!}{
\def\arraystretch{2}
\begin{tabular}{c||c|c|c|c|c|c|c|c|c|c|c|c|c|c|c|c|c|c|c|c|c|c}
	$n$ / $k$ & $2$ & $3$ & $4$ & $5$ & $6$ & $7$ & $8$ & $9$ & $10$ & $11$ & $12$ & $13$ & $14$ & $15$ & $16$ & $17$ & $18$ & $19$ & $20$ & $21$ & $22$ & $23$ \\
	\hline \hline 
	$6$ & \mytableformat{2}{7}{16} & \mytableformat{3}{8}{18} & \mytableformat{2}{7}{16}  \\
	\hline 
	$7$ & \mytableformat{4}{15}{30} & \mytableformat{9}{27}{36} & \mytableformat{9}{27}{36} & \mytableformat{4}{15}{30}  \\
	\hline 
	$8$ & \mytableformat{6}{18}{48} & \mytableformat{11}{27}{60} & \mytableformat{10}{28}{64} & \mytableformat{11}{27}{60} & \mytableformat{6}{18}{48}  \\
	\hline 
	$9$ & \mytableformat{6}{27}{42} & \mytableformat{15}{45}{54} & \mytableformat{27}{57}{60} & \mytableformat{27}{57}{60} & \mytableformat{15}{45}{54} & \mytableformat{6}{27}{42}  \\
	\hline 
	$10$ & \mytableformat{5}{26}{64} & \mytableformat{18}{40}{84} & \mytableformat{18}{46}{96} & \mytableformat{25}{47}{100} & \mytableformat{18}{46}{96} & \mytableformat{18}{40}{84} & \mytableformat{5}{26}{64}  \\
	\hline 
	$11$ & \mytableformat{15}{45}{90} & \mytableformat{40}{105}{120} & \mytableformat{60}{135}{140} & \mytableformat{70}{145}{150} & \mytableformat{70}{145}{150} & \mytableformat{60}{135}{140} & \mytableformat{40}{105}{120} & \mytableformat{15}{45}{90}  \\
	\hline 
	$12$ & \mytableformat{12}{34}{80} & \mytableformat{23}{51}{108} & \mytableformat{26}{62}{128} & \mytableformat{36}{68}{140} & \mytableformat{32}{70}{144} & \mytableformat{36}{68}{140} & \mytableformat{26}{62}{128} & \mytableformat{23}{51}{108} & \mytableformat{12}{33}{80}  \\
	\hline 
	$13$ & \mytableformat{20}{66}{132} & \mytableformat{66}{162}{180} & \mytableformat{102}{210}{216} & \mytableformat{114}{234}{240} & \mytableformat{120}{246}{252} & \mytableformat{120}{246}{252} & \mytableformat{114}{234}{240} & \mytableformat{102}{210}{216} & \mytableformat{66}{162}{180} & \mytableformat{20}{66}{132}  \\
	\hline 
	$14$ & \mytableformat{13}{57}{144} & \mytableformat{45}{96}{198} & \mytableformat{51}{117}{240} & \mytableformat{66}{132}{270} & \mytableformat{63}{141}{288} & \mytableformat{75}{144}{294} & \mytableformat{63}{141}{288} & \mytableformat{66}{132}{270} & \mytableformat{51}{117}{240} & \mytableformat{42}{96}{198} & \mytableformat{13}{47}{144}  \\
	\hline 
	$15$ & \mytableformat{16}{73}{104} & \mytableformat{52}{132}{144} & \mytableformat{76}{172}{176} & \mytableformat{92}{196}{200} & \mytableformat{100}{212}{216} & \mytableformat{108}{220}{224} & \mytableformat{108}{220}{224} & \mytableformat{100}{212}{216} & \mytableformat{92}{196}{200} & \mytableformat{76}{172}{176} & \mytableformat{52}{132}{144} & \mytableformat{16}{76}{104}  \\
	\hline 
	$16$ & \mytableformat{23}{77}{224} & \mytableformat{74}{148}{312} & \mytableformat{80}{188}{384} & \mytableformat{108}{216}{440} & \mytableformat{100}{236}{480} & \mytableformat{128}{248}{504} & \mytableformat{120}{252}{512} & \mytableformat{128}{248}{504} & \mytableformat{100}{236}{480} & \mytableformat{108}{216}{440} & \mytableformat{80}{188}{384} & \mytableformat{74}{150}{312} & \mytableformat{23}{78}{224}  \\
	\hline 
	$17$ & \mytableformat{40}{108}{240} & \mytableformat{136}{312}{336} & \mytableformat{184}{408}{416} & \mytableformat{232}{472}{480} & \mytableformat{240}{520}{528} & \mytableformat{264}{552}{560} & \mytableformat{280}{568}{576} & \mytableformat{280}{568}{576} & \mytableformat{264}{552}{560} & \mytableformat{240}{520}{528} & \mytableformat{232}{472}{480} & \mytableformat{184}{408}{416} & \mytableformat{136}{312}{336} & \mytableformat{40}{102}{240}  \\
	\hline 
	$18$ & \mytableformat{25}{74}{192} & \mytableformat{60}{130}{270} & \mytableformat{78}{165}{336} & \mytableformat{99}{192}{390} & \mytableformat{102}{213}{432} & \mytableformat{117}{228}{462} & \mytableformat{111}{237}{480} & \mytableformat{123}{240}{486} & \mytableformat{111}{237}{480} & \mytableformat{117}{228}{462} & \mytableformat{102}{213}{432} & \mytableformat{99}{192}{390} & \mytableformat{78}{165}{336} & \mytableformat{60}{132}{270} & \mytableformat{25}{72}{192}  \\
	\hline 
	$19$ & \mytableformat{48}{110}{306} & \mytableformat{180}{400}{432} & \mytableformat{252}{531}{540} & \mytableformat{297}{621}{630} & \mytableformat{342}{693}{702} & \mytableformat{369}{747}{756} & \mytableformat{378}{783}{792} & \mytableformat{396}{801}{810} & \mytableformat{396}{801}{810} & \mytableformat{378}{783}{792} & \mytableformat{369}{747}{756} & \mytableformat{342}{693}{702} & \mytableformat{297}{621}{630} & \mytableformat{252}{531}{540} & \mytableformat{180}{405}{432} & \mytableformat{48}{131}{306}  \\
	\hline 
	$20$ & \mytableformat{35}{73}{288} & \mytableformat{90}{197}{408} & \mytableformat{108}{251}{512} & \mytableformat{148}{296}{600} & \mytableformat{148}{332}{672} & \mytableformat{180}{360}{728} & \mytableformat{176}{380}{768} & \mytableformat{200}{392}{792} & \mytableformat{192}{396}{800} & \mytableformat{200}{392}{792} & \mytableformat{176}{380}{768} & \mytableformat{180}{360}{728} & \mytableformat{148}{332}{672} & \mytableformat{148}{296}{600} & \mytableformat{108}{252}{512} & \mytableformat{90}{192}{408} & \mytableformat{35}{73}{288}  \\
	\hline 
	$21$ & \mytableformat{44}{137}{228} & \mytableformat{126}{297}{324} & \mytableformat{186}{402}{408} & \mytableformat{234}{474}{480} & \mytableformat{258}{534}{540} & \mytableformat{276}{582}{588} & \mytableformat{294}{618}{624} & \mytableformat{312}{642}{648} & \mytableformat{324}{654}{660} & \mytableformat{324}{654}{660} & \mytableformat{312}{642}{648} & \mytableformat{294}{618}{624} & \mytableformat{276}{581}{588} & \mytableformat{258}{534}{540} & \mytableformat{234}{474}{480} & \mytableformat{186}{402}{408} & \mytableformat{126}{302}{324} & \mytableformat{44}{109}{228}  \\
	\hline 
	$22$ & \mytableformat{41}{111}{400} & \mytableformat{135}{273}{570} & \mytableformat{165}{355}{720} & \mytableformat{210}{419}{850} & \mytableformat{215}{475}{960} & \mytableformat{260}{520}{1050} & \mytableformat{250}{555}{1120} & \mytableformat{290}{580}{1170} & \mytableformat{285}{595}{1200} & \mytableformat{305}{600}{1210} & \mytableformat{285}{595}{1200} & \mytableformat{290}{580}{1170} & \mytableformat{250}{555}{1120} & \mytableformat{260}{520}{1050} & \mytableformat{215}{475}{960} & \mytableformat{210}{420}{850} & \mytableformat{165}{355}{720} & \mytableformat{135}{277}{570} & \mytableformat{41}{97}{400}  \\
	\hline 
	$23$ & \mytableformat{88}{161}{462} & \mytableformat{286}{596}{660} & \mytableformat{385}{825}{836} & \mytableformat{484}{977}{990} & \mytableformat{528}{1111}{1122} & \mytableformat{605}{1221}{1232} & \mytableformat{638}{1309}{1320} & \mytableformat{660}{1375}{1386} & \mytableformat{693}{1419}{1430} & \mytableformat{715}{1441}{1452} & \mytableformat{715}{1440}{1452} & \mytableformat{693}{1419}{1430} & \mytableformat{660}{1375}{1386} & \mytableformat{638}{1309}{1320} & \mytableformat{605}{1221}{1232} & \mytableformat{528}{1111}{1122} & \mytableformat{484}{976}{990} & \mytableformat{385}{823}{836} & \mytableformat{286}{611}{660} & \mytableformat{88}{164}{462}  \\
	\hline 
	$24$ & \mytableformat{45}{97}{352} & \mytableformat{122}{239}{504} & \mytableformat{140}{315}{640} & \mytableformat{192}{376}{760} & \mytableformat{200}{428}{864} & \mytableformat{240}{472}{952} & \mytableformat{248}{508}{1024} & \mytableformat{272}{536}{1080} & \mytableformat{260}{556}{1120} & \mytableformat{288}{568}{1144} & \mytableformat{280}{572}{1152} & \mytableformat{288}{568}{1144} & \mytableformat{260}{556}{1120} & \mytableformat{272}{535}{1080} & \mytableformat{248}{507}{1024} & \mytableformat{240}{472}{952} & \mytableformat{200}{428}{864} & \mytableformat{192}{376}{760} & \mytableformat{140}{314}{640} & \mytableformat{122}{242}{504} & \mytableformat{45}{89}{352}  \\
	\hline 
	$25$ & \mytableformat{100}{139}{460} & \mytableformat{290}{609}{660} & \mytableformat{390}{828}{840} & \mytableformat{490}{988}{1000} & \mytableformat{570}{1130}{1140} & \mytableformat{620}{1250}{1260} & \mytableformat{660}{1350}{1360} & \mytableformat{710}{1429}{1440} & \mytableformat{720}{1489}{1500} & \mytableformat{750}{1530}{1540} & \mytableformat{770}{1550}{1560} & \mytableformat{770}{1550}{1560} & \mytableformat{750}{1530}{1540} & \mytableformat{720}{1489}{1500} & \mytableformat{710}{1429}{1440} & \mytableformat{660}{1349}{1360} & \mytableformat{620}{1250}{1260} & \mytableformat{570}{1130}{1140} & \mytableformat{490}{987}{1000} & \mytableformat{390}{829}{840} & \mytableformat{290}{606}{660} & \mytableformat{100}{156}{460}  \\
	\hline 
\end{tabular}
}
\end{center}
\end{table}

It can be seen that for many parameters, we obtain lower and upper bounds that are quite close to each other. For instance, for $[n,k] = [11,5]$, there are at least $145$ and at most $150$ equivalence classes. The bounds are also almost tight for $[n,k] = [15,6]$ ($LB2 = 212$, $UB = 216$) and $[n,k] = [25,10]$ ($LB2 = 1489$, $UB = 1500$).

For even length, the best possible ratio (attained for $k = n/2$ in all cases) of $LB2$ and $UB$ is in all cases close to, but never above $0.5$.
This might be due to a further symmetry in the code parameters, besides $\sim_\mathrm{p}$, resulting in equivalent codes, but it needs to be investigated further.

Although the lower bounds on consecutive sums and intersections (UB1) are in general quite good (for odd lengths, they find roughly half of the possible equivalence classes), random sums and intersections (UB2) of codes under automorphisms appear to perform better in distinguishing generalized twisted Gabidulin codes.
Our observation is that consecutive sums of two twisted codes always have the same dimensions if the code parameters are related by certain symmetries that the random sums can often distinguish.
This should be further investigated in future work.

Another interesting observation is that $UB1$ is symmetric in $k$ for a given length (i.e., $UB1(n,k) = UB1(n,n-k)$ for all $n,k$).
When looking closer at the raw results, one can observe that the number of equivalence classes for $[n,k]$ from the consecutive sums ($s_i^\autom$) in all checked cases equals the number obtained for $[n,n-k]$ from the consecutive intersections ($t_i^\autom$).
This symmetry, together with $t_i^\autom(C)=n-s_i^\autom(C^\perp)$ (cf.~Proposition~\ref{prop:DualitySiTi}), shows that the bounds obtained from consecutive sums of the $[n,k]$ twisted codes and the duals of the $[n,n-k]$ codes (which are again $[n,k]$ codes) are the same.
This indicates, but certainly does not prove, that the duals of the tested twisted Gabidulin codes are equivalent to twisted Gabidulin codes.
To the best of our knowledge, nothing is known about the duals of twisted Gabidulin codes, except for Theorem~\ref{thm:dualTwist} (narrow sense twisted Gabidulin codes, i.e., $t=1$, $h=0$).

We computed the values of Table~\ref{tab:lower_upper_bounds_equivalence_classes_gTGab} using the computer algebra SageMath v8.1~\cite{sagemath}.
The computation of the entire table using a non-optimized implementation took less than the equivalent of $45$ days on a single thread of an Intel X5650 CPU (launch year 2010). Note, for comparison, that checking the equivalence of only two codes of length $25$ using the definition, Definition~\ref{isometries}, would involve brute-forcing through all full-rank matrices in $\mathbb{F}_3^{25 \times 25}$, which are more than $2^{989}$ many. This is far away from what is assumed to be computable on an ordinary computer today.

\section{Characterization Results for Gabidulin Codes}\label{sec:CharGabidulin}

In this section we study the $\sigma$-sequences, in order to derive characterization results of some families of codes. 
Unfortunately, for $m$ large enough almost all the rank-metric codes have the same   sequence $\{s_i^\autom(C)\}$, hence it seems unlikely that asymptotically we can get nice characterizations.
 This is explained by the following result, due to Coggia and Couvreur.

\begin{proposition}\cite[Proposition 2]{coggia2019security}
If $C$ is an $\Fmk$ code chosen at random and uniformly among all the possible $\Fmk$ codes, then for any non-negative integer $b$ and for a positive
integer $i< k$, we have
$$\mathrm{Pr}\big\{s_i^\theta(C)\leq \min\{n,(i+1)k\}-b\big\}=\mathcal O(q^{-mb}),$$
 for $m\rightarrow +\infty$.
\end{proposition}

However, this happens only when $m$ is big. In particular,  one can expect that codes which have no maximal dimension (usually) have good algebraic structures. Moreover, restricting to  MRD codes and the case $n=m$ has a different effect. An idea of this different behavior is explained by the following result due to Payne in 1971. The original result is formulated in a completely different way, since it was determined in the framework of hyperovals and linearized $o$-polynomials. See \cite{caullery2015classification} for more details.

\begin{theorem}\cite{payne1971complete}\label{thm:payne}
Let $C$ be an $[n,2]_{2^n}$ MRD code. Then there exists a generator $\theta$ of $\Gal(\Fm/\Fq)$ such that
$s_i^{\theta}(C)=2+i$ for every  $i \in \{0,\ldots, n-2\}$.
\end{theorem}

Reformulated, Theorem \ref{thm:payne} states that  all the $[n,2]_{2^n}$ MRD codes are Gabidulin codes.

Now we are going to use the sequences for characterizing Gabidulin codes. The following result follows from \cite[Lemma 3.5]{giuzzi2019identifiers}, but we are going to include a proof for completeness, which uses the tools developed in this work. Note that for MRD codes the same result was shown in \cite[Proposition 4.6]{ho16}.

\begin{lemma}\label{lem:k-1intersect}
Let $0<k<n\leq m$ be integers,  $C$ be an $\Fmk$ code, and $\theta$ be a generator of $\Gal(\Fm/\Fq)$. If $s_1^{\theta}(C)=k+1$, then there exists $g \in \Fm^n$, $0\leq t \leq k$ such that
$$C:=C_1 \oplus \G_{t,\theta}(g),$$
where $C_1$ is an $[n,k-t]_{q^m}$ code which has a basis  of rank 1 vectors, and $\rk_q(g)> t$.
\end{lemma}

\begin{proof}
 Let $U_1=\{v \in \Fm^n \mid \rk_q(v)=1\}$. Then we can write $C=C_1\oplus C'$, where
$C_1=\langle C \cap U_1 \rangle$. Hence $C'\cap U_1 =\emptyset$. In particular, if $\dim(C_1)=k-t$, then $C'$ is an $[n,t,d]_{q^m}$ code with $d>1$. Moreover, $C + \theta(C)=C_1 \oplus ( C' + \theta(C'))$. Therefore, we can assume $C_1=\{0\}$, and that the minimum distance of $C$ is greater than $1$, for the rest of the proof. 

We proceed by induction on $k$. For $k=1$ it is trivially true. Suppose now that $k\geq 2$ and that the statement is true for $k-1$; we want to prove the lemma for an $\Fmk$ code $C$. By hypothesis, $s_1^\theta(C)=k+1$, hence $t_1^\theta(C)=k-1$, by Proposition \ref{prop:SiTi}. Consider the code $D:=C\cap\theta(C)=\T_1^\theta(C)$. This code has dimension $k-1$. Moreover, $D\cap U_1=\emptyset$, and 
$$k-1\geq t_2^\theta(C)=t_1^\theta(C) -\Lambda_1^\theta(C) \geq t_1^\theta(C)-\Lambda_0^\theta(C)=k-2.$$
However, observe that $t_1^\theta(D)=t_2^\theta(C)$. Therefore,  if $t_2^\theta(C)=k-1$, then $t_1^\theta(D)=t_0^\theta(D)$, which implies, by part 3 of Proposition \ref{prop:propertiesTi}, that $D$ has a basis of elements in $\Fq^n$, or, equivalently, that $D=\langle D\cap  U_1 \rangle$. Since $k-1>0$, $D\subseteq C$ and $C\cap U_1=\emptyset$, which is a contradiction. Hence, we necessarily have that $t_1^\theta(D)=t_2^\theta(C)=k-2$. Thus, by inductive hypothesis, $D=\G_{k-1}(h)$, for some $h\in \Fm^n$ with $\rk_q(h)\geq k$. Moreover, $\theta^{-1}(h) \in \theta^{-1}(D)=\theta^{-1}(C)\cap C \subseteq C$. Therefore, 
$C\supseteq \langle \theta^{-1}(h),h,\ldots, \theta^{k-2}(h)\rangle$. Since $\rk_q(h)\geq k$, by Proposition \ref{cor:rankMoore} we get $C=\G_{k,\theta}(g)$, where $g:=\theta^{-1}(h)$. We only need to show that $\rk_q(g)>k$. Suppose $\rk_q(g)=k$, and let $\{f_1,\ldots, f_k\}$ be a basis for $\rsu_q(g)$. Then, there exists $A\in \GL_n(q)$ such that $gA=(f_1,\ldots, f_k, 0, \ldots, 0)$. Moreover, the code $C\cdot A=\G_{k,\theta}(gA) \sim C$ and has the same parameters. It is easy to see that this code has generator matrix $(I_k \mid 0)$, since the last $n-k$ entries of the code $C\cdot A$ are all zeros. This implies that $C\cdot A$ has (a basis of) codewords of rank $1$, which yields a contradiction. Hence, $\rk_q(g)>k$ and this concludes the proof.
\end{proof}

\begin{remark} 
Observe that in Lemma \ref{lem:k-1intersect}, the notation $\G_{k,\theta}(g)$ is used to indicate the code
$\{ (f(g_1),\ldots, f(g_n)) \mid f \in \G_{k,\theta}\}$,
which is not necessarily a $\theta$-Gabidulin code, since $\rk_q(g)$ can be smaller than $n$. Moreover, we have that
$C_1=\{0\}$ if and only if the minimum distance of $C$ is strictly greater than $1$.
\end{remark}

From Lemma \ref{lem:k-1intersect} we can derive a new criterion for characterizing a Gabidulin code. In order to put all the criteria together, we state a very general characterization theorem which includes also results from \cite{ho16} and \cite{ne18sys}.

\begin{theorem}[Characterization of $\theta$-Gabidulin codes]\label{thm:characterizationGab}
Let $C\subseteq \Fm^n$ be a linear code of dimension $k$ and let $\theta$ be a generator of $\Gal(\Fm/\Fq)$. The following are equivalent:
\begin{enumerate}
\item $C$ is a $\theta$-Gabidulin code of dimension $k$. 
\item $C^{\perp}$ is a $\theta$-Gabidulin code of dimension $n-k+1$. 
\item $C$ is MRD  and $s_1^\theta(C)=k+1$. %
\item $C$ is MRD  and $t_1^\theta(C)=k-1$.

\item $(s_i^\theta(C))_{i=0}^{n-k}=(k,k+1,\ldots,n)$ and $d(C)>1$.
\item $(t_i^\theta(C))_{i=0}^{k}=(k,k-1,\ldots,0)$ and $d(C^\perp)>1$. 

\item $s_1^\theta(C)=k+1$ and $s^\theta_{n-k}(C)=n$ and $d(C)>1$.
\item $t_1^\theta(C)=k-1$ and $t^\theta_{k}(C)=0$ and $d(C^\perp)>1$.

\item $(\Delta_i^\theta(C))_{i=0}^{n-k}=(1,1,\ldots, 1, 0)$ and $d(C)>1$.
\item $(\Lambda_i^\theta(C))_{i=0}^{k}=(1,1,\ldots, 1, 0)$ and $d(C^\perp)>1$.

\item  $\Delta_0^\theta(C)=\Delta_{n-k-1}^\theta(C)=1$ and $d(C)>1$. 
\item  $\Lambda_0^\theta(C)=\Lambda_{k-1}^\theta(C)=1$ and $d(C^\perp)>1$. 

\item $C=\rs(I_k\mid X)$, where: 
\begin{itemize}
\item[\emph{(a)}] $\rk(\theta(X)-X)=1$,
\item[\emph{(b)}]  the $q$-rank  of the first row of $\theta(X)-X$ is $n-k$,
\item[\emph{(c)}]  the  $q$-rank  of the first column  of  $\theta(X)-X$ is $k$.
\end{itemize}
\end{enumerate}
\end{theorem}

\begin{proof}\quad
\begin{enumerate}
\item[1.$\iff$2.] This is well known and can be found e.g. in \cite{ga85a}.
\item[1.$\iff$3.]
This was shown in \cite[Theorem 4.8]{ho16}.
\item[1.$\iff$5.]
That 1. implies 5. was shown in part 1 of Proposition \ref{prop:GabidulinSeq}.
It remains to show the other direction. 
Since $s_1^\theta(C)=1$, then, by Lemma \ref{lem:k-1intersect} and the fact that $d>1$, we have that $C=\G_{k,\theta}(g)$, where $\rk_q(g)>k$. Moreover, $\cS_{n-k}^\theta(C)=\G_{n,\theta}(g)$. By hypothesis, we also have that $n=s_{n-k}^\theta(C)=\dim(\G_{n,\theta}(g))$. This implies, by Proposition \ref{cor:rankMoore}, that $\dim(\G_{n,\theta}(g))=\rk_q(g)$. Therefore, $C$ is a $\theta$-Gabidulin code.

\item[5.$\iff$7.]
Moreover, if $s_1^\theta(C)=k+1$, and $s_{n-k}^\theta(C)=n$, then $\Delta_0^{\theta}(C)=s_1^\theta(C)-k=1$, and by part 7 of Proposition \ref{prop:propertiesSi}, we have 
$$n=s_{n-k}^\theta(C)=k+\sum_{i=0}^{n-k-1}\Delta_i^\theta(C) \leq k+\sum_{i=0}^{n-k-1}\Delta_0^\theta(C)=k+(n-k).$$
Therefore, $\Delta_i^\theta(C)=1$, for every $i=0,\ldots, n-k-1$.
\item[5.$\iff$9.]
The equivalence follows from the definition of $\Delta_i^\theta$ and part 7 of Proposition \ref{prop:propertiesSi}.
\item[9.$\iff$11.]
If $\Delta_{n-k-1}^\theta(C)=\Delta_0^\theta(C)=1$, by Proposition \ref{prop:propertiesSi}, we have $1=\Delta_0^\theta(C) \geq \ldots \geq \Delta_{n-k-1}^\theta(C)=1$, hence we have all equalities.

\item[4.,6.,8.,10.,12.]
It is easy to see that 3.$\iff$4., 5.$\iff$6.,7.$\iff$8.,9.$\iff$10. and 11.$\iff$12. by Proposition \ref{prop:SiTi}.

\item[1.$\iff$13.]
This was shown in \cite[Theorem 11]{ne18sys}.  \qedhere
\end{enumerate}
\end{proof}

\begin{remark}
As explained in Section \ref{sec:invariants}, we can efficiently compute the sum or intersection sequences of any given code. For the characterization result above, however, we need to check if the code (or its dual) has minimum distance greater than one. Although determining the exact minimum rank distance of a code is a computationally heavy task, checking if the minimum distance is one or greater can be done very efficiently by determining the $\Fq$-(subfield) subcode of the code. This method of finding rank one codewords is explained in detail in \cite{horlemann2018extension}.
\end{remark}

After this characterization result, we conclude the section by proving that Gabidulin's new codes of Definitions \ref{def:GabNewIcodes} and \ref{def:GabNewIIcodes} are actually the classical Gabidulin codes, whenever they are MRD. 

\begin{theorem}
Let $1\leq k \leq m$ be integers and  $\theta$ be a generator of $\Gal(\Fm/\Fq)$. Let, moreover,  $\eta \in \Fm$ with $\Norm(\eta) \neq (-1)^{km}$, and $g\in \Fm^n$ be such that $\rk_q(g)=n$.
\begin{enumerate}
\item If  $m-k > k$,  then the code $C:=\NewG{k,\theta}{\eta,I}(g)$ is a $\theta$-Gabidulin code.
\item If $m-k\leq k$, then the code $D:=\NewG{k,\theta}{\eta,II}(g)$ is a $\theta$-Gabidulin code.
\end{enumerate}
\end{theorem}

\begin{proof}
\begin{enumerate}
\item It is immediate to observe that $s_1^\theta(C) = k+1$. Then, since by Proposition \ref{prop:newGabsareMRD} $C$ is MRD, we conclude using part 3 of the characterization result in Theorem \ref{thm:characterizationGab}.
\item By Proposition \ref{prop:newGabsareMRD}, we know that $D$ is MRD. Moreover, $\theta(D)$ is generated by
\begin{itemize}
\item[-] $\theta^{i}(g)+\theta^{i}(\eta) \theta^{k+i}(g)$ for $1 \leq i < m-k$, which are already contained in $D$,
\item[-] $\theta^i(g)$ for $m-k < i \leq k$, which are already contained in $D$,
\item[-] $\theta^k(g)$, which is linearly independent from $D$, and
\item[-] $\theta^{m-k}(g)+\theta^{m-k}(\eta) \theta^{m}(g)=\theta^{m-k}(g)+\theta^{m-k}(\eta) g $, 
which is a linear combination of $\theta^{m-k}(g)$ (which is in $D$), $\theta^k(g)$ (which is in $\theta(D)$), and $g+\theta^{k}(\eta)\theta^k(g)$ (which is in $D$).
\end{itemize}
Hence, $\cS_1^{\theta}(D) = \langle D, \theta^k(g)\rangle$ has dimension $k+1$. We conclude using part 3 of Theorem~\ref{thm:characterizationGab}. \qedhere
\end{enumerate}
\end{proof}

\begin{remark}
At this point one may ask whether a characterization result similar to Theorem \ref{thm:characterizationGab} can be developed for twisted Gabidulin codes, in particular if it is possible to characterize these codes via their $\sigma$-sequences. Unfortunately, this seems to be impossible, even for the case $n=m$, in which we can show that there exists many codes with the same sequence of $\sigma$-sums: 

Let $g \in \F_{q^n}^n$ be such that $\rk_q(g)=n$, $\theta $ be a generator of $\Gal(\F_{q^n}/\Fq)$ and $\eta \in \F_{q^n}^*$ such that $\mathrm{N}_{\F_{q^n}/\Fq}(\eta) \neq (-1)^{kn}$, and consider the $\theta$-twisted Gabidulin code  $C:=\TGabShee{k, \theta}{\eta}{g}$. Moreover, we define
$$\G_{I,\theta}(g'):=\left\langle \theta_i(g') \mid i \in I \right\rangle,$$
for some $I \subseteq \{0,1,\ldots, n-1\}$, and $g'\in \F_{q^n}^n$ with $\rk_q(g')=n$. Observe that Gabidulin codes are a particular case of $\G_{I,\theta}(g')$, with $I=\{0,\ldots, k-1\}$. It was shown in \cite[Theorem 2.2]{csajbok2018maximum} that this family of codes corresponds to the family of codes having left and right idealizers isomorphic to $\F_{q^n}$. Moreover, in \cite{lunardon2018generalized}, it was proven that the right idealizer of the twisted Gabidulin code  $C$ is isomorphic to $\F_{q^{\gcd(k,n)}}$ whenever $1<k<n-1$, in the case $n=m$. Therefore, codes of the form $\G_{I,\theta}(g')$ cannot be equivalent to twisted Gabidulin codes, since the right idealizer is invariant under code equivalence. For more details, the interested reader is referred to \cite{csajbok2018maximum,lunardon2018generalized}. 

Now we consider a set $I=\{0,\ldots, k\}\setminus \{j\}$, for some $1\leq j \leq k-1$,  a vector $h\in \F_{q^n}^n$ with $\rk_q(h)=n$ and the corresponding code $\G_{I,\theta}(h)$. Using Proposition \ref{cor:rankMoore}, one can easily show that $s_i^\theta(\G_{I,\theta}(h))=\min \{e,n\}$, where 
$$e=\begin{cases} k & \mbox{ if } i=0 \\
k+1+i & \mbox{ if } i\geq 1. \end{cases}$$
In particular the sequence of the  $\theta$-sum dimensions of $C$ is equal to the one of $\G_{I,\theta}(h)$. It is easy to see that also the sequences of $\theta$-intersection dimensions coincide. This shows that a result in the spirit of Theorem \ref{thm:characterizationGab} seems not possible for twisted Gabidulin codes. Moreover, it was proven in \cite{csajbok2018maximum} that when $n=7$ and $q$ is odd, or when $n=8$ and  $q \equiv 1 \mod  3$, the codes $\G_{I,\theta}(h)$ with $I=\{0,1,3\}$ are MRD. Therefore, even with the further assumption to take an MRD code, the characterization of $\theta$-twisted Gabidulin codes in terms of their $\theta$-sequences seems not feasible. A step in this direction was done in \cite{giuzzi2019identifiers}, where a characterization of $\theta$-Gabidulin codes involving $t_0^\theta(C)$, $t_1^\theta(C)$, and $t_2^\theta(C)$ was given. However, the characterization given there required the existence of a certain element of maximum rank $n$ with some special properties. An interesting problem would be if one can find better conditions involving sum or intersection sequences, possibly for different $\theta$'s, in order to characterize twisted Gabidulin codes.
\end{remark}

\section{Conclusion}\label{sec:conc}

We showed that the sum and the intersection sequence of a rank-metric code (under an automorphism) are invariants of the equivalence class of the code. This gives an efficiently computable criterion to check if two codes are inequivalent. 

We determined many exact and some bounds on the values of these sequences for known maximum rank distance code families, namely Gabidulin, twisted Gabidulin and generalized twisted Gabidulin codes. 
Based on these results we derived sufficient conditions on the parameters $n$ (length), $k$ (dimension) and $m$ (field extension degree) of the codes, such that (generalized) twisted Gabidulin codes are inequivalent to Gabidulin codes, and such that certain generalized twisted Gabidulin codes are inequivalent to narrow sense twisted Gabidulin codes. 

Furthermore, we used the invariants to derive  upper and lower bounds on the number of inequivalent classical and (generalized) twisted Gabidulin codes, where for Gabidulin codes with $m=n$ the bounds coincide and give the exact number of inequivalent Gabidulin codes. Finally, we gave a characterization of Gabidulin codes in terms of these sequences. We used this latter result to show that Gabidulin's new codes (which correspond to generalized twisted Gabidulin codes with certain parameters) are in fact Gabidulin codes in the original sense. A similar characterization for twisted Gabidulin codes is not possible, since there is a counter example of an inequivalent code construction that has the same sum (or intersection) sequence as a twisted Gabidulin code. The question remains if it is possible to characterize twisted Gabidulin codes by its sum or intersection sequence together with another easily computable criterion.

Besides the main purpose of being an easily computable criterion to verify if a new code construction is inequivalent to other known code constructions, the results of this paper might also be of interest for code-based cryptography. Although not explicitly mentioned, the sum sequence of Gabidulin codes has frequently been used as a distinguisher from random linear codes, for attacking McEliece type of cryptosystems that use Gabidulin codes for the private key. On one hand, it is therefore a promising idea to use non-Gabidulin MRD codes for the design of new cryptosystems. On the other hand, the invariants of the code families treated in this work are all again quite different from the behavior of a random code, which raises the question if similar distinguisher attacks (as for Gabidulin codes) are possible for (generalized) twisted Gabidulin codes.

For the sake of simplicity we have formulated all main results in this paper for codes that are evaluated in a full rank vector $g\in \mathbb F_{q^m}^n$. However, most of the results can easily be carried over to the case where $g$ does not have full rank. Moreover, we have shown how the invariants can be used to show the inequivalence of certain generalized twisted Gabidulin codes to other code families. This can similarly be done for many more subfamilies of generalized twisted Gabidulin codes, and more generally for any other type of evaluation code based on linearized polynomials.

Lastly, we would like to state the open problem of generalizing Theorem \ref{thm:numberinequivalentTGab}, which gives the exact number of inequivalent twisted Gabidulin codes in the case $n=m$, to the case $n<m$. This requires an estimate on the cardinality of the set of generators $\theta$ of $\Gal(\Fm/\Fq)$ for which a given code is $\theta$-twisted Gabidulin, which would then imply the number of inequivalent codes by Theorem \ref{thm:numTGab}.

\bibliographystyle{abbrv}
\bibliography{biblio}

\appendix

\section{Examples}\label{app:examples}

Tables~\ref{tab:m=15_n=8_k=3_example} and \ref{tab:m=23_n=20_k=9_example} exemplify, for two different sets of code parameters, Proposition~\ref{prop:GabidulinSeq} (sequences for Gabidulin codes) and Proposition~\ref{prop:TwGabidulinSeq} (sequences for Twisted Gabidulin codes) from Section~\ref{sec:sequences}, and also show the actual sequences $s_i$ for comparison. The sequences were computed using SageMath v8.1~\cite{sagemath}.

{\footnotesize
	\newcommand{\cellcolored}{\cellcolor{black!10}}
	\begin{longtable}{|l|l|l|l|}
		\caption{Example illustrating Proposition~\ref{prop:GabidulinSeq} (sequences for Gabidulin codes) and Proposition~\ref{prop:TwGabidulinSeq} (sequences for Twisted Gabidulin codes). Code parameters: $q=2$, $m=15$, $n=8$, $k=3$. Evaluation points $g = (\alpha^{16474}$, $\alpha^{23822}$, $\alpha^{10386}$, $\alpha^{28105}$, $\alpha^{21661}$, $\alpha^{2599}$, $\alpha^{30721}$, $\alpha^{198}) \in \Fqm^{n}$ and twist coefficient $\eta = \alpha^{22859} \in \Fqm$, where $\alpha$ is a primitive element of $\Fqm$ with $\alpha^{15} = \alpha^{5} + \alpha^{4} + \alpha^{2} + 1$.} \label{tab:m=15_n=8_k=3_example} \\
		\hline
		$r$ & Code & Case in Prop.~\ref{prop:GabidulinSeq} or Prop.~\ref{prop:TwGabidulinSeq} & $(s^{\theta^r}_1(C), s^{\theta^r}_2(C), \dots )$ \\
		\hline \hline
		$   1$ & Gabidulin & \cellcolored Actual Sequence & \cellcolored $( 4, 5, 6, 7, 8, 8, \dots )$ \\
		\cline{3-4}
		& & 1a: $r \leq k$ & $( 4, 5, 6, 7, 8, 8, \dots )$ \\
		& & 3b: $r < \min\{m-k,m-n+k\}$ & $s_1 \geq   -3$ \\
		\cline{2-4}
		& Tw.\ Gab.\ & \cellcolored Actual Sequence & \cellcolored $( 5, 6, 7, 8, 8, \dots )$ \\
		\cline{3-4}
		& & 1a: $r < k$ & $( 5, 6, 7, 8, 8, \dots )$ \\
		& & 3b: $r < \min\{m-k+1,m-n+k\}$ & $s_1 \geq   -3$ \\
		\hline \hline 
		$   2$ & Gabidulin & \cellcolored Actual Sequence & \cellcolored $( 5, 7, 8, 8, \dots )$ \\
		\cline{3-4}
		& & 1a: $r \leq k$ & $( 5, 7, 8, 8, \dots )$ \\
		& & 3b: $r < \min\{m-k,m-n+k\}$ & $s_1 \geq   -2$ \\
		\cline{2-4}
		& Tw.\ Gab.\ & \cellcolored Actual Sequence & \cellcolored $( 6, 8, 8, \dots )$ \\
		\cline{3-4}
		& & 1a: $r < k$ & $( 6, 8, 8, \dots )$ \\
		& & 3b: $r < \min\{m-k+1,m-n+k\}$ & $s_1 \geq   -2$ \\
		\hline \hline 
		$   3$ & Gabidulin & \cellcolored Actual Sequence & \cellcolored $( 6, 8, 8, \dots )$ \\
		\cline{3-4}
		& & 1a: $r \leq k$ & $( 6, 8, 8, \dots )$ \\
		& & 3b: $r < \min\{m-k,m-n+k\}$ & $s_1 \geq   -1$ \\
		\cline{2-4}
		& Tw.\ Gab.\ & \cellcolored Actual Sequence & \cellcolored $( 6, 8, 8, \dots )$ \\
		\cline{3-4}
		& & 2a: $k \leq r \leq n-k$ & $s_1 =     6$ \\
		& & 3b: $r < \min\{m-k+1,m-n+k\}$ & $s_1 \geq   -1$ \\
		\hline \hline 
		$   4$ & Gabidulin & \cellcolored Actual Sequence & \cellcolored $( 6, 8, 8, \dots )$ \\
		\cline{3-4}
		& & 2a: $k < r \leq n-k$ & $s_1 =     6$ \\
		& & 3b: $r < \min\{m-k,m-n+k\}$ & $s_1 \geq    0$ \\
		\cline{2-4}
		& Tw.\ Gab.\ & \cellcolored Actual Sequence & \cellcolored $( 6, 8, 8, \dots )$ \\
		\cline{3-4}
		& & 2a: $k \leq r \leq n-k$ & $s_1 =     6$ \\
		& & 3b: $r < \min\{m-k+1,m-n+k\}$ & $s_1 \geq    0$ \\
		\hline \hline 
		$   5$ & Gabidulin & \cellcolored Actual Sequence & \cellcolored $( 6, 8, 8, \dots )$ \\
		\cline{3-4}
		& & 2a: $k < r \leq n-k$ & $s_1 =     6$ \\
		& & 3b: $r < \min\{m-k,m-n+k\}$ & $s_1 \geq    1$ \\
		\cline{2-4}
		& Tw.\ Gab.\ & \cellcolored Actual Sequence & \cellcolored $( 6, 8, 8, \dots )$ \\
		\cline{3-4}
		& & 2a: $k \leq r \leq n-k$ & $s_1 =     6$ \\
		& & 3b: $r < \min\{m-k+1,m-n+k\}$ & $s_1 \geq    1$ \\
		\hline \hline 
		$   6$ & Gabidulin & \cellcolored Actual Sequence & \cellcolored $( 6, 8, 8, \dots )$ \\
		\cline{3-4}
		& & 3a: $\max\{k,n-k\} < r$ & $s_1 \geq    5$ \\
		& & 3b: $r < \min\{m-k,m-n+k\}$ & $s_1 \geq    2$ \\
		\cline{2-4}
		& Tw.\ Gab.\ & \cellcolored Actual Sequence & \cellcolored $( 6, 8, 8, \dots )$ \\
		\cline{3-4}
		& & 3a: $\max\{k-1,n-k\} < r$ & $s_1 \geq    5$ \\
		& & 3b: $r < \min\{m-k+1,m-n+k\}$ & $s_1 \geq    2$ \\
		\hline \hline 
		$   7$ & Gabidulin & \cellcolored Actual Sequence & \cellcolored $( 6, 7, 8, 8, \dots )$ \\
		\cline{3-4}
		& & 3a: $\max\{k,n-k\} < r$ & $s_1 \geq    4$ \\
		& & 3b: $r < \min\{m-k,m-n+k\}$ & $s_1 \geq    3$ \\
		\cline{2-4}
		& Tw.\ Gab.\ & \cellcolored Actual Sequence & \cellcolored $( 6, 8, 8, \dots )$ \\
		\cline{3-4}
		& & 3a: $\max\{k-1,n-k\} < r$ & $s_1 \geq    4$ \\
		& & 3b: $r < \min\{m-k+1,m-n+k\}$ & $s_1 \geq    3$ \\
		\hline \hline 
		$   8$ & Gabidulin & \cellcolored Actual Sequence & \cellcolored $( 6, 7, 8, 8, \dots )$ \\
		\cline{3-4}
		& & 3a: $\max\{k,n-k\} < r$ & $s_1 \geq    3$ \\
		& & 3b: $r < \min\{m-k,m-n+k\}$ & $s_1 \geq    4$ \\
		\cline{2-4}
		& Tw.\ Gab.\ & \cellcolored Actual Sequence & \cellcolored $( 6, 8, 8, \dots )$ \\
		\cline{3-4}
		& & 3a: $\max\{k-1,n-k\} < r$ & $s_1 \geq    3$ \\
		& & 3b: $r < \min\{m-k+1,m-n+k\}$ & $s_1 \geq    4$ \\
		\hline \hline 
		$   9$ & Gabidulin & \cellcolored Actual Sequence & \cellcolored $( 6, 8, 8, \dots )$ \\
		\cline{3-4}
		& & 3a: $\max\{k,n-k\} < r$ & $s_1 \geq    2$ \\
		& & 3b: $r < \min\{m-k,m-n+k\}$ & $s_1 \geq    5$ \\
		\cline{2-4}
		& Tw.\ Gab.\ & \cellcolored Actual Sequence & \cellcolored $( 6, 8, 8, \dots )$ \\
		\cline{3-4}
		& & 3a: $\max\{k-1,n-k\} < r$ & $s_1 \geq    2$ \\
		& & 3b: $r < \min\{m-k+1,m-n+k\}$ & $s_1 \geq    5$ \\
		\hline \hline 
		$  10$ & Gabidulin & \cellcolored Actual Sequence & \cellcolored $( 6, 8, 8, \dots )$ \\
		\cline{3-4}
		& & 3a: $\max\{k,n-k\} < r$ & $s_1 \geq    1$ \\
		& & 2b: $m-n+k \leq r < m-k$ & $s_1 =     6$ \\
		\cline{2-4}
		& Tw.\ Gab.\ & \cellcolored Actual Sequence & \cellcolored $( 6, 8, 8, \dots )$ \\
		\cline{3-4}
		& & 3a: $\max\{k-1,n-k\} < r$ & $s_1 \geq    1$ \\
		& & 2b: $m-n+k \leq r \leq m-k$ & $s_1 =     6$ \\
		\hline \hline 
		$  11$ & Gabidulin & \cellcolored Actual Sequence & \cellcolored $( 6, 8, 8, \dots )$ \\
		\cline{3-4}
		& & 3a: $\max\{k,n-k\} < r$ & $s_1 \geq    0$ \\
		& & 2b: $m-n+k \leq r < m-k$ & $s_1 =     6$ \\
		\cline{2-4}
		& Tw.\ Gab.\ & \cellcolored Actual Sequence & \cellcolored $( 6, 8, 8, \dots )$ \\
		\cline{3-4}
		& & 3a: $\max\{k-1,n-k\} < r$ & $s_1 \geq    0$ \\
		& & 2b: $m-n+k \leq r \leq m-k$ & $s_1 =     6$ \\
		\hline \hline 
		$  12$ & Gabidulin & \cellcolored Actual Sequence & \cellcolored $( 6, 8, 8, \dots )$ \\
		\cline{3-4}
		& & 3a: $\max\{k,n-k\} < r$ & $s_1 \geq   -1$ \\
		& & 1b: $m-k \leq r$ & $( 6, 8, 8, \dots )$ \\
		\cline{2-4}
		& Tw.\ Gab.\ & \cellcolored Actual Sequence & \cellcolored $( 6, 8, 8, \dots )$ \\
		\cline{3-4}
		& & 3a: $\max\{k-1,n-k\} < r$ & $s_1 \geq   -1$ \\
		& & 2b: $m-n+k \leq r \leq m-k$ & $s_1 =     6$ \\
		\hline \hline 
		$  13$ & Gabidulin & \cellcolored Actual Sequence & \cellcolored $( 5, 7, 8, 8, \dots )$ \\
		\cline{3-4}
		& & 3a: $\max\{k,n-k\} < r$ & $s_1 \geq   -2$ \\
		& & 1b: $m-k \leq r$ & $( 5, 7, 8, 8, \dots )$ \\
		\cline{2-4}
		& Tw.\ Gab.\ & \cellcolored Actual Sequence & \cellcolored $( 6, 8, 8, \dots )$ \\
		\cline{3-4}
		& & 3a: $\max\{k-1,n-k\} < r$ & $s_1 \geq   -2$ \\
		& & 1b: $m-k < r$ & $( 6, 8, 8, \dots )$ \\
		\hline \hline 
		$  14$ & Gabidulin & \cellcolored Actual Sequence & \cellcolored $( 4, 5, 6, 7, 8, 8, \dots )$ \\
		\cline{3-4}
		& & 3a: $\max\{k,n-k\} < r$ & $s_1 \geq   -3$ \\
		& & 1b: $m-k \leq r$ & $( 4, 5, 6, 7, 8, 8, \dots )$ \\
		\cline{2-4}
		& Tw.\ Gab.\ & \cellcolored Actual Sequence & \cellcolored $( 5, 6, 7, 8, 8, \dots )$ \\
		\cline{3-4}
		& & 3a: $\max\{k-1,n-k\} < r$ & $s_1 \geq   -3$ \\
		& & 1b: $m-k < r$ & $( 5, 6, 7, 8, 8, \dots )$ \\
		\hline \hline 
	\end{longtable}
}

{\footnotesize
	\newcommand{\cellcolored}{\cellcolor{black!10}}
	\begin{longtable}{|l|l|l|l|}
		\caption{Example illustrating Proposition~\ref{prop:GabidulinSeq} (sequences for Gabidulin codes) and Proposition~\ref{prop:TwGabidulinSeq} (sequences for Twisted Gabidulin codes). Code parameters: $q=3$, $m=23$, $n=20$, $k=9$. Evaluation points $g = (\alpha^{18291492625}$, $\alpha^{30157479146}$, $\alpha^{61931009420}$, $\alpha^{46672256788}$, $\alpha^{48458087457}$, $\alpha^{45285722774}$, $\alpha^{75023150823}$, $\alpha^{7059856837}$, $\alpha^{6759919186}$, $\alpha^{27228306115}$, $\alpha^{63169590947}$, $\alpha^{60982249453}$, $\alpha^{53931149991}$, $\alpha^{65993950263}$, $\alpha^{30419168464}$, $\alpha^{58409498579}$, $\alpha^{46827933410}$, $\alpha^{67114805914}$, $\alpha^{51682126798}$, $\alpha^{31714555456}) \in \Fqm^{n}$ and twist coefficient $\eta = \alpha^{67060309696} \in \Fqm$, where $\alpha$ is a primitive element of $\Fqm$ with $\alpha^{23} = 2\alpha^{3} + 2 \alpha + 2$.} \label{tab:m=23_n=20_k=9_example} \\
		\hline
		$r$ & Code & Case in Prop.~\ref{prop:GabidulinSeq} or Prop.~\ref{prop:TwGabidulinSeq} & $(s^{\theta^r}_1(C), s^{\theta^r}_2(C), \dots )$ \\
		\hline \hline
		$   1$ & Gabidulin & \cellcolored Actual Sequence & \cellcolored $( 10, 11, 12, 13, 14, 15, 16, 17, 18, 19, 20, 20, \dots )$ \\
		\cline{3-4}
		& & 1a: $r \leq k$ & $( 10, 11, 12, 13, 14, 15, 16, 17, 18, 19, 20, 20, \dots )$ \\
		& & 3b: $r < \min\{m-k,m-n+k\}$ & $s_1 \geq    7$ \\
		\cline{2-4}
		& Tw.\ Gab.\ & \cellcolored Actual Sequence & \cellcolored $( 11, 12, 13, 14, 15, 16, 17, 18, 19, 20, 20, \dots )$ \\
		\cline{3-4}
		& & 1a: $r < k$ & $( 11, 12, 13, 14, 15, 16, 17, 18, 19, 20, 20, \dots )$ \\
		& & 3b: $r < \min\{m-k+1,m-n+k\}$ & $s_1 \geq    7$ \\
		\hline \hline 
		$   2$ & Gabidulin & \cellcolored Actual Sequence & \cellcolored $( 11, 13, 15, 17, 19, 20, 20, \dots )$ \\
		\cline{3-4}
		& & 1a: $r \leq k$ & $( 11, 13, 15, 17, 19, 20, 20, \dots )$ \\
		& & 3b: $r < \min\{m-k,m-n+k\}$ & $s_1 \geq    8$ \\
		\cline{2-4}
		& Tw.\ Gab.\ & \cellcolored Actual Sequence & \cellcolored $( 12, 14, 16, 18, 20, 20, \dots )$ \\
		\cline{3-4}
		& & 1a: $r < k$ & $( 12, 14, 16, 18, 20, 20, \dots )$ \\
		& & 3b: $r < \min\{m-k+1,m-n+k\}$ & $s_1 \geq    8$ \\
		\hline \hline 
		$   3$ & Gabidulin & \cellcolored Actual Sequence & \cellcolored $( 12, 15, 18, 20, 20, \dots )$ \\
		\cline{3-4}
		& & 1a: $r \leq k$ & $( 12, 15, 18, 20, 20, \dots )$ \\
		& & 3b: $r < \min\{m-k,m-n+k\}$ & $s_1 \geq    9$ \\
		\cline{2-4}
		& Tw.\ Gab.\ & \cellcolored Actual Sequence & \cellcolored $( 13, 16, 19, 20, 20, \dots )$ \\
		\cline{3-4}
		& & 1a: $r < k$ & $( 13, 16, 19, 20, 20, \dots )$ \\
		& & 3b: $r < \min\{m-k+1,m-n+k\}$ & $s_1 \geq    9$ \\
		\hline \hline 
		$   4$ & Gabidulin & \cellcolored Actual Sequence & \cellcolored $( 13, 17, 20, 20, \dots )$ \\
		\cline{3-4}
		& & 1a: $r \leq k$ & $( 13, 17, 20, 20, \dots )$ \\
		& & 3b: $r < \min\{m-k,m-n+k\}$ & $s_1 \geq   10$ \\
		\cline{2-4}
		& Tw.\ Gab.\ & \cellcolored Actual Sequence & \cellcolored $( 14, 18, 20, 20, \dots )$ \\
		\cline{3-4}
		& & 1a: $r < k$ & $( 14, 18, 20, 20, \dots )$ \\
		& & 3b: $r < \min\{m-k+1,m-n+k\}$ & $s_1 \geq   10$ \\
		\hline \hline 
		$   5$ & Gabidulin & \cellcolored Actual Sequence & \cellcolored $( 14, 19, 20, 20, \dots )$ \\
		\cline{3-4}
		& & 1a: $r \leq k$ & $( 14, 19, 20, 20, \dots )$ \\
		& & 3b: $r < \min\{m-k,m-n+k\}$ & $s_1 \geq   11$ \\
		\cline{2-4}
		& Tw.\ Gab.\ & \cellcolored Actual Sequence & \cellcolored $( 15, 20, 20, \dots )$ \\
		\cline{3-4}
		& & 1a: $r < k$ & $( 15, 20, 20, \dots )$ \\
		& & 3b: $r < \min\{m-k+1,m-n+k\}$ & $s_1 \geq   11$ \\
		\hline \hline 
		$   6$ & Gabidulin & \cellcolored Actual Sequence & \cellcolored $( 15, 20, 20, \dots )$ \\
		\cline{3-4}
		& & 1a: $r \leq k$ & $( 15, 20, 20, \dots )$ \\
		& & 3b: $r < \min\{m-k,m-n+k\}$ & $s_1 \geq   12$ \\
		\cline{2-4}
		& Tw.\ Gab.\ & \cellcolored Actual Sequence & \cellcolored $( 16, 20, 20, \dots )$ \\
		\cline{3-4}
		& & 1a: $r < k$ & $( 16, 20, 20, \dots )$ \\
		& & 3b: $r < \min\{m-k+1,m-n+k\}$ & $s_1 \geq   12$ \\
		\hline \hline 
		$   7$ & Gabidulin & \cellcolored Actual Sequence & \cellcolored $( 16, 20, 20, \dots )$ \\
		\cline{3-4}
		& & 1a: $r \leq k$ & $( 16, 20, 20, \dots )$ \\
		& & 3b: $r < \min\{m-k,m-n+k\}$ & $s_1 \geq   13$ \\
		\cline{2-4}
		& Tw.\ Gab.\ & \cellcolored Actual Sequence & \cellcolored $( 17, 20, 20, \dots )$ \\
		\cline{3-4}
		& & 1a: $r < k$ & $( 17, 20, 20, \dots )$ \\
		& & 3b: $r < \min\{m-k+1,m-n+k\}$ & $s_1 \geq   13$ \\
		\hline \hline 
		$   8$ & Gabidulin & \cellcolored Actual Sequence & \cellcolored $( 17, 20, 20, \dots )$ \\
		\cline{3-4}
		& & 1a: $r \leq k$ & $( 17, 20, 20, \dots )$ \\
		& & 3b: $r < \min\{m-k,m-n+k\}$ & $s_1 \geq   14$ \\
		\cline{2-4}
		& Tw.\ Gab.\ & \cellcolored Actual Sequence & \cellcolored $( 18, 20, 20, \dots )$ \\
		\cline{3-4}
		& & 1a: $r < k$ & $( 18, 20, 20, \dots )$ \\
		& & 3b: $r < \min\{m-k+1,m-n+k\}$ & $s_1 \geq   14$ \\
		\hline \hline 
		$   9$ & Gabidulin & \cellcolored Actual Sequence & \cellcolored $( 18, 20, 20, \dots )$ \\
		\cline{3-4}
		& & 1a: $r \leq k$ & $( 18, 20, 20, \dots )$ \\
		& & 3b: $r < \min\{m-k,m-n+k\}$ & $s_1 \geq   15$ \\
		\cline{2-4}
		& Tw.\ Gab.\ & \cellcolored Actual Sequence & \cellcolored $( 18, 20, 20, \dots )$ \\
		\cline{3-4}
		& & 2a: $k \leq r \leq n-k$ & $s_1 =    18$ \\
		& & 3b: $r < \min\{m-k+1,m-n+k\}$ & $s_1 \geq   15$ \\
		\hline \hline 
		$  10$ & Gabidulin & \cellcolored Actual Sequence & \cellcolored $( 18, 20, 20, \dots )$ \\
		\cline{3-4}
		& & 2a: $k < r \leq n-k$ & $s_1 =    18$ \\
		& & 3b: $r < \min\{m-k,m-n+k\}$ & $s_1 \geq   16$ \\
		\cline{2-4}
		& Tw.\ Gab.\ & \cellcolored Actual Sequence & \cellcolored $( 18, 20, 20, \dots )$ \\
		\cline{3-4}
		& & 2a: $k \leq r \leq n-k$ & $s_1 =    18$ \\
		& & 3b: $r < \min\{m-k+1,m-n+k\}$ & $s_1 \geq   16$ \\
		\hline \hline 
		$  11$ & Gabidulin & \cellcolored Actual Sequence & \cellcolored $( 18, 19, 20, 20, \dots )$ \\
		\cline{3-4}
		& & 2a: $k < r \leq n-k$ & $s_1 =    18$ \\
		& & 3b: $r < \min\{m-k,m-n+k\}$ & $s_1 \geq   17$ \\
		\cline{2-4}
		& Tw.\ Gab.\ & \cellcolored Actual Sequence & \cellcolored $( 18, 20, 20, \dots )$ \\
		\cline{3-4}
		& & 2a: $k \leq r \leq n-k$ & $s_1 =    18$ \\
		& & 3b: $r < \min\{m-k+1,m-n+k\}$ & $s_1 \geq   17$ \\
		\hline \hline 
		$  12$ & Gabidulin & \cellcolored Actual Sequence & \cellcolored $( 18, 19, 20, 20, \dots )$ \\
		\cline{3-4}
		& & 3a: $\max\{k,n-k\} < r$ & $s_1 \geq   17$ \\
		& & 2b: $m-n+k \leq r < m-k$ & $s_1 =    18$ \\
		\cline{2-4}
		& Tw.\ Gab.\ & \cellcolored Actual Sequence & \cellcolored $( 18, 20, 20, \dots )$ \\
		\cline{3-4}
		& & 3a: $\max\{k-1,n-k\} < r$ & $s_1 \geq   17$ \\
		& & 2b: $m-n+k \leq r \leq m-k$ & $s_1 =    18$ \\
		\hline \hline 
		$  13$ & Gabidulin & \cellcolored Actual Sequence & \cellcolored $( 18, 20, 20, \dots )$ \\
		\cline{3-4}
		& & 3a: $\max\{k,n-k\} < r$ & $s_1 \geq   16$ \\
		& & 2b: $m-n+k \leq r < m-k$ & $s_1 =    18$ \\
		\cline{2-4}
		& Tw.\ Gab.\ & \cellcolored Actual Sequence & \cellcolored $( 18, 20, 20, \dots )$ \\
		\cline{3-4}
		& & 3a: $\max\{k-1,n-k\} < r$ & $s_1 \geq   16$ \\
		& & 2b: $m-n+k \leq r \leq m-k$ & $s_1 =    18$ \\
		\hline \hline 
		$  14$ & Gabidulin & \cellcolored Actual Sequence & \cellcolored $( 18, 20, 20, \dots )$ \\
		\cline{3-4}
		& & 3a: $\max\{k,n-k\} < r$ & $s_1 \geq   15$ \\
		& & 1b: $m-k \leq r$ & $( 18, 20, 20, \dots )$ \\
		\cline{2-4}
		& Tw.\ Gab.\ & \cellcolored Actual Sequence & \cellcolored $( 18, 20, 20, \dots )$ \\
		\cline{3-4}
		& & 3a: $\max\{k-1,n-k\} < r$ & $s_1 \geq   15$ \\
		& & 2b: $m-n+k \leq r \leq m-k$ & $s_1 =    18$ \\
		\hline \hline 
		$  15$ & Gabidulin & \cellcolored Actual Sequence & \cellcolored $( 17, 20, 20, \dots )$ \\
		\cline{3-4}
		& & 3a: $\max\{k,n-k\} < r$ & $s_1 \geq   14$ \\
		& & 1b: $m-k \leq r$ & $( 17, 20, 20, \dots )$ \\
		\cline{2-4}
		& Tw.\ Gab.\ & \cellcolored Actual Sequence & \cellcolored $( 18, 20, 20, \dots )$ \\
		\cline{3-4}
		& & 3a: $\max\{k-1,n-k\} < r$ & $s_1 \geq   14$ \\
		& & 1b: $m-k < r$ & $( 18, 20, 20, \dots )$ \\
		\hline \hline 
		$  16$ & Gabidulin & \cellcolored Actual Sequence & \cellcolored $( 16, 20, 20, \dots )$ \\
		\cline{3-4}
		& & 3a: $\max\{k,n-k\} < r$ & $s_1 \geq   13$ \\
		& & 1b: $m-k \leq r$ & $( 16, 20, 20, \dots )$ \\
		\cline{2-4}
		& Tw.\ Gab.\ & \cellcolored Actual Sequence & \cellcolored $( 17, 20, 20, \dots )$ \\
		\cline{3-4}
		& & 3a: $\max\{k-1,n-k\} < r$ & $s_1 \geq   13$ \\
		& & 1b: $m-k < r$ & $( 17, 20, 20, \dots )$ \\
		\hline \hline 
		$  17$ & Gabidulin & \cellcolored Actual Sequence & \cellcolored $( 15, 20, 20, \dots )$ \\
		\cline{3-4}
		& & 3a: $\max\{k,n-k\} < r$ & $s_1 \geq   12$ \\
		& & 1b: $m-k \leq r$ & $( 15, 20, 20, \dots )$ \\
		\cline{2-4}
		& Tw.\ Gab.\ & \cellcolored Actual Sequence & \cellcolored $( 16, 20, 20, \dots )$ \\
		\cline{3-4}
		& & 3a: $\max\{k-1,n-k\} < r$ & $s_1 \geq   12$ \\
		& & 1b: $m-k < r$ & $( 16, 20, 20, \dots )$ \\
		\hline \hline 
		$  18$ & Gabidulin & \cellcolored Actual Sequence & \cellcolored $( 14, 19, 20, 20, \dots )$ \\
		\cline{3-4}
		& & 3a: $\max\{k,n-k\} < r$ & $s_1 \geq   11$ \\
		& & 1b: $m-k \leq r$ & $( 14, 19, 20, 20, \dots )$ \\
		\cline{2-4}
		& Tw.\ Gab.\ & \cellcolored Actual Sequence & \cellcolored $( 15, 20, 20, \dots )$ \\
		\cline{3-4}
		& & 3a: $\max\{k-1,n-k\} < r$ & $s_1 \geq   11$ \\
		& & 1b: $m-k < r$ & $( 15, 20, 20, \dots )$ \\
		\hline \hline 
		$  19$ & Gabidulin & \cellcolored Actual Sequence & \cellcolored $( 13, 17, 20, 20, \dots )$ \\
		\cline{3-4}
		& & 3a: $\max\{k,n-k\} < r$ & $s_1 \geq   10$ \\
		& & 1b: $m-k \leq r$ & $( 13, 17, 20, 20, \dots )$ \\
		\cline{2-4}
		& Tw.\ Gab.\ & \cellcolored Actual Sequence & \cellcolored $( 14, 18, 20, 20, \dots )$ \\
		\cline{3-4}
		& & 3a: $\max\{k-1,n-k\} < r$ & $s_1 \geq   10$ \\
		& & 1b: $m-k < r$ & $( 14, 18, 20, 20, \dots )$ \\
		\hline \hline 
		$  20$ & Gabidulin & \cellcolored Actual Sequence & \cellcolored $( 12, 15, 18, 20, 20, \dots )$ \\
		\cline{3-4}
		& & 3a: $\max\{k,n-k\} < r$ & $s_1 \geq    9$ \\
		& & 1b: $m-k \leq r$ & $( 12, 15, 18, 20, 20, \dots )$ \\
		\cline{2-4}
		& Tw.\ Gab.\ & \cellcolored Actual Sequence & \cellcolored $( 13, 16, 19, 20, 20, \dots )$ \\
		\cline{3-4}
		& & 3a: $\max\{k-1,n-k\} < r$ & $s_1 \geq    9$ \\
		& & 1b: $m-k < r$ & $( 13, 16, 19, 20, 20, \dots )$ \\
		\hline \hline 
		$  21$ & Gabidulin & \cellcolored Actual Sequence & \cellcolored $( 11, 13, 15, 17, 19, 20, 20, \dots )$ \\
		\cline{3-4}
		& & 3a: $\max\{k,n-k\} < r$ & $s_1 \geq    8$ \\
		& & 1b: $m-k \leq r$ & $( 11, 13, 15, 17, 19, 20, 20, \dots )$ \\
		\cline{2-4}
		& Tw.\ Gab.\ & \cellcolored Actual Sequence & \cellcolored $( 12, 14, 16, 18, 20, 20, \dots )$ \\
		\cline{3-4}
		& & 3a: $\max\{k-1,n-k\} < r$ & $s_1 \geq    8$ \\
		& & 1b: $m-k < r$ & $( 12, 14, 16, 18, 20, 20, \dots )$ \\
		\hline \hline 
		$  22$ & Gabidulin & \cellcolored Actual Sequence & \cellcolored $( 10, 11, 12, 13, 14, 15, 16, 17, 18, 19, 20, 20, \dots )$ \\
		\cline{3-4}
		& & 3a: $\max\{k,n-k\} < r$ & $s_1 \geq    7$ \\
		& & 1b: $m-k \leq r$ & $( 10, 11, 12, 13, 14, 15, 16, 17, 18, 19, 20, 20, \dots )$ \\
		\cline{2-4}
		& Tw.\ Gab.\ & \cellcolored Actual Sequence & \cellcolored $( 11, 12, 13, 14, 15, 16, 17, 18, 19, 20, 20, \dots )$ \\
		\cline{3-4}
		& & 3a: $\max\{k-1,n-k\} < r$ & $s_1 \geq    7$ \\
		& & 1b: $m-k < r$ & $( 11, 12, 13, 14, 15, 16, 17, 18, 19, 20, 20, \dots )$ \\
		\hline \hline 
	\end{longtable}
}

\end{document}

%% file: Spec.tex
\theoremstyle{plain}
\newtheorem{theorem}{Theorem}[section]

\newtheorem{proposition}[theorem]{Proposition}
\newtheorem{corollary}[theorem]{Corollary}
\newtheorem{lemma}[theorem]{Lemma}

\theoremstyle{definition}
\newtheorem{definition}[theorem]{Definition}

\newtheorem{remark}[theorem]{Remark}

\newcommand{\C}{\mathcal{C}}

\newcommand{\G}{\mathcal G}

\newcommand{\cL}{\mathcal L}

\newcommand{\HH}{\mathcal H}

\newcommand{\II}{\mathcal I}
\newcommand{\T}{\mathcal T}
\newcommand{\cS}{\mathcal S}

\newcommand{\F}{\mathbb{F}}

\newcommand{\N}{\mathbb{N}}
\newcommand{\Fm}{\mathbb{F}_{q^m}}

\newcommand{\Fq}{\mathbb{F}_{q}}

\newcommand{\Z}{\mathbb{Z}}

\newcommand{\mat}{\Fq^{n \times m}}
\newcommand{\drk}{\mathrm{d}_{\mathrm{rk}}}

\newcommand{\Norm}{\mathrm{N}_{\Fm/\Fq}}

\newcommand{\Fmk}{[n, k]_{q^m}}
\newcommand{\Fmkd}{[n, k, d]_{q^m}}

\newcommand{\qbinom}{\genfrac{[}{]}{0pt}{}}

\newcommand{\autom}{\sigma}
\newcommand{\Saut}{\mathcal S^{\autom}}
\newcommand{\Taut}{\T^{\autom}}

\newcommand{\Fqm}{\mathbb{F}_{q^m}}
\newcommand{\Fqr}{\mathbb{F}_{q^r}}
\newcommand{\Fqn}{\mathbb{F}_{q^n}}

\newcommand{\numTwists}{\ell}
\newcommand{\tVec}{t}
\newcommand{\hVec}{h}
\newcommand{\etaVec}{\eta}

\newcommand{\HShee}[2]{\HH_{#1}^{#2}}
\newcommand{\TGabShee}[3]{\HH_{#1}^{#2}(#3)}
\newcommand{\PTw}[2]{\II_{#1}^{#2}}
\newcommand{\GTw}[3]{\II_{#1}^{#2}(#3)}

\newcommand{\NewG}[2]{\mathcal J_{#1}^{#2}}

\DeclareMathOperator{\End}{End}
\DeclareMathOperator{\Gal}{Gal}

\DeclareMathOperator{\rs}{rowsp}
\DeclareMathOperator{\Aut}{Aut}

\DeclareMathOperator{\GL}{GL}
\DeclareMathOperator{\rk}{rk}

\DeclareMathOperator{\Gab}{Gab}
\DeclareMathOperator{\TGab}{TGab}
\DeclareMathOperator{\gTGab}{gTGab}
\DeclareMathOperator{\Gr}{Gr}

\DeclareMathOperator{\rsu}{supp}

\usepackage{color}

\newcommand{\mytableformat}[3]{$\substack{#1 \\ #2 \\ #3}$}

\makeatletter
\newcommand{\mylabel}[2]{#2\def\@currentlabel{#2}\label{#1}}
\makeatother